\documentclass{CSML}
\pdfoutput=1

\usepackage{lastpage}
\newcommand{\dOi}{13(3:20)2017}
\lmcsheading{}{1--\pageref{LastPage}}{}{}%
{Mar.~29,~2017}{Sep.~06, 2017}{}


\ACMCCS{[{\bf Theory of computation}]: Logic---Verification by model checking}

\usepackage{hyperref}
\hypersetup{hidelinks}
\usepackage{etex}
\usepackage{mathtools,xspace}
\usepackage{amsmath}
\usepackage{amsfonts}
\usepackage{amssymb}
\usepackage{amsthm}
\usepackage{tabularx}
\usepackage{nicefrac}
\usepackage{proof}
\usepackage{cite}

\newtheoremstyle{thmC}%
  {6pt}
  {6pt}
  {\itshape}
  {}
  {\bfseries}
  {{\bfseries .}}
  {5pt plus 1pt minus 1pt}
  {\thmname{#1} \thmnumber{#2} \thmnote{\normalfont#3}}

\newtheoremstyle{defC}%
  {6pt}
  {6pt}
  {\normalfont}
  {}
  {\bfseries}
  {{\bfseries .}}
  {5pt plus 1pt minus 1pt}
  {\thmname{#1} \thmnumber{#2} \thmnote{\normalfont#3}}

\theoremstyle{thmC}
\newtheorem{thmC}[thm]{Theorem}
\theoremstyle{defC}
\newtheorem{defC}[thm]{Definition}

\usepackage[all]{xy}
\SelectTips{cm}{}  
\xyoption{v2}
\xyoption{curve}
\xyoption{2cell}
\UseAllTwocells
\SilentMatrices
\def\labelstyle{\scriptstyle}

\newdir{ >}{{}*!/-8pt/@{>}}

\usepackage{wrapfig}
\usepackage{floatflt}
\usepackage{tikz}

\usepackage{tikz}
\usepackage{tikz-qtree}
\usetikzlibrary{automata}
\usetikzlibrary{trees}
\usetikzlibrary{arrows,positioning,calc}
\usepackage{pgfplots}
\usepackage{pgfplotstable}

\usepackage{cancel}

\usepackage{stmaryrd}

\newcommand{\undef}{\mathrm{undefined}}

\newcommand{\pow}{\mathcal{P}}
\newcommand{\dist}{\mathcal{D}}
\newcommand{\lift}{\mathcal{L}}
\newcommand{\giry}{\mathcal{G}}
\newcommand{\place}{\underline{\phantom{n}}\,} 
\newcommand{\Sets}{\mathbf{Sets}}

\newcommand{\Meas}{\mathbf{Meas}}

\newcommand{\lang}{L}

\newcommand{\sigalg}{\mathfrak{F}}
\newcommand{\FSigma}{F_{\Sigma}}

\newcommand{\sol}{\text{sol}}

\newcommand{\Future}{\mathsf{F}}
\newcommand{\Globally}{\mathsf{G}}

\newcommand{\defarrow}{\overset{\text{def}}{\Leftrightarrow}}

\newcommand{\kto}{\mathrel{\ooalign{\hfil\raisebox{.3pt}{$\shortmid$}\hfil\crcr$\rightarrow$}}}

\newcommand{\longkto}{\mathrel{\ooalign{\hfil\raisebox{.3pt}{$\shortmid$}\hfil\crcr$\longrightarrow$}}}
\newcommand{\Kl}[0]{\mathcal{K}\mspace{-1mu}\ell}

\newcommand{\id}[0]{\mathrm{id}}
\newcommand{\tr}[0]{{\sf tr}}
\newcommand{\trinf}[0]{{\sf tr^{\infty}}}
\newcommand{\trB}[0]{{\sf tr^{\rm B}}}

\newcommand{\accstate}{\xymatrix{*+<1em>[o][F=]{}}}
\newcommand{\nonaccstate}{\xymatrix{*+<1em>[o][F-]{}}}

\newcommand{\Cppo}[0]{\mathbf{Cppo}}

\newcommand{\op}[0]{\mathrm{op}}

\newcommand{\myTree}[0]{\mathrm{Tree}}

\newcommand{\seq}[2]{{#1}_{1},\dotsc,{#1}_{#2}}
\newcommand{\empseq}[0]{\varepsilon}

\newcommand{\kar}{\ar|-*\dir{|}}

\newcommand{\C}{\mathbb{C}}

\newcommand{\codomRestr}[2]{{#1}\mbox{$\upharpoonright^{#2}$}}
\newcommand{\codomJoin}[1]{\langle{\kern-.5ex}\langle #1
\rangle{\kern-.5ex}\rangle}
\newcommand{\codomJoinL}{\langle{\kern-.4ex}\langle}
\newcommand{\codomJoinR}{\rangle{\kern-.4ex}\rangle}

\newcommand{\iso}{\mathrel{\stackrel{
           \raisebox{.5ex}{$\scriptstyle\cong\,$}}{
           \raisebox{0ex}[0ex][0ex]{$\rightarrow$}}}}
\newcommand{\oF}{\overline{F}}
\newcommand{\nat}{\mathbb{N}}
\newcommand{\Dom}{\mathrm{Dom}}
\newcommand{\Run}{\mathrm{Run}}

\newcommand{\AccRun}{\mathrm{AccRun}}

\newcommand{\NAccRuninf}{\mathrm{NAccRun}^{\mathrm{inf}}}
\newcommand{\Branch}{\mathrm{Branch}}

\newcommand{\Cyl}{\mathrm{Cyl}}

\newcommand{\X}{\mathcal{X}}
\newcommand{\NDL}{\mathrm{NoDiv}}

\newcommand{\DelSt}{\mathrm{DelSt}}

\newcommand{\approximant}{p}

\newcommand{\initSet}{I}
\newcommand{\initVec}{\iota}
\newcommand{\Acc}{\mathsf{Acc}}
\newcommand{\myalphabet}{\mathsf{A}}
\newcommand{\myparagraph}[1]{\vspace{1mm}\noindent\textbf{#1}.}

\newcommand{\initarrow}{\raisebox{\heightof{!}}{\,\rotatebox{180}{$!$}}}
\newcommand{\finalarrow}{\,!}

\newif\ifignore 
\ignorefalse

\definecolor{applegreen}{rgb}{0.55, 0.71, 0.0}
\definecolor{cadmiumgreen}{rgb}{0.0, 0.42, 0.24}

\newcounter{asmenumi}
\makeatletter
  \def\@thmcountersep{.}
\makeatother

\theoremstyle{plain}
\newtheorem{sublem}[thm]{Sublemma}

\theoremstyle{definition}

\dedicatory{In honor of Ji\v{r}\'{i} Ad\'{a}mek on the occasion of his seventieth birthday}

\begin{document}
\title[Fair Simulation for Nondeterministic and Probabilistic B\"{u}chi
Automata]{Fair Simulation for Nondeterministic and Probabilistic B\"{u}chi
Automata: a Coalgebraic Perspective}

\author[Natsuki Urabe]{Natsuki Urabe\rsuper{*}}	
\address{Dept.\ Computer Science, The University of Tokyo\\ Hongo 7-3-1, Tokyo 113-8656, 
  Japan}	
\email{urabenatsuki@is.s.u-tokyo.ac.jp}  
\thanks{\lsuper{*}JSPS Research Fellow}	


\author[Ichiro Hasuo]{Ichiro Hasuo}	
\address{
 National Institute of Informatics\\
 Hitotsubashi 2-1-2, Tokyo 101-8430, Japan
}	
\email{
i.hasuo@acm.org
}  

%

\begin{abstract}
Notions of \emph{simulation}, among other uses, provide a
 computationally tractable and sound (but not necessarily complete)
 proof method for language inclusion. They have been comprehensively
 studied by Lynch and Vaandrager for nondeterministic and timed systems;
 for \emph{B\"{u}chi} automata the notion of
 \emph{fair simulation} has been introduced by Henzinger, Kupferman and
 Rajamani. We contribute to a generalization of fair simulation in two
 different directions: one for nondeterministic \emph{tree} automata
 previously studied by Bomhard; and the other for
 \emph{probabilistic} word automata with finite state spaces, both
 under the B\"{u}chi acceptance condition. The former nondeterministic
 definition is formulated in terms of systems of fixed-point equations,
 hence is readily translated to parity games and is then amenable to
 Jurdzi\'{n}ski's algorithm; the latter probabilistic definition bears
 a strong ranking-function flavor. These two different-looking
 definitions are derived from one source, namely our \emph{coalgebraic}
 modeling of B\"{u}chi automata.
 Based on these coalgebraic observations,
 we also prove their soundness: a
 simulation indeed witnesses language inclusion.
\end{abstract}

\maketitle


\section{Introduction}\label{sec:intro}
Notions of \emph{simulation}---typically defined as a binary relation
 subject to a coinductive ``one-step mimicking''
condition---have been studied extensively in formal verification and process theory. 
 Sometimes existence of a simulation itself is interesting---taking
 it as the definition of an abstraction/refinement relationship---but
 another notable use  is as a \linebreak
 \emph{proof method} for
 language inclusion. Language inclusion is fundamental in model
 checking but often hard to check itself; looking for a
 simulation---which \emph{witnesses} language inclusion, by its
 \emph{soundness} property, in a step-wise
 manner---is then a sound (but generally not complete) alternative.
 For example, (finite) language inclusion 
 between weighted automata with weights in the semiring of real numbers 
 with ordinary addition and multiplication
 is undecidable~\cite{blondel03undecidableproblems}, while existence of
 certain simulations is PTIME, see~\cite{urabeH17MatrixSimulationJourn}.

 Simulation  notions have been introduced for many different types of
 systems: nondeterministic~\cite{lynch95forwardand},
 timed~\cite{lynch96forwardand} and probabilistic~\cite{JonssonL91},
 among others.  Conventionally many studies take the trivial acceptance
 condition (any run 
that does not diverge, i.e.\ that does not come to a deadend, is
 accepted). Recently, however, there have been several works on
 simulations under the \emph{B\"uchi} and \emph{parity} acceptance
 conditions~\cite{henzingerKR02fairsimulation,etessamiWS05fairsimulation,fritzW06simulationrelations}.
 In such settings a simulation notion is subject to an (inevitable)
 nonlocal \emph{fairness} condition (on top of the local condition of one-step mimicking); and
 often a fair simulation is characterized as a winning strategy of a
 suitable \emph{parity game}, which is then searched using 
Jurdzi\'{n}ski's algorithm~\cite{Jurdzinski00}.

Like simulation notions for the trivial acceptance condition, fair simulation can be used for proving language inclusion. 
Moreover, it also has a logical characterization: there exists a certain universal fragment of the alternation-free $\mu$-calculus
such that for two systems $\mathcal{X}$ and $\mathcal{Y}$, there exists a fair simulation from $\mathcal{X}$ to $\mathcal{Y}$ if and only if
all the formulas in the fragment that are satisfied in $\mathcal{Y}$ are also satisfied in 
$\mathcal{X}$~\cite{henzingerKR02fairsimulation}.
In contrast, differently from many other simulation notions, fair simulations cannot be used for state-space reduction~\cite{etessamiWS05fairsimulation}. In~\cite{etessamiWS05fairsimulation}, a weaker simulation notion called \emph{delayed simulation}
is introduced for this purpose. 

\subsection{Contributions}
It is in the context of fair simulation for word/tree automata
with nondeterministic/probabilistic branching that the current paper
contributes:
\begin{enumerate}
 \item\label{item:introContribNondetTree}
      We define fair simulation for \emph{nondeterministic tree}
       automata
       with the B\"uchi acceptance condition. We express the notion 
       using a \emph{system of fixed-point equations}---with explicit
       $\mu$'s and $\nu$'s indicating least or greatest---and thus the definition
       makes sense for infinite-state automata too. We also interpret it
       in terms of a parity game, which is subject to an algorithmic
       search when the problem instance is finitary. The resulting
       parity game 
       essentially coincides with the one in~\cite{Bomhard08BScThesis}. 
 \item\label{item:introContribProbWord}
       We define fair simulation for \emph{probabilistic word} automata
       with the B\"uchi acceptance condition, this time with the
       additional condition that on the simulating side we have a
       finite-state
       automaton. 
       This simulation notion is given by a \emph{matrix}
       (instead of a relation); this follows our previous
       work~\cite{urabeH17MatrixSimulationJourn} that uses \emph{linear
       programming} to search for such a matrix simulation. 
       Our current notion also requires suitable \emph{approximation
       sequences} for witnessing well-foundedness, with a similar
       intuition to \emph{ranking functions}. 
\end{enumerate}
For the former \emph{nondeterministic tree} setting ((\ref{item:introContribNondetTree}) in the above), a notion of fair
simulation---in addition to direct and delayed simulation---has already been
introduced in~\cite{Bomhard08BScThesis}. Their definition
focuses on finite state spaces and is
formulated   using a parity game. In contrast, 
our notion is given in terms of fixed-point equations---it generalizes the $\mu$-calculus characterization
of fair simulations from words (see e.g.~\cite{JuvekarP06}) to trees. An obvious advantage of our fixed-point characterization over the parity game-based one is that ours makes sense for infinite state systems too.
For the latter \emph{probabilistic word} setting ((\ref{item:introContribProbWord})  in the above), 
we introduce fair simulation for the first time
  (to the best of our knowledge).

In both settings our main technical result
is \emph{soundness}, that is, existence of a fair simulation
implies trace inclusion. We also exhibit nontrivial examples of fair
simulations.

\subsection{Theoretical Backgrounds} 
Our two simulation
notions (for nondeterministic tree automata and probabilistic word ones)
look rather different, but they are derived from the same
theoretical insights. The insights come from: 1) the theory of
\emph{coalgebra}~\cite{jacobs16CoalgBook,Rutten00a}, in particular
the generic \emph{Kleisli theory of trace and simulation}~\cite{jacobs04tracesemantics,cirstea10genericinfinite,hasuo06genericforward,urabeH17MatrixSimulationJourn,urabeH15CALCO};
and 2) our recent work~\cite{HasuoSC16} on a lattice-theoretic
foundation of \emph{nested/alternating fixed points}, where we generalize
\emph{progress measures}, a central notion in Jurdzi\'{n}ski's
algorithm~\cite{Jurdzinski00} for  parity games. We rely on both
of these series of work
also for soundness proofs, where we follow yet another recent work
of ours~\cite{urabeSH16parityTrace} in which we characterize the
accepted language of a B\"uchi automaton by an ``equational system'' of
diagrams in a Kleisli category. In this paper we shall briefly describe
these general
theories behind our current results, focusing on their instances that
are relevant.

\subsection{A Tribute to Ji\v{r}\'{i} Ad\'{a}mek}
Working in coalgebraic modeling of dynamical systems and their analysis,
we owe to Ji\v{r}\'{i} almost the field itself on which
we stand.

His earlier works like~\cite{Adamek74,adamekK79leastfixed} paved the way to 
one of the most fundamental ideas in the field, namely
the identification of final/terminal coalgebras as categorical fully abstract
domains of behaviors of non-terminating systems. 
Fixed-point equations---to which a final coalgebra is the greatest categorical 
solution in case an equation is given by a functor---have been a central theme
in 
his more recent works too. These include the line of work pursued in~\cite{AdamekMV11} and many others,  where ``solution operators'' for fixed-point equations and their axioms are studied in an elegantly categorical fashion. All these works of his have been a great source of inspiration for us. 


We wish to dedicate the current work to Ji\v{r}\'{i}. It continues our recent line of work~\cite{HasuoSC16,urabeSH16parityTrace,UrabeHH17toAppear} in which we pursue: categorical understanding of nested and alternating least and greatest fixed-point equations, and proof methods for such fixed-point specifications.  Categorical solutions to fixed-point equations play a central role here, much like in Ji\v{r}\'{i}'s series of work, while we believe our use of orders between arrows and explicit $\mu$'s and $\nu$'s (like in (\ref{eq:diagramsForEqSysForBuechiAcceptance})) is a crucial step ahead towards accommodating complex fixed-point specifications (like persistence and recurrence) in the categorical study of coalgebras. More specifically we contribute fair simulation notions as witnesses for B\"{u}chi language inclusion. Hopefully our results demonstrate potential practical values of mathematical and  categorical understanding of systems, a theme Ji\v{r}\'{i} has been pursuing throughout his career.

\subsection{Organization of the Paper} 
In Section~\ref{sec:prelim} we introduce \emph{equational systems},
essentially fixing notations for alternating greatest and least fixed
points. These notations---and the idea that fixed-point
equations play important roles---are used
in Sections~\ref{sec:fwdFairSimNondet}--\ref{sec:fwdFairSimProb} where we
concretely describe: our system models;
their accepted languages; simulation definitions;
and soundness results.
In this paper we consider \emph{nondeterministic tree} automata and \emph{probabilistic word} automata, both with the B\"uchi
acceptance condition, as system models.
 Up to this point everything is in
set-theoretic terms, without category theory.

The rest of the paper is devoted to soundness proofs and the
theoretical perspectives behind. In~Section~\ref{sec:coalgebraicBackgrounds}
we review the coalgebraic backgrounds: Kleisli categories, coalgebras, trace semantics~\cite{jacobs04tracesemantics,hasuo07generictrace,cirstea10genericinfinite}, simulations
(under the trivial acceptance condition)~\cite{hasuo06genericforward,urabeH15CALCO},
and coalgebraic trace for B\"uchi automata~\cite{urabeSH16parityTrace}. 
Finally, in Section~\ref{sec:soundnessProof} we take a coalgebraic look at 
simulations under the B\"uchi condition: our first attempt (fair
simulation \emph{with dividing}) is sound but not practically desirable; 
we show how we can circumvent this additional construct of dividing, and 
how we can obtain the concrete definitions
in Sections~\ref{sec:fwdFairSimNondet}--\ref{sec:fwdFairSimProb}.

In Section~\ref{sec:concluFW} we conclude and suggest some directions of future work.

\section{Preliminaries: Equational Systems}\label{sec:prelim}

\emph{Nested, alternating} greatest and least fixed
points---as in a $\mu$-calculus formula $\nu u_{2}. \mu u_{1}.\, (p\land
u_{2})\lor \Box u_{1}$---are omnipresent in specification and
verification. For their relevance to the B\"uchi acceptance condition
one can recall 
the well-known translation of LTL formulas to B\"{u}chi automata and
vice versa (see~\cite{Vardi95} for example). 
To express such fixed points we follow~\cite{CleavelandKS92,ArnoldNiwinski}
and use
\emph{equational systems}---instead of textual $\mu$-calculus-like
presentations. 
\begin{defi}[equational system]
\label{def:eqSys}
Let $L_{1},\dotsc,L_{m}$ be posets. 
We write $\sqsubseteq$ for the orders over the posets.
An \emph{equational system} $E$
over $L_{1},\dotsc,L_{m}$ is an expression 
\begin{equation}\label{eq:sysOfEq}
 \begin{array}{c}
  u_{1}=_{\eta_{1}}f_{1}(u_{1},\dotsc, u_{m})\enspace,\qquad
  \dotsc,\qquad  
  u_{m}=_{\eta_{m}} f_{m}(u_{1},\dotsc, u_{m})
 \end{array}
\end{equation}
where: $u_{1},\dotsc,u_{m}$ are \emph{variables},
 $\eta_{1},\dotsc,\eta_{m}\in\{\mu,\nu\}$, and 
$f_{i}\colon L_{1}\times\cdots\times L_{m}\to L_{i}$ is a monotone
 function.
 A variable $u_{j}$ is  a \emph{$\mu$-variable} if
  $\eta_{j}=\mu$; it is a \emph{$\nu$-variable} if $\eta_{j}=\nu$.
\end{defi}

%
\subsection{Solutions of Equational Systems}\label{sec:appendixEqSys}
In this section we define the \emph{solution} of an equational system.
For the equational system $E$ in Def.~\ref{def:eqSys}, its solution is 
defined as a family $(l_1^{\sol},\ldots,l_m^{\sol})\in L_1\times\cdots\times L_m$.
We first briefly sketch its definition.

 We assume that $L_{i}$'s have enough suprema and infima.
The definition 
 proceeds as follows: 1) we 
 solve the first equation of (\ref{eq:sysOfEq}) for $u_1$
 to obtain an interim solution 
 $u_{1}=l^{(1)}_{1}(u_{2},\dotsc,u_{m})$ that is parameterized by $u_2,\ldots,u_m$; 2) it is used in the second
 equation to eliminate $u_{1}$ and yield a new equation 
 $u_{2}=_{\eta_{2}}f^{\ddagger}_{2}(u_{2},\dotsc,u_{m})$; 3) solving it
 again gives an interim solution 
 $u_{2}=l^{(2)}_{2}(u_{3},\dotsc,u_{m})$; 4) continuing this way from
 left to right eventually eliminates all the variables and  leads to a closed solution
 $u_{m}=l^{(m)}_{m}\in L_{m}$; and 5) by propagating these closed
 solutions back from right to left, we obtain closed solutions for all of
 $u_{1},\dotsc,u_{m}$. 
 To summarize, when we are solving the $i$-th equation 
 $u_{i}=_{\eta_{i}}f_{i}(u_{1},\dotsc, u_{m})$, we first substitute 
 $u_1,\ldots,u_{i-1}$ with the current interim solutions $l_1^{(i-1)}(u_i,\ldots,u_n),\ldots,l_{i-1}^{(i-1)}(u_i,\ldots,u_n)$,
 and solve the equation for $u_i$, regarding $u_{i+1},\ldots,u_n$ as parameters.
 We give now a formal definition.
\begin{defi}[solution]\label{def:solOfEqSys}
Let $E$ be the equational system in Definition~\ref{def:eqSys}.
For each $i\in[1,m]$ and $j\in[1,i]$, we define monotone functions 
$f^{\ddagger}_i\colon L_{i}\times\cdots\times L_{m}\rightarrow L_{i}$ and
$l^{(i)}_{j}\colon L_{i+1}\times\cdots\times L_{m}\rightarrow L_{j}$	 
	by induction on $i$ as follows.
	\begin{itemize}
	\item
	When $i=1$,
\begin{align*}
 	 f^{\ddagger}_{1}(l_{1},\dotsc,l_{m})&:= f_{1}(l_{1},\dotsc,l_{m}),\qquad\qquad\text{and}\\
	 l^{(1)}_{1}(l_{2},\dotsc, l_{m})&:= 
	 \begin{cases}
	 l_{\mathrm{lfp}} & (\text{$\eta_1=\mu$ and 
 $f^{\ddagger}_{1}(\place,l_{2},\dotsc, l_{m})$ has the lfp $l_{\mathrm{lfp}}\in L_1$}) \\
 	 l_{\mathrm{gfp}} & (\text{$\eta_1=\nu$ and 
 $f^{\ddagger}_{1}(\place,l_{2},\dotsc, l_{m})$ has the gfp $l_{\mathrm{gfp}}\in L_1$}) \\
\undef & (\text{otherwise}).
 \end{cases}
\end{align*}
Note here that 
completeness of $L_1$ is not assumed and therefore
the monotone function
$f^{\ddagger}_{1}(\place,l_{2},\dotsc, l_{m}) \colon L_{1}\to L_{1}$
does not necessarily have the lfp or gfp.

\item
For the step case, the function $f^{\ddagger}_{i+1}$ is defined 
using the
$i$-th interim solutions $l^{(i)}_{1},\dotsc,l^{(i)}_{i}$ for the variables
	$u_{1},\dotsc, u_{i}$ obtained so far:
\[
  f^{\ddagger}_{i+1}(l_{i+1},\dotsc, l_{m}):=
\begin{cases}
    f_{i+1}\bigl(\,l^{(i)}_{1}(l_{i+1},\dotsc, l_{m}),\;\dotsc,\;
 l^{(i)}_{i}(l_{i+1},\dotsc, l_{m}), \;
l_{i+1},\dotsc, l_{m}
\,\bigr)   & \\
& \hspace{-7cm}\text{($l^{(i)}_{j}(l_{i+1},\dotsc, l_{m})$ is defined for each $j\in[1,i]$)} \\
\undef 
&\hspace{-7cm}\text{(otherwise)}\,.
\end{cases}
\]
For $j=i+1$, $l_j^{(i+1)}$ is defined by
	\[
l^{(i+1)}_{i+1}(l_{i+2},\dotsc, l_{m})
:= 
\begin{cases}
	 l_{\mathrm{lfp}} & 
	 \left(\begin{aligned} \mbox{}&\text{$\eta_{i+1}=\mu$, and}\\ 
 \mbox{}&\text{$f^{\ddagger}_{i+1}(\place,l_{i+2},\dotsc, l_{m})$ has the lfp $l_{\mathrm{lfp}}\in L_{i+1}$}
\end{aligned} \right) \\[5mm]
	 l_{\mathrm{gfp}} & 
	 \left(\begin{aligned} \mbox{}&\text{$\eta_{i+1}=\nu$, and}\\ 
 \mbox{}&\text{$f^{\ddagger}_{i+1}(\place,l_{i+2},\dotsc, l_{m})$ has the gfp $l_{\mathrm{gfp}}\in L_{i+1}$}
\end{aligned} \right) \\[5mm]
\undef & \ (\text{otherwise}).
\end{cases}
\]
For $j\in [1,i]$, $l_j^{(i+1)}$ is defined using $l^{(i+1)}_{i+1}(l_{i+2},\dotsc, l_{m})$ as 
follows.
%
\[
 l^{(i+1)}_{j}(l_{i+2},\dotsc, l_{m})
:=
\begin{cases}
 l^{(i)}_{j}\bigl(\,l^{(i+1)}_{i+1}(l_{i+2},\dotsc, l_{m}),\; l_{i+2},\dotsc,l_{m}\,\bigr) & \\
 &\hspace{-2cm}\text{($l^{(i+1)}_{i+1}(l_{i+2},\dotsc, l_{m})$ is defined)} \\
 \undef & \hspace{-2cm} \text{(otherwise)}
 \end{cases}
 \]
%
\end{itemize}

A family $(l_1^{\sol},\ldots,l_m^{\sol})\in L_1\times\cdots\times L_m$ is called
the \emph{solution} of $E$ if $l_j^{(m)}:1\to L_j$ is defined and $l_j^{\sol}=l_j^{(m)}(*)$ 
(here $*$ is the unique element in $1$) for each $j\in[1,m]$.
\end{defi}
Note that the order of equations \emph{matters}.
For $(u=_{\mu}v, v=_{\nu} u)$ the solution is $u=v=\top$ while 
for $(v=_{\nu} u, u=_{\mu}v)$ the solution is $u=v=\bot$.
%
It is easy to see that all the functions $f^{\ddagger}_{i}$ and
 $l^{(i)}_{j}$ involved here are monotone. 
By the definition above, a solution exists if 
 the function  
\begin{equation}\label{eq:1703211255}
 f^\ddagger_{i}(\place,l_{i+1},\ldots,l_m)\;:\quad L_i\longrightarrow L_i
\end{equation}  
in Definition~\ref{def:solOfEqSys}
%
 has both the least and the greatest fixed points
 for each $i\in [1,m]$ and $l_{i+1}\in L_{i+1},\ldots,l_m\in L_m$.
Their existence
 depends on how ``complete'' each $L_{i}$ is and how ``continuous'' each $f_i$ is. 
In the following proposition
 we present two sufficient conditions for existence of 
the least and the greatest fixed points.
 
\begin{prop}\label{prop:suffCondSol}
Let $E$ be the equational system in Definition~\ref{def:eqSys}.
If either of the following conditions is satisfied, then
$E$ has a (necessarily unique) solution.
\begin{enumerate}
  \renewcommand{\labelenumi}{(\alph{enumi})}
  \renewcommand{\theenumi}{\alph{enumi}}
  \item\label{item:prop:suffCondSol1}
   For each $i\in[1,m]$, 
   the poset $L_i$ is 
   a complete lattice.

   
  \item\label{item:prop:suffCondSol2}
  For each $i\in[1,m]$ we have the following.
  \begin{itemize}
   \item  $L_i$ has both the least and greatest elements.
   \item $L_{i}$ is both $\omega$-complete and $\omega^\op$-complete, that is, 
	 every increasing (or decreasing) $\omega$-chain has a supremum (or an infimum, respectively).
   \item For each 
$l_{i+1}\in L_{i+1},\ldots,l_m\in L_m$, the function
\begin{displaymath}
 f^\ddagger_{i}(\place,l_{i+1},\ldots,l_m)\;:\quad L_i\longrightarrow L_i
\end{displaymath}  
in Definition~\ref{def:solOfEqSys}
is both $\omega$-continuous and $\omega^\op$-continuous, that is, the aforementioned suprema and infima are preserved by the function.     \qed
  \end{itemize}
  \end{enumerate}
\end{prop}

If Condition~(\ref{item:prop:suffCondSol1}) above is satisfied then 
existence of the least and the greatest fixed points of the function in (\ref{eq:1703211255}) is ensured by
the Knaster--Tarski theorem. 
In contrast, if Condition~(\ref{item:prop:suffCondSol2}) is satisfied then existence of the least and the greatest fixed points
is ensured by the Kleene fixed-point theorem.

As we will see later,  Condition~(\ref{item:prop:suffCondSol1}) is suitable for the nondeterministic setting while
Condition~(\ref{item:prop:suffCondSol2}) is suitable for the probabilistic setting.



\subsection{Progress Measure}\label{sec:progMeas}
The notion of (lattice-theoretic) \emph{progress
measure}~\cite{HasuoSC16}, although not explicit, plays an important
role in the current paper. 
We first briefly review its idea.

Verification
of a fixed-point specification amounts mathematically to
\emph{underapproximating} the fixed point.\footnote{In some cases we might be interested in approximations with respect to 
\emph{distance} rather than order~\cite{breugelW05BehaviouralPseudometric}.
In such cases we can use the Banach fixed-point theorem instead of the Knaster-Tarski or Cousot-Cousot one.}
 This is usually done very
differently for gfp's and lfp's. For a gfp $\nu f$ one provides an
\emph{invariant} $l$---a post-fixed point $l\sqsubseteq f(l)$---and
then the Knaster-Tarski theorem yields $l\sqsubseteq \nu f$. However, for an lfp 
$\mu f$, the same argument (namely finding a pre-fixed point $f(l)\sqsubseteq l$) would give an \emph{overapproximation};
instead we should appeal to the Cousot-Cousot theorem~\cite{cousotC79} and
consider
the \emph{approximation sequence} $\bot\sqsubseteq
f(\bot)\sqsubseteq\cdots$. The sequence eventually converges to $\mu f$
(possibly after transfinite induction);\footnote{In case $f$ is
continuous the sequence converges after $\omega$ steps. This is the
Kleene fixed-point theorem.} hence for every ordinal
$\alpha$, the approximant $f^{\alpha}(\bot)$ is an
underapproximation of $\mu f$. This is the underlying principle of
proofs by \emph{ranking functions} of termination, for example.

 Progress measures in~\cite{HasuoSC16}, generalizing the
combinatorial notion of the same name in
Jurdzi\'{n}ski's algorithm for parity games~\cite{Jurdzinski00}, are
roughly combination of invariants and ranking functions. The latter two
 must be combined in an intricate manner so that they respect
the order of equations in~(\ref{eq:sysOfEq}) (that is, priorities in parity
games or $\mu$-calculus formulas); we do so with the help of
a suitable truncated order. 

Use of parity games is nowadays omnipresent,
and  the study of fair
simulations is not an exception~\cite{etessamiWS05fairsimulation}. 
Following those previous works, the basic idea behind our developments (below)
 is to generalize: parity games to equational systems
(Definition~\ref{def:eqSys}); and accordingly, Jurdzi\'{n}ski's
(combinatorial) progress measure to our lattice-theoretic one~\cite{HasuoSC16}. 

In the rest of this section we formally state 
the formal definition of progress measure, as well as its soundness and completeness
results (against the solution of an equational system).
To this end,
we first review the notion of \emph{prioritized ordinal}, which
embodies the idea of \emph{priority} in parity games.
See~\cite{HasuoSC16} for the relationship between the notion of prioritized ordinal and the notion of priority in parity games.
\begin{defi}[prioritized ordinal, $\leq_{i}$]\label{def:prioritizedOrdinal}
Let $E$ be the equational system in~(\ref{eq:sysOfEq}) of Definition~\ref{def:eqSys}.
Let us collect the indices of $\mu$-variables:
\begin{math}
 \{i_{1},\dotsc,i_{k}\}
 =
 \{i\in[1,m]\mid \eta_{i}=\mu \text{ in~(\ref{eq:sysOfEq})}\},
\end{math}
and assume that 
 $i_{1}<\cdots <i_{k}$. 
 A \emph{prioritized ordinal} for $E$ is a $k$-tuple
 $(\seq{\alpha}{k})$ of ordinals. 

 For each $i\in [1,m]$ we define a  preorder $\leq_{i}$ between
 prioritized ordinals---called the \linebreak
  \emph{$i$-th truncated
 pointwise order}---as follows.
 If $i_k< i$, then 
 \begin{math}
(\alpha_{1},\dotsc,\alpha_{k})
  \leq_{i}
  (\alpha'_{1},\dotsc,\alpha'_{k})
 \end{math}
is always true. 
Otherwise, let $a\in [1,k]$ be such that
	       \begin{math}
		i_{1}<\cdots <i_{a-1}<i\le i_{a}<\cdots < i_{k},
	       \end{math}
 that is, $u_{i_{a}}$ is the $\mu$-variable with the smallest priority 
 above that of $i$.
 Then we define
 \begin{math}
(\alpha_{1},\dotsc,\alpha_{k})
  \leq_{i}
  (\alpha'_{1},\dotsc,\alpha'_{k})
 \end{math}
 if
 we have
 $\alpha_i\leq\alpha'_i$ for each $i\in[a,k]$.
\end{defi}

\begin{defi}[progress measure for an equational system]
\label{def:progMeas}
Let $E$ be the equational system 
in Definition~\ref{def:eqSys}.
We further assume that for each $i\in[1,m]$, $L_i$ has the smallest element $\bot$. 
 A \emph{progress measure} $p$ for $E$ is given by a tuple
 \begin{math}
  p
  =
  \bigl(\,
  (
  \overline{\alpha_{1}},\dotsc,
  \overline{\alpha_{k}}),
  \,
  \bigl(\,\approximant_{i}(\alpha_{1},\dotsc,\alpha_{k})\,\bigr)_{i,\seq{\alpha}{k}}
\,\bigr)
 \end{math}
that
consists of:
 \begin{itemize}
  \item the \emph{maximum prioritized ordinal}
 $(\overline{\alpha_{1}},\dotsc, \overline{\alpha_{k}})$; and 
  \item the \emph{approximants} $\approximant_{i}(\alpha_{1},\dotsc,\alpha_{k})\in L_{i}$, defined for
	each
	$i\in[1,m]$ and each
	prioritized ordinal
 $(\alpha_{1},\dotsc,\alpha_{k})$
 such that
	$
	\alpha_{1}\le\overline{\alpha_{1}},\dotsc,
	\alpha_{k}\le\overline{\alpha_{k}}
	$. 
 \end{itemize}
 The approximants $\approximant_{i}(\alpha_{1},\dotsc,\alpha_{k})$
 are subject to:
	\begin{enumerate}
	 \item \label{item:progressMeasDefMonotonicity}
	      \textbf{(Monotonicity)}
	       For each $i\in[1,m]$,
	       \begin{math}
		(\alpha_{1},\dotsc,\alpha_{k})
		\leq_{i}
		(\alpha'_{1},\dotsc,\alpha'_{k})
	       \end{math} 
	       implies
	       \begin{math}
		\approximant_{i}(\alpha_{1},\dotsc,\alpha_{k})
		\sqsubseteq
		\approximant_{i}(\alpha'_{1},\dotsc,\alpha'_{k})
	       \end{math}.

	 \item\label{item:progressMeasDefMuVarBaseCase}
	      \textbf{($\mu$-variables, base case)}
	       Let $a\in [1,k]$. Then $\alpha_{a}=0$ implies
	       $\approximant_{i_{a}}(\alpha_{1},\dotsc, \alpha_{a},\dotsc,\alpha_{k})=\bot$.
	 \item\label{item:progressMeasDefMuVarStepCase}
	      \textbf{($\mu$-variables, step case)}
	      Let $a\in [1,k]$.
	       Then there exist ordinals
	       $\beta_{1},\dotsc,
	       \beta_{a-1}$ such that
 $\beta_{1}\le\overline{\alpha_{1}},\dotsc,
	       \beta_{a-1}\le \overline{\alpha_{a-1}}$ and
	       \begin{multline}\label{eq:progressMeasDefMuVarStepCase}
		\approximant_{i_{a}} (\alpha_{1},\dotsc,\alpha_{a-1},
		\alpha_{a}+1,\alpha_{a+1},\dotsc,\alpha_{k})
		\\
		\;\sqsubseteq\;
		f_{i_{a}}
		\left(\,
		\begin{array}{c}
		\approximant_{1} (\beta_{1},\dotsc,\beta_{a-1},
		 \alpha_{a},\alpha_{a+1},\dotsc,\alpha_{k}),
		 \\
		 \dotsc,
		  \\
		\approximant_{m} (\beta_{1},\dotsc,\beta_{a-1},
		\alpha_{a},\alpha_{a+1},\dotsc,\alpha_{k})
		\end{array}
		\,\right).
	       \end{multline}	       

	 \item\label{item:progressMeasDefMuVarLimitCase}
	      \textbf{($\mu$-variables, limit case)}
	      Let $a\in [1,k]$ and let $\alpha_{a}$ be a limit ordinal. 
	      Then the supremum 
	      $\bigsqcup_{\beta<\alpha_{a}}\approximant_{i_{a}} (\alpha_{1},\dotsc,\beta,\dotsc,\alpha_{k})\in L_{i_a}$
	      exists and we have:
	      \begin{equation}\label{eq:progressMeasDefMuVarLimitCase}
		\approximant_{i_{a}} (\alpha_{1},\dotsc,
		\alpha_{a},\dotsc,\alpha_{k})
		\sqsubseteq 
		\bigsqcup_{\beta<\alpha_{a}}
				\approximant_{i_{a}} (\alpha_{1},\dotsc,
		\beta,\dotsc,\alpha_{k})\enspace.
	       \end{equation}

	 \item\label{item:progressMeasDefNuVar}
	      \textbf{($\nu$-variables)}
	       Let $i\in [1,m]\setminus\{i_{1},\dotsc, i_{k}\}$;
	      and let $a\in [1,k]$ be such that
	       \begin{math}
		i_{1}<\cdots <i_{a-1}<i<i_{a}<\cdots < i_{k}
	       \end{math}.
	       Let $(\alpha_{1},\dotsc,\alpha_{k})$ be a prioritized ordinal.
	       Then
there exist
	       ordinals
	       $\beta_{1},\dotsc,\beta_{a-1}$ such that
 $\beta_{1}\le \overline{\alpha_{1}},\dotsc,
	       \beta_{a-1}\le \overline{\alpha_{a-1}}$ and
	       \begin{equation}\label{eq:progressMeasDefNuVar}
		\approximant_{i} (\alpha_{1},\dotsc,
		\alpha_{a-1},\alpha_{a},\dotsc,\alpha_{k})
        \;\sqsubseteq\;
		f_{i}
		\left(\,
		\begin{array}{c}
		\approximant_{1} (\beta_{1},\dotsc,\beta_{a-1},
		 \alpha_{a},\dotsc,\alpha_{k}),
		 \\
		 \dotsc,
		  \\
		\approximant_{m} (\beta_{1},\dotsc,\beta_{a-1},
		\alpha_{a},\dotsc,\alpha_{k})
		\end{array}
		\,\right).
	       \end{equation}	       
	\end{enumerate}
\end{defi}
The definition  combines the features of
ranking functions (Conditions~2--4) and those of invariants (Condition~5). Note also that in
each clause ordinals with smaller priorities can be modified to
arbitrary $\beta_{i}$. 

\begin{rem}\label{rem:diffFromHSC16_1}
The definition of a progress measure in Definition~\ref{def:progMeas} is 
slightly different from the one in~\cite{HasuoSC16}, in the following points.
\begin{enumerate}
\item \label{item:difference1}
Condition~(\ref{item:progressMeasDefMonotonicity})
is 
given using the truncated \emph{pointwise} order instead of the truncated \linebreak \emph{lexicographic} order.

\item \label{item:difference3}
It is not assumed that each $L_i$ is a complete lattice.
Instead,
in Condition~(\ref{item:progressMeasDefMuVarLimitCase}),
existence of the supremum is explicitly required.
\end{enumerate}

The difference (\ref{item:difference1})
is
 made for the sake of
cleanliness of the soundness proof for our notion of simulation (Theorem~\ref{thm:soundnessFwdFairBuechi}).
The difference (\ref{item:difference3}) is made because,
in the probabilistic setting (see e.g.\ Example~\ref{example:giryNotDCpo}), 
we should consider progress measures where each $L_i$ is not a complete lattice or even a dcpo.
Because of the 
latter difference,
in the correctness theorem below, we 
need
 extra assumptions
((\ref{item:thm:correctnessOfProgMeasEqSys1}) and (\ref{item:thm:correctnessOfProgMeasEqSys2}))
that do not appear in the correctness theorem in~\cite{HasuoSC16}.
\end{rem}

Despite these modifications,
the notion of progress measure in Definition~\ref{def:progMeas} shares correctness properties
with the original definition in~\cite{HasuoSC16}---soundness and completeness.
The proofs are almost the same as the ones in~\cite{HasuoSC16}.

\begin{thm}[{correctness of progress measures}]
\label{thm:correctnessOfProgMeasEqSys}
Let $E$ be the equational system~(\ref{eq:sysOfEq})
and assume
 that $E$ has the solution
$(
l^{\sol}_{1},
\dotsc,
l^{\sol}_{m}
)$.
  We further assume 
  that for each $i\in[1,m]$, 
  \begin{enumerate}
    \renewcommand{\labelenumi}{(\roman{enumi})}
  \renewcommand{\theenumi}{\roman{enumi}}
  \item\label{item:thm:correctnessOfProgMeasEqSys1} 
  the poset $L_{i}$
  has the least element $\bot$ and  
  is $\omega$-complete; and

  \item\label{item:thm:correctnessOfProgMeasEqSys2} 
  for each 
  $l_{i+1}\in L_{i+1},\ldots,l_m\in L_m$,
  the function 
\begin{equation*}
 f^\ddagger_{i}(\place,l_{i+1},\ldots,l_m)\;:\quad L_i\longrightarrow L_i
\end{equation*}  
in Definition~\ref{def:solOfEqSys}
  is $\omega$-continuous.
  \end{enumerate}
%
%
%
Then we have the following.
\begin{enumerate}
 \item\label{item:soundnessProgressMeas} \textbf{(Soundness)}
 For each progress measure 
  $p=\bigl((\overline{\alpha_{1}},\dotsc,\overline{\alpha_{k}}),\bigr(\approximant_{i}(\alpha_{1},\dotsc,\alpha_{k})\bigr)_{i,  
  \seq{\alpha}{k}}\bigr)$ 
we have
 \begin{math}
  \approximant_{i}(
  \overline{\alpha_{1}},\dotsc,
  \overline{\alpha_{k}})
  \sqsubseteq
  l^{\sol}_{i}
 \end{math} 
for each $i\in
        [1,m]$.
        
 \item\label{item:completenessProgressMeas} 
  \textbf{(Completeness)}
  There exists a progress measure 
  $p=\bigl((\overline{\alpha_{1}},\dotsc,\overline{\alpha_{k}}),\bigr(\approximant_{i}(\alpha_{1},\dotsc,\alpha_{k})\bigr)_{i,  
  \seq{\alpha}{k}}\bigr)$ 
  that achieves the solution, that is, 
 \begin{math}
  \approximant_{i}
  (
  \overline{\alpha_{1}},\dotsc,
  \overline{\alpha_{k}})
  =
  l^{\sol}_{i}
 \end{math}
 for each $i\in[1,m]$.
 Moreover 
 we can find $p$ such that $\overline{\alpha_i}\leq \omega$ for each $i\in[1,m]$.
       \qed
\end{enumerate}
 \end{thm}

%

\section{Fair Simulation for Nondeterministic B\"uchi Tree Automata}
\label{sec:fwdFairSimNondet}
A \emph{ranked alphabet} is a set $\Sigma$  with a function $|\place|:\Sigma\to\mathbb{N}$ that
gives an \emph{arity} to each $\sigma\in\Sigma$.

\begin{defi}[NBTA]\label{def:nondetBuechiTreeAutom}
A \emph{nondeterministic B\"uchi tree automaton}  (NBTA)
is given by a quintuple $\mathcal{X}=(X,\Sigma,\delta,\initSet,\Acc)$ consisting of 
a state space $X$, a ranked alphabet $\Sigma$, a \emph{transition function}
$\delta:X\to\pow(\coprod_{\sigma\in\Sigma} X^{|\sigma|})$,
a set $\initSet\subseteq X$ of the \emph{initial states}, and
a set $\Acc\subseteq X$ of the \emph{accepting states}
 (often designated by $\accstate$). 
\end{defi}


%
\begin{exa}\label{example:NBTA}
We define an NBTA $\mathcal{X}=(X,\Sigma,\delta,\initSet,\Acc)$ as follows.

 \noindent\begin{minipage}{0.5\hsize}
\begin{itemize}
\item $X=\{x_1,x_2\}$
\item $\Sigma=\{a,b\}$ where $|a|=|b|=2$
\item $\delta(x_1)=\delta(x_2)=\{(a,(x_1,x_1)),(b,(x_2,x_2))\}$
\item $I=\{x_1\}$
\item $\Acc=\{x_2\}$
\end{itemize}
\end{minipage}
\begin{minipage}{0.5\hsize}
\[
\vspace{-3mm}
\raisebox{5mm}{
\begin{xy}
(-10,22)*{\mathcal{X}}="",
(0,15)*+[Fo]{x_1} = "x1",
(10,10)*{\Box} = "xb1",
(0,8)*{\Box} = "xc1",
(20,15)*+[Foo]{x_2} = "x2",
(10,20)*{\Box} = "xb2",
(20,8)*{\Box} = "xc2",
%
\ar (-10,15);"x1"*+++[o]{}
\ar ^{b} "x1"*++{};"xb1"*+[o]{}
\ar @<.5mm> "xb1";"x2"*+++[o]{}
\ar @<-.5mm> "xb1";"x2"*+++[o]{}
\ar _(.7){a} "x2"*++{};"xb2"*+[o]{}
\ar @<.5mm> "xb2";"x1"*+++[o]{}
\ar @<-.5mm> "xb2";"x1"*+++[o]{}
\ar ^(.7){a} @/^2mm/"x1";"xc1"*+[o]{}
\ar @/^2mm/ "xc1";"x1"*++[o]{}
\ar @/^4mm/ "xc1";"x1"*++[o]{}
\ar _(.7){b}@/_2mm/ "x2";"xc2"*+[o]{}
\ar @/_2mm/ "xc2";"x2"*++[o]{}
\ar @/_4mm/ "xc2";"x2"*++[o]{}
\end{xy}
}
\]
\end{minipage}
Then $\mathcal{X}$ can be illustrated as in the above.
Here $x\xrightarrow{\sigma}\Box\rightrightarrows {\tiny\begin{matrix}y\\[-.5mm] z\end{matrix}}$ denotes $(\sigma,(y,z))\in\delta(x)$.
 \end{exa}


\subsection{Accepted Languages of Nondeterministic B\"uchi Tree Automata}\label{subsec:basicNBTA}
%
We start with reviewing
necessary notions for defining accepted (tree)
languages of NBTAs. 
They are all as usual.

\begin{nota}\label{nota:generativeProbBuechiTreeAutom}
We let $\nat^{*}$ and $\nat^{\omega}$ denote the sets of finite and
infinite sequences over natural numbers, respectively. Moreover we let
$\nat^{\infty}:=\nat^{*}\cup\nat^{\omega}$.
Concatenation of finite/infinite sequences, and/or characters are denoted
simply by juxtaposition. Given an infinite sequence
$\pi=\pi_{1}\pi_{2}\dotsc\in\nat^{\omega}$ (here $\pi_{i}\in\nat$), its
prefix $\pi_{1}\dotsc\pi_{n}$ is denoted by $\pi_{\le n}$. 
\end{nota}

The following formalization of trees and related notions are standard,
with its variations used in~\cite{carayolHS14} for example.

\begin{defi}[$\Sigma$-tree]\label{def:treesAndRelatedNotions}
Let $\Sigma$ be a ranked alphabet, with each element $\sigma\in \Sigma$
coming with its arity $|\sigma|\in \nat$. 
 A \emph{$\Sigma$-tree} $\tau$ is given by a nonempty subset
       $\Dom(\tau)\subseteq\nat^{*}$ (called the \emph{domain} of $\tau$)
       and a \emph{labeling} function $\tau\colon \Dom(\tau)\to\Sigma$ that
       are subject to the following conditions.\footnote{We shall use the same
       notation $\tau$ for a tree itself and its labeling
       function. Confusion is unlikely.}
       \begin{enumerate}
	\item $\Dom(\tau)$ is \emph{prefix-closed}: for any $w\in \nat^{*}$  and
	      $i\in \nat$, $wi\in \Dom(\tau)$ implies $w\in \Dom(\tau)$. 
	      See Figure~\ref{fig:positionInATree}.
	\item $\Dom(\tau)$ is \emph{lower-closed}:  for any $w\in \nat^{*}$  and
	      $i,j\in \nat$, $wj\in \Dom(\tau)$ and $i\le j$ imply $wi\in
	      \Dom(\tau)$. 
	      See Figure~\ref{fig:positionInATree}.
	\item The labeling function is consistent with arities: for any $w\in
	      \Dom(\tau)$, let $\sigma=\tau(w)$. Then $w0,w1,\dotsc,
	      w(|\sigma|-1)$ belong to $\Dom(\tau)$, and $wi\not\in
	      \Dom(\tau)$ for each $i$ such that $|\sigma|\le i$. 
	      See Figure~\ref{fig:arityOfLabelsAndNumOfSucc}.
       \end{enumerate}

The set of all $\Sigma$-trees shall be denoted by $\myTree_{\Sigma}$. 
\end{defi}

\begin{figure}[tbp]
\begin{minipage}{.38\textwidth}\centering
       \begin{tikzpicture}[level distance=3em,sibling distance=1em]
 \Tree[.$\varepsilon$ [.$0$ [.$00$ [.{\scriptsize $\vdots$} ]
	[.{\scriptsize $\vdots$} ] ] ] [.$1$ [.$10$ ] [.$11$
	[.{\scriptsize $\vdots$} ] ] [.$12$ ] ] ]
       \end{tikzpicture}
\caption{Positions in a tree}
\label{fig:positionInATree}
\end{minipage}
\begin{minipage}{.6\textwidth}\centering
       \begin{tikzpicture}[level distance=3em,sibling distance=1em]
 \Tree[.$\sigma^{(2)}_{1}$ [.$\sigma^{(1)}_{1}$ [.$\sigma^{(2)}_{2}$ [.{\scriptsize $\vdots$} ]
	[.{\scriptsize $\vdots$} ] ] ] [.$\sigma^{(3)}_{1}$ [.$\sigma^{(0)}_{1}$ ] [.$\sigma^{(1)}_{2}$
	[.{\scriptsize $\vdots$} ] ] [.$\sigma^{(0)}_{2}$ ] ] ]
       \end{tikzpicture}
\caption{The arity  of a label, and  the number of successors. Here
 $\sigma^{(i)}_{j}\in \Sigma$ is assumed to be of arity $i$. }
\label{fig:arityOfLabelsAndNumOfSucc}
\end{minipage}
\end{figure}

Intuitively, a $\Sigma$-tree is a possibly infinite tree whose nodes are
labeled from  $\Sigma$ and each node, say labeled by
$\sigma$, has precisely $|\sigma|$ children.
A sequence $w\in \nat^{*}$ is understood as a \emph{position} in a tree.

The following definitions are standard, too, in the tree-automata literature.

\begin{defi}[run]\label{def:runOfProbTreeAutom}
A \emph{run} $\rho$ of
an NBTA (Definition~\ref{def:nondetBuechiTreeAutom})
 $\X=(X,\Sigma,\delta,\initSet,\Acc)$
is a (possibly infinite) tree whose nodes are $(\Sigma\times X)$-labeled.
That should be consistent with arities of symbols, and
compatible with the initial states ($\initSet\subseteq X$) and the transition $\delta$ of the automaton $\X$.
Precisely, it is given by the following conditions:
 \begin{enumerate}
  \item A nonempty subset $\Dom(\rho)\subseteq\nat^{*}$ that is subject to
    the same conditions (of being prefix-closed and lower-closed) as
    for $\Sigma$-trees (Definition~\ref{def:treesAndRelatedNotions}).
  \item A labeling function $\rho\colon \Dom(\rho)\to \Sigma\times X$
    such that, if $\rho(w)=(\sigma,x)$, then $w$ has precisely
    $|\sigma|$ successors $w0,w1,\dotsc,w(|\sigma|-1)\in
    \Dom(\rho)$.
  \item Successors are reachable by a transition, in the sense that
    $(\sigma_{w}, (x_{w0}, \dotsc, x_{w|\sigma|-1})) \in \delta(x_{w})$ holds, where
    $\rho(w)$ is labeled with $(\sigma_{w},x_{w})$, and
    $\rho(wi)$ is labeled with $(\sigma_{wi},x_{wi})$ for each $0 \leq i < |\sigma|$.
  \item The root is labeled with an initial state, that is,
    $x_{\varepsilon}\in I$ where
    $\rho(\varepsilon) = (\sigma_{\varepsilon}, x_{\varepsilon})$.
 \end{enumerate}
 The set of all runs of the NBTA $\X$ is denoted by $\Run^\pow_\X $.

 The map denoted by $\DelSt\colon \Run^{\pow}_{\X}\to \myTree_{\Sigma}$
 takes a run $\rho\in\Run^{\pow}_\X$, removes its $X$-labels
 applying the first projection to each label, and returns the
 resulting $\Sigma$-labeled tree. The resulting tree is easily seen to be a
 $\Sigma$-tree by Definition~\ref{def:treesAndRelatedNotions}. We say that a run $\rho$
 is \emph{over} the $\Sigma$-tree $\DelSt(\rho)$. 
\end{defi}

In summary, a (possibly infinite) $(\Sigma \times X)$-labeled tree $\rho$ is 
 a run of an NBTA 
  $\X=(X,\Sigma,\delta,\initSet,\Acc)$ if: the $X$-label of its root
 is  initial  $s\in \initSet$; and for each node with a
 label $(\sigma,x)$, it has $|\sigma|$ children and we have
 $\bigl(\sigma,(x_{1},\dotsc,x_{|\sigma|})\bigr)\in \delta(x)$ where 
 $x_{1},\dotsc,x_{|\sigma|}$ are the $X$-labels of its children.

We next define a notion of branch. 

\begin{defi}[branch]\label{def:branch}
Let $\tau$ be a $\Sigma$-tree. A \emph{branch} of $\tau$ is either: 
\begin{itemize}
 \item 
 an
 infinite sequence $\pi=\pi_{1}\pi_{2}\dotsc\in\nat^{\omega}$ (where $\pi_{i}\in\nat$) such that any finite prefix 
 $\pi_{\le n}=\pi_{1}\dotsc\pi_{n}$ of it belongs to $\Dom(\tau)$; or 
 \item 
a
 finite sequence $\pi=\pi_{1}\dotsc\pi_{n}\in \nat^{*}$ where
       $\pi_{i}\in\nat$ that belongs to $\Dom(\tau)$
 and such that $\pi0\not\in\Dom(\tau)$.\footnote{This means that $\pi$ is a
 leaf of $\tau$, and that $\tau(\pi)$ is a $0$-ary symbol.}
\end{itemize}
The set of all branches of a $\Sigma$-tree $\tau$ is denoted by
 $\Branch(\tau)$. 
The notion of branch is defined similarly for a run, with
 $\Branch(\rho)$ denoting the set of all branches of $\rho$. 
\end{defi}

We define a notion of accepting run.
A run $\rho$ of an NBTA $\X$ is said to be accepting
  if any infinite branch $\pi$ of the tree $\rho$ satisfies
  the B\"{u}chi acceptance condition, that is, it visits
   accepting states (in $\Acc$) infinitely often.
 The sets of runs and accepting runs of $\X$ are denoted by
  $\Run^\pow_\X$ and $\AccRun^\pow_\X$, respectively. 
  Formally, they are defined as follows.

\begin{defi}[accepting run]\label{def:acceptingRun}
 A run $\rho$ of an NBTA 
 $\X=(X,\Sigma,\delta,\initSet,\Acc)$ is said to be \emph{accepting} 
 if, any branch $\pi\in\Branch(\rho)$ of it is \emph{accepting} in the
 following sense. 
\begin{itemize}
 \item The branch $\pi$ is an infinite sequence
       $\pi=\pi_{1}\pi_{2}\dotsc\in\nat^{\omega}$, and 
       the labels along the branch
	$(\sigma_{\varepsilon},x_{\varepsilon})
	\,
	(\sigma_{\pi_{1}},x_{\pi_{1}})
	\,
	(\sigma_{\pi_{1}\pi_{2}},x_{\pi_{1}\pi_{2}})
	\cdots$
	(here 
	$(\sigma_{w},x_{w}):=\rho(w)$ for each $w\in\nat^{*}$)
       visit accepting states infinitely often, that is, there exists an infinite
       sequence $n_{1}<n_{2}<\cdots$ of natural numbers such that
       $x_{\pi_{1}\dotsc \pi_{n_{i}}}\in F$ for each $i\in \nat$; or
 \item the branch $\pi$ is a finite sequence
       $\pi=\pi_{1}\dotsc\pi_{n}\in\nat^{*}$.
\end{itemize}
The set of all accepting runs over $\X$ is denoted by $\AccRun^{\pow}_{\X}$.
\end{defi}

Using the notions defined so far, we can define accepted languages of NBTAs as follows.

\begin{defi}[accepted language $L(\X)$]\label{def:BuechiLangConventionally}
For an NBTA $\mathcal{X}$,
its \emph{(B\"uchi) language} $\lang(\mathcal{X})\subseteq \myTree_{\Sigma}$ 
is defined by $\lang(\mathcal{X})=\{\DelSt(\rho)\mid \rho\in\AccRun^\pow_\X\}$.
\end{defi}

\begin{exa}\label{example:NBTAAccLang}
For the NBTA $\mathcal{X}$ in Example~\ref{example:NBTA}, the B\"uchi language $\lang(\mathcal{X})$ collects
all the $\{a,b\}$-labeled infinite binary trees where $b$ appears infinitely many times on each branch.
\end{exa}

\subsection{Fair Simulation for Nondeterministic B\"uchi Tree Automata}
\label{subsec:fwdFairSimNondet}
In this section we introduce \emph{fair simulation} for NBTAs; this is our first contribution. 
For finite-state NBTAs,
 our fair simulation notion is essentially the same as the one in~\cite{Bomhard08BScThesis}.
However, unlike the notion in~\cite{Bomhard08BScThesis} that is defined combinatorially via a parity game, 
ours is expressed by means of equational systems (Section~\ref{sec:prelim}), hence is applicable to infinitary settings.
%


Here is a brief description of a parity game.
For formal definitions, see~\cite{thomas2002automata} for example.
A \emph{parity game} is a game played by two players called \emph{Even} and \emph{Odd} over a 
finite-state directed graph $G=(V,E)$.
Each node $v\in V$ is called a \emph{position}, and the set $V$ of positions is divided into two parts---the one where
Even chooses the next move and the one where Odd chooses the next move.
We assume that a game is equipped with a \emph{priority function} $p:V\to\{0,1,\ldots,n\}$ that assigns a 
natural number called a \emph{priority} to each state. 


Once an initial state $v_0$ and strategies (functions from finite sequences of positions to a position) 
for Even and Odd are fixed, 
a run $\rho=v_0v_1\ldots\in V^\omega$, an infinite sequence over $G$, is determined in a natural manner.
A run is \emph{winning} for Even (respectively Odd) if the \emph{maximum} priority that  appears infinitely often in $\rho$ 
 is even (respectively odd).
A parity game is said to be \emph{winning} for Even from a position $v_0$ if there exists a strategy for Even such that,
regardless of the strategy of Odd, the resulting run from $v_0$ is winning for Even.
A notion of winning for Odd is defined similarly.
It is known that parity games satisfy \emph{determinacy}~\cite{thomas2002automata}: 
for each parity game and each state in it, the game is winning from the state for exactly one of Even and Odd.

We hereby review the combinatorial definition of fair simulation in~\cite{Bomhard08BScThesis} via a parity game,
to show an intuition behind our definition.

\begin{defi}[parity game for NBTA fair simulation,~\cite{Bomhard08BScThesis}]\label{def:parityFairSimulation}
Let $\mathcal{X}=(X,\Sigma,\delta_\mathcal{X},\initSet_\mathcal{X},\Acc_\mathcal{X})$ and 
$\mathcal{Y}=(Y,\Sigma,\delta_\mathcal{Y},\initSet_\mathcal{Y},\Acc_\mathcal{Y})$ be 
 NBTAs such that $X$ and $Y$ are finite.
 Let $X_1=X\setminus \Acc_\mathcal{X}$,
 $X_2=\Acc_{\mathcal{X}}$, and similarly for
 $Y=Y_{1}\cup Y_{2}$.
We define a parity game $G_{\mathcal{X},\mathcal{Y}}$ 
as follows.
\[
\setlength{\tabcolsep}{8pt}
\def\arraystretch{1.1}
\begin{tabularx}{\textwidth}{@{}l|c|l|c@{}}
Position & \!\!Player\!\! & The set of possible moves & Priority \\\hline
$*$ & Odd & $\initSet_\X$ & $0$ \\
$x\in X$ & Even & $\{(x,y)\mid y\in \initSet_{\mathcal{Y}}\}$ & $0$ \\
$(x,y)\in X\times Y$ & Odd & 
$
{\scriptsize
\left\{
\begin{pmatrix}
(\sigma,x_1,\ldots,x_{|\sigma|}),\\
y
\end{pmatrix}
\,\middle|\,
\begin{aligned}
&(\sigma,x_1,\ldots,x_{|\sigma|}) \\
&\qquad\quad \in \delta_{\mathcal{X}}(x)
\end{aligned}
\right\}\!\!\!
}
$
& 
$\!\!{\scriptsize\begin{cases}
0 & ((x,y)\in X_1\times Y_1) \\
1 & ((x,y)\in X_2\times Y_1) \\
2 & ((x,y)\in X\times Y_2) 
 \end{cases}}$ \!\!\\
 $\begin{aligned}
&((\sigma,x_1,\ldots,x_{|\sigma|}),y) \\
&\quad\;\in(\textstyle{\coprod_{\sigma\in\Sigma}X^{|\sigma|}})\times Y
\end{aligned}$
& Even & 
$\left\{
\begin{aligned}
&((x_1,y_1),\ldots,\\
&\quad (x_{|\sigma|},y_{|\sigma|}))
\end{aligned}
\,\middle|\, 
\begin{aligned}
&(\sigma,y_1,\ldots,y_{|\sigma|}) \\
&\qquad\quad \in \delta_{\mathcal{Y}}(y)
\end{aligned}
\right\}$
& $0$ \\
$(p_1,\ldots,p_{n})\in (X\times Y)^*$ & Odd & $\{p_i\mid 1\leq i\leq
	 n\}$ & $0$
\end{tabularx}
\]
%
\end{defi}
%
%
Note that as the number of positions of the game is finite,
the problem to determine the winner of $G_{\mathcal{X},\mathcal{Y}}$
is decidable.

We now introduce our fair simulation notion by means of equational systems.
We will later show that for finite-state NBTAs, our simulation notion is essentially the same 
as the one in Definition~\ref{def:parityFairSimulation}.

\begin{defi}[fair simulation for NBTAs]\label{def:simForNondetBuechiRel}
Let $\mathcal{X}=(X,\Sigma,\delta_\mathcal{X},\initSet_\mathcal{X},\Acc_\mathcal{X})$ and 
$\mathcal{Y}=(Y,\Sigma,\delta_\mathcal{Y},\initSet_\mathcal{Y},\Acc_\mathcal{Y})$ be NBTAs.
We define 
 $X_1,X_2$ and $Y_1,Y_2$  as in Definition~\ref{def:parityFairSimulation}.
A \emph{fair simulation} from $\mathcal{X}$ to $\mathcal{Y}$ is a
 relation $R\subseteq X\times Y$ such that:
\begin{enumerate}
\item\label{item:def:simForNondetBuechi2Rel}
 For all $x\in \initSet_\mathcal{X}$, there exists $y\in \initSet_\mathcal{Y}$ such that $(x,y)\in R$.

\item 
Let $u_1^\sol,\dotsc,u_4^\sol$ be the solution of 
the following equational system (note $\mu$'s vs.\ $\nu$'s). 
\begin{equation}\label{eq:1601061008Rel}
\begin{array}{rll}
u_1 &=_{\nu}\, \Box_{\mathcal{X},1}(\Diamond_{\mathcal{Y},1}\left(\bigwedge_{\Sigma}(u_1\cup u_2\cup u_3\cup u_4)\right)) &\subseteq X_1\times Y_1  \\ 
u_2 &=_{\mu}\, \Box_{\mathcal{X},2}(\Diamond_{\mathcal{Y},1}\left(\bigwedge_{\Sigma}(u_1\cup u_2\cup u_3\cup u_4)\right)) &\subseteq X_2\times Y_1  \\ 
u_3 &=_{\nu}\, \Box_{\mathcal{X},1}(\Diamond_{\mathcal{Y},2}\left(\bigwedge_{\Sigma}(u_1\cup u_2\cup u_3\cup u_4)\right)) &\subseteq X_1\times Y_2  \\ 
u_4 &=_{\nu}\, \Box_{\mathcal{X},2}(\Diamond_{\mathcal{Y},2}\left(\bigwedge_{\Sigma}(u_1\cup u_2\cup u_3\cup u_4)\right)) &\subseteq X_2\times Y_2   
\end{array}
\end{equation}
Then 
$R$ is below the solution, that is, $R\subseteq u_{1}^\sol\cup \cdots\cup u_{4}^\sol$.
\end{enumerate}
Here the functions 
$\Box_{\mathcal{X},i}:\pow(\coprod_{\sigma\in\Sigma}X^{|\sigma|}\times Y)\to\pow(X_i\times Y)$,
$\Diamond_{\mathcal{Y},j}:\pow(\coprod_{\sigma\in\Sigma}X^{|\sigma|} \times \coprod_{\sigma\in\Sigma}Y^{|\sigma|})\to\pow(\coprod_{\sigma\in\Sigma}X^{|\sigma|}\times Y_j)$ 
and
$\bigwedge_{\Sigma}:\pow(X\times Y)\to\pow(\coprod_{\sigma\in\Sigma}X^{|\sigma|} \times \coprod_{\sigma\in\Sigma}Y^{|\sigma|})$
are defined as follows.
\[
\begin{array}{rl}
 \Box_{\mathcal{X},i}(S)&:=\{(x,y)\in X_i\times Y\mid \forall \mathbf{x}'\in \delta_\mathcal{X}(x).\, 
 (\mathbf{x}',y)\in S\} \\ 
 \Diamond_{\mathcal{Y},j}(T)&:=\left\{(\mathbf{x}',y)\in \coprod_{\sigma\in\Sigma}X^{|\sigma|}\times Y_j \,\middle|\, \exists \mathbf{y}'\in \delta_{\mathcal{Y}}(y).\, (\mathbf{x}',\mathbf{y}')\in T\right\}  \\
 \textstyle{\bigwedge}_\Sigma(U)&:=\left\{
 \begin{aligned}
 &\bigl((\sigma,x_1,\ldots,x_{|\sigma|}),(\sigma',y_1,\ldots,y_{|\sigma'|})\bigr) \\
 &\qquad\qquad\qquad\qquad\in\textstyle{\coprod_{\sigma\in\Sigma}X^{|\sigma|} \times \coprod_{\sigma\in\Sigma}Y^{|\sigma|}} 
 \end{aligned}\,\middle|\,
 \begin{aligned}
 & \sigma=\sigma', \\
 & \forall i.\, (x_i,y_i)\in  U
 \end{aligned}
 \right\}
\end{array}
\]
%
\end{defi}

\begin{thm}[soundness]\label{thm:soundnessOfSimulationEqSys}
In the setting of Definition~\ref{def:simForNondetBuechiRel}, 
existence of a fair simulation from $\X$ to $\mathcal{Y}$ implies
language inclusion, that is, $\lang(\mathcal{X})\subseteq\lang(\mathcal{Y})$.
\end{thm}

Our proof of Theorem~\ref{thm:soundnessOfSimulationEqSys} relies on
a categorical theory developed in later sections,
and will be given in Section~\ref{subsec:nondetSetting}.

\begin{exa}\label{example:fwdFairSimNondetTree}
Let $\mathcal{X}$ and $\mathcal{Y}$ be the NBTAs illustrated below,
where  a transition $z\xrightarrow{\sigma}(z_1,z_2)$ is represented by 
$z\!\xrightarrow{\sigma}\!{\raisebox{-.8pt}{$\scriptstyle\Box$}} \!
\rightrightarrows
\!
{\scriptsize\begin{matrix}z_1 \\ z_2\end{matrix}}$.
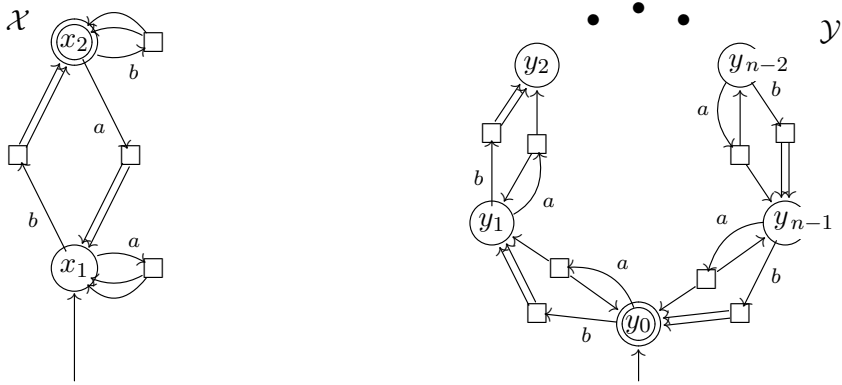
\[
  \xy <1.5mm,0cm>:
(10,22)*{\mathcal{X}}="",
(15,0)*+[Fo]{x_1} = "x1",
(10,10)*{\Box} = "xb1",
(22,0)*{\Box} = "xc1",
(15,20)*+[Foo]{x_2} = "x2",
(20,10)*{\Box} = "xb2",
(22,20)*{\Box} = "xc2",
(82,21)*{\mathcal{Y}}="",
(65,-5)*+[Foo]{y_0} = "y0", 
(58,0)*{\Box} = "yb0", 
(56,-4)*{\Box} = "yc0", 
(52,4)*+[Fo]{y_1} = "y1", 
(56,11)*{\Box} = "yb1", 
(52,12)*{\Box} = "yc1", 
(56,18)*+[Fo]{y_2} = "y2", 
(74,18)*+[Fo]{\phantom{y_2}} = "y4",
(76,18)*{\colorbox{white}{\smash{\makebox[0em]{$y$}}${\ }_{n-2}$}} = "", 
(74,10)*{\Box} = "yb4", 
(78,12)*{\Box} = "yc4", 
(78,4)*+[Fo]{\phantom{y_2}} = "y5",
(80,4)*{\colorbox{white}{\smash{\makebox[0em]{$y$}}${\ }_{n-1}$}} = "", 
(71,-1)*{\Box} = "yb5", 
(74,-4)*{\Box} = "yc5", 
(61,22)*{\bullet},
(65,23)*{\bullet},
(69,22)*{\bullet},
\ar (15,-10);"x1"*+++[o]{}
\ar ^{b} "x1"*++{};"xb1"*+[o]{}
\ar @<.5mm> "xb1";"x2"*+++[o]{}
\ar @<-.5mm> "xb1";"x2"*+++[o]{}
\ar _(.7){a} "x2"*++{};"xb2"*+[o]{}
\ar @<.5mm> "xb2";"x1"*+++[o]{}
\ar @<-.5mm> "xb2";"x1"*+++[o]{}
\ar ^(.7){a} @/^2mm/"x1";"xc1"*+[o]{}
\ar @/^2mm/ "xc1";"x1"*++[o]{}
\ar @/^4mm/ "xc1";"x1"*++[o]{}
\ar _(.7){b}@/_2mm/ "x2";"xc2"*+[o]{}
\ar @/_2mm/ "xc2";"x2"*++[o]{}
\ar @/_4mm/ "xc2";"x2"*++[o]{}
\ar (65,-10);"y0"*+++[o]{}
\ar _{a} @/_3mm/"y0"*++{};"yb0"*+[o]{}
\ar  "yb0";"y0"*+++[o]{}
\ar  "yb0";"y1"*+++[o]{}
\ar ^{b} "y0";"yc0"*+[o]{}
\ar @<.5mm> "yc0";"y1"*+++[o]{}
\ar @<-.5mm> "yc0";"y1"*+++[o]{}
\ar _{a} @/_3mm/"y1";"yb1"*+[o]{}
\ar  "yb1";"y1"*+++[o]{}
\ar  "yb1";"y2"*+++[o]{}
\ar ^{b} "y1"*++{};"yc1"*+[o]{}
\ar @<.5mm> "yc1";"y2"*+++[o]{}
\ar @<-.5mm> "yc1";"y2"*+++[o]{}
\ar _{a} @/_3mm/"y4";"yb4"*+[o]{}
\ar  "yb4";"y4"*+++[o]{}
\ar  "yb4";"y5"*+++[o]{}
\ar ^{b} "y4";"yc4"*+[o]{}
\ar @<.5mm> "yc4";"y5"*+++[o]{}
\ar @<-.5mm> "yc4";"y5"*+++[o]{}
\ar _{a} @/_3mm/"y5";"yb5"*+[o]{}
\ar "yb5";"y5"*+++[o]{}
\ar "yb5";"y0"*+++[o]{}
\ar ^{b} "y5";"yc5"*+[o]{}
\ar @<.5mm> "yc5";"y0"*+++[o]{}
\ar @<-.5mm> "yc5";"y0"*+++[o]{}
\endxy
\]
Here the ranked alphabet is given by $\Sigma=\{a,b\}$ where $|a|=|b|=2$.
Let $X$ and $Y$ be the state spaces of $\mathcal{X}$ and $\mathcal{Y}$ respectively,
and
define $X_1,X_2$ and $Y_1,Y_2$ as in Definition~\ref{def:simForNondetBuechiRel}.

We can see that $R_1=X_1\times Y_1$, $R_2=X_2\times Y_1$, $R_3=X_1\times Y_2$ and $R_4=X_2\times Y_2$
are the solution of the equational system (\ref{eq:1601061008Rel}) in Definition~\ref{def:simForNondetBuechiRel}
induced by $\mathcal{X}$ and $\mathcal{Y}$ here. 
Hence $R=X\times Y$ is a  fair simulation from $\mathcal{X}$ to $\mathcal{Y}$,
and this implies language inclusion. 
\end{exa}



We conclude this section by showing a relationship between the simulation notion via parity games
(Definition~\ref{def:parityFairSimulation})
and our simulation notion (Definition~\ref{def:simForNondetBuechiRel}).
Roughly speaking, 
a parity game is understood as a combinatorial presentation of
an equational system like (\ref{eq:1601061008Rel})
over finite lattices $L_{1},\dotsc,L_{m}$~\cite{HasuoSC16}.
%
%
%
If NBTAs $\mathcal{X}$ and $\mathcal{Y}$ have finite state-spaces,
translating (\ref{eq:1601061008Rel}) leads to the parity game in Definition~\ref{def:parityFairSimulation}.
Formally, we have the following proposition.
The proof is similar to the one for~\cite[Corollary A.5]{hasuoSC15arXiv}.

\begin{prop}\label{prop:parityGameInducedNew}
Let $\mathcal{X}=(X,\Sigma,\delta_\mathcal{X},\initSet_\mathcal{X},\Acc_\mathcal{X})$ and 
$\mathcal{Y}=(Y,\Sigma,\delta_\mathcal{Y},\initSet_\mathcal{Y},\Acc_\mathcal{Y})$ be 
 NBTAs such that $X$ and $Y$ are finite.
Then a fair simulation (Def.~\ref{def:simForNondetBuechiRel}) from $\mathcal{X}$ to $\mathcal{Y}$ exists \linebreak
if and only if
the player Even is winning in the parity game $G_{\mathcal{X},\mathcal{Y}}$ in Def.~\ref{def:parityFairSimulation} from $*$;
if that is the case we have
 $\lang(\mathcal{X})\subseteq\lang(\mathcal{Y})$.
\qed
\end{prop}

\section{Fair Simulation for Finite-State Probabilistic B\"uchi Word Automata}
\label{sec:fwdFairSimProb}

This is the second section in which we describe our technical contributions in 
concrete set-theoretic terms. They are derived from the theoretical backgrounds that we describe in later sections. In this section we focus on probabilistic systems.

In what follows we adopt the following conventions. 
The $(x,y)$-entry of a matrix $A\in[0,1]^{X\times Y}$ is denoted by
$A_{x,y}$;  the  $x$-th entry of a vector $\iota\in [0,1]^X$ is 
$\iota_x$. 
For
$A,B\in[0,1]^{X\times Y}$, we write $A\leq B$ if $A_{x,y}\leq B_{x,y}$ for all $x$ and $y$.

\begin{defi}[PBWA]\label{def:generativeProbBuechiTreeAutom}
A \emph{(generative) probabilistic B\"uchi word automaton} 
(PBWA) is 
  a quintuple
  $\mathcal{X}=(X,\myalphabet,M,\initVec,\Acc)$ consisting of 
  a countable state space $X$,
  a countable alphabet $\myalphabet$,
  transition matrices $M(a)\in[0,1]^{X\times X}$ for each $a\in\myalphabet$,
 an initial distribution $\initVec\in[0,1]^X$, 
  and a set $\Acc\subseteq X$ of accepting states.
  We require that the matrices $M(a)$ and the vector $\initVec$ are substochastic: 
 $\sum_{a\in\myalphabet}\sum_{x'\in X}(M(a))_{x,x'}\leq 1$ for
 each
 $x\in X$,
  and $\sum_{x\in X}\initVec_x\leq 1$.
\end{defi}

 Note that the initial vector and transition matrices are \emph{sub-}stochastic:
 $\sum_{a\in\myalphabet}\sum_{x'\in X}(M(a))_{x,x'}$ and $\sum_{x\in X}\initVec_x$ are allowed to be 
strictly smaller than $1$.
The missing probabilities 
are
for \emph{divergence}.
We require
$\sum_{a}\sum_{x'}(M(a))_{x,x'}\le 1$: this means our automaton is
 \emph{generative} and it chooses which character $a\in \myalphabet$ to
 \emph{output}. This is in contrast to a \emph{reactive} automaton (that
 \emph{reads} characters), in which case we would require
 $\sum_{x'}(M(a))_{x,x'}\le 1$ for each $a$.

\begin{exa}\label{example:PBWA}
We define a PBWA $\mathcal{X}=(X,\myalphabet,M,\initVec,\Acc)$ as follows.
%
\begin{itemize}
\item $X=\{x_1,x_2,x_3,x_4,x_5\}$
\item $\myalphabet=\{a,b\}$
\item 
$M(a)=\scalebox{0.8}{\bordermatrix{ & x_1 & x_2 & x_3 & x_4 & x_5 \cr x_1 & \nicefrac{1}{2}&  \nicefrac{1}{3}& 0& 0& 0\cr x_2 & 0& \nicefrac{1}{2} & \nicefrac{1}{3}& 0& 0\cr x_3 & 0& \nicefrac{1}{2}& \nicefrac{1}{2}& 0& 0\cr x_4 & 0& 0& 0& \nicefrac{1}{2}& \nicefrac{1}{2}\cr x_5 &0 & 0& 0& \nicefrac{1}{2}& \nicefrac{1}{2}}}$ \;\;and\;\;
$M(b)=\scalebox{0.8}{\bordermatrix{ & x_1 & x_2 & x_3 & x_4 & x_5 \cr x_1 & 0& 0& 0& \nicefrac{1}{6}& 0\cr x_2 &0 &0 &0 &0 & 0\cr x_3 & 0& 0&0 &0 & 0\cr x_4 &0 & 0& 0& 0& 0\cr x_5 & 0& 0& 0& 0& 0}}$  
\item $\initVec=\scalebox{0.8}{\bordermatrix{ & x_1 & x_2 & x_3 & x_4 & x_5 \cr & 1 & 0& 0& 0& 0}}$
\item $\Acc=\{x_3,x_5\}$
\end{itemize}
Then $\mathcal{X}$ is illustrated as below.
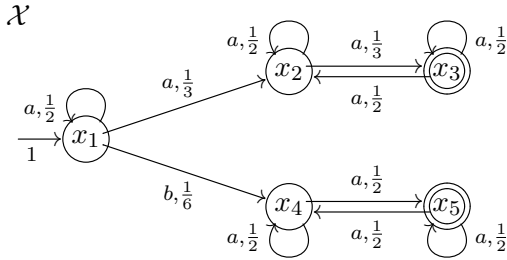
\[
\vspace{-3mm}
\raisebox{-7mm}{
\xy <1.5mm,0mm>:
(-6,29)*{\mathcal{X}}="",
(0,18)*+[Fo]{x_1} = "x1",
(18,24)*+[Fo]{x_2} = "x2",
(32,24)*+[Foo]{x_3} = "x3",
(18,12)*+[Fo]{x_4} = "x4",
(32,12)*+[Foo]{x_5} = "x5",
%
\ar _(.2){1} (-6,18);"x1"*+++[o]{}
\ar @(ur,ul) _(.8){a,\frac{1}{2}} "x1"*++{};"x1"*++[o]{}
\ar @(ur,ul) _(.8){a,\frac{1}{2}} "x2"*++{};"x2"*++[o]{}
\ar @(ur,ul) _(.2){a,\frac{1}{2}} "x3"*++{};"x3"*++[o]{}
\ar @(dr,dl) ^(.8){a,\frac{1}{2}} "x4"*++{};"x4"*++[o]{}
\ar @(dr,dl) ^(.2){a,\frac{1}{2}} "x5"*++{};"x5"*++[o]{}
\ar ^(.5){a,\frac{1}{3}} "x1"*++{};"x2"*+++[o]{}
\ar @<.7mm>^{a,\frac{1}{3}} "x2"*++{};"x3"*+++[o]{}
\ar @<.7mm>^{a,\frac{1}{2}} "x3"*++{};"x2"*+++[o]{}
\ar _(.5){b,\frac{1}{6}} "x1"*++{};"x4"*+++[o]{}
\ar @<.7mm>^{a,\frac{1}{2}} "x4"*++{};"x5"*+++[o]{}
\ar @<.7mm>^{a,\frac{1}{2}} "x5"*++{};"x4"*+++[o]{}
%
\endxy
}
\]
\end{exa}

In the next section we shall  give a definition of  accepted 
languages of PBWAs. This is rather standard (see~\cite{carayolHS14} for
a reactive variant).

\subsection{Accepted Languages of Probabilistic B\"{u}chi Word Automata}
\label{subsec:treeRunAccRunProb}
%
The language $\lang(\mathcal{X})$---a subprobability measure that tells which words are generated by what
probabilities---is essentially the \emph{push-forward
measure}~\cite{doob94measuretheory} obtained from
the one over the set $\Run^\giry_{\X}$ of \emph{runs} of a PBWA.

\begin{defi}[run]
For a PBWA $\mathcal{X}=(X,\myalphabet,M,\iota,\Acc)$,
a \emph{run} over $\mathcal{X}$ is an infinite word $\rho\in (\myalphabet\times X)^\omega$.
The set of all runs over $\mathcal{X}$ is denoted by $\Run^\giry_{\X}$.
A \emph{partial run} over $\mathcal{X}$ is a finite word $\xi\in (\myalphabet\times X)^*\times X$.
%
A run $\rho=(a_0,x_0)(a_1,x_1)\ldots\in \Run^\giry_{\X}$ is \emph{accepting} if 
$x_i\in \Acc$ for infinitely many $i$'s.
The set of all accepting runs over $\mathcal{X}$ is denoted by $\AccRun^{\giry}_{\X}$.
\end{defi}

We define the language of $\mathcal{X}$ as a subprobability measure over the set $\myalphabet^\omega$ of infinite words.
The set $\myalphabet^{\omega}$ of all infinite words over $\myalphabet$
carries a canonical ``cylindrical'' measurable structure generated by 
$\{w\myalphabet^\omega\mid w\in\myalphabet^*\}$
(see~\cite{baierK08principlesofmodelchecking} for example).
The set  $(\myalphabet\times X)^{\omega}$ of 
 runs comes with a  cylindrical measurable structure, too.

\begin{defi}
Let $\mathcal{X}=(X,\myalphabet,M,\iota,\Acc)$ be a PBWA.
For $w\in\myalphabet^*$, the \emph{cylinder set} generated 
by $w$ is a set 
\[
\Cyl(w):=\{ww'\in\myalphabet^\omega\mid w\in\myalphabet^*, w'\in\myalphabet^\omega\}\,. 
\]
We write $\sigalg_{\myalphabet^\omega}$ for
the smallest 
$\sigma$-algebra 
over $\myalphabet^\omega$ 
that is generated by 
the cylinder sets
$\{\Cyl(w)\mid w\in\myalphabet^*\}$. 

Similarly, for a partial run $\xi=(a_0,x_0)\ldots(a_{i-1},x_{i-1})x_i\in(\myalphabet\times X)^*\times X$,
the \emph{cylinder set} generated by $\xi$ is a set 
\[
\Cyl_\X(\xi):=\{(a_0,x_0)\ldots (a_i,x_i)(a_{i+1},x_{i+1})\ldots \in\Run^\giry_{\X}
\mid  a_{i},a_{i+1},\ldots\in \myalphabet, x_{i+1},x_{i+2},\ldots \in X\}\,.
\]
We write $\sigalg_{\X}$ for the $\sigma$-algebra over $\Run^\giry_{\X}$ 
generated by the cylinder sets 
$\{\Cyl_\X(\xi)\mid \xi\in(\myalphabet\times X)^*\times X\}$\,. 

We define $\DelSt:\Run^\giry_{\X}\to\myalphabet^\omega$ by 
$\DelSt((a_0,x_0)(a_1,x_1)\ldots)\;:=\;a_0a_1\ldots$\,.
\end{defi}

 Now it can be shown
that the set 
$\AccRun^{\giry}_{\X}$
of \emph{accepting} runs---visiting $\accstate$
infinitely often---is a measurable subset. 
This result (as stated in the following lemma)
is
much like~\cite[Lemma~36]{carayolHS14} and hardly novel.

\begin{lem}\label{lem:acceptingRunsAreMeasurableWord}
 The set $\AccRun^\giry_{\X}$ of accepting runs
  is an
 $\mathfrak{F}_{\X}$-measurable subset
  of $\Run^\giry_{\X}$.
\end{lem}

\begin{proof}
For each $k\in\mathbb{N}$, we define a set $\NAccRuninf_{k}\subseteq\Run_{\X}$ as follows.
\begin{equation}\label{eq:10142240}
\NAccRuninf_{k}\;:=\;\{(a_0,x_0)(a_1,x_1)\ldots\mid x_{k}\notin\Acc\}\,.
\end{equation}
Then we have:
\[
\NAccRuninf_{k} = 
\bigcup_{a_0\in\myalphabet}\ldots\bigcup_{a_{k-1}\in\myalphabet}
\bigcup_{x_0\in X}\ldots\bigcup_{x_{k-1}\in X}
\bigcup_{x_k\in X\setminus \Acc}
\Cyl\bigl((a_0,x_0)\ldots(a_{k-1},x_{k-1})x_{k}\bigr)\,.
\]
As $X$ and $\myalphabet$ are countable sets, 
by definition of the $\sigma$-algebra $\sigalg_{\X}$, $\NAccRuninf_{k}$ is measurable.

By definition of $\AccRun_\X$, 
it is easy to see that:
\[
\AccRun_\X=
\Run_\X\setminus
\bigcup_{m\in\mathbb{N}}
\bigcap_{n\in\mathbb{N}}
\NAccRuninf_{m+n}\,.
\]
Hence $\AccRun_\X$ is measurable.
\end{proof}

The following notion of \emph{no-divergence} probability plays an
important role.
Recall that
  a PBWA can exhibit divergence.
  
\begin{defi}[$\NDL_{\X}$]\label{def:noDeadendProbWord}
Let $\X=(X,\myalphabet,M,\initVec,\Acc)$ be
a PBWA.
  For each $k\in \nat$, 
  $\NDL_{\X,k}\colon X\to [0,1]$  is 
  defined inductively by:
\begin{equation}\label{eq:10151159}
  \begin{aligned}
  \NDL_{\X,0}(x)
  \;&:=\;
  1\enspace,
  \\
  \NDL_{\X,k+1}(x)
  \;&:=\;{}
  \sum_{a\in\myalphabet}
    \sum_{x'\in X}
  \left(M(a)\right)_{x,x'}\cdot
  \NDL_{\X,k}(x')
  \enspace.
  \end{aligned}
\end{equation}
Note that as $\sum_{a\in\myalphabet}\sum_{x'\in X}(M(a))_{x,x'}\leq 1$ for each $x$,
$\NDL_{\X,k}(x)$ is decreasing with respect to $k$.
We define 
a function $\NDL_{\X}\colon X\to [0,1]$ by
$\NDL_{\X}(x):= \lim_{k\to\infty}  \NDL_{\X,k}(x)$.
\end{defi}
Intuitively, 
$\NDL_{\X}(x)$ is a probability in which
  an execution of $\X$ from the state $x$ does not exhibit divergence.
This probability is used to define  probabilistic accepted languages of PBWAs.

We can now define a
subprobability
measure $\mu^{\Run^{\giry}_{\X}}_{\X}$ 
on $\Run^\giry_{\X}$
induced by the PBWA $\mathcal{X}$ (via the Carath\'eodory theorem).

\begin{defi}[$\mu_{\X}^{\Run^\giry_{\X}}$ over $\Run^\giry_{\X}$]\label{def:muXOnRunX}
Let $\X=(X,\myalphabet,M,\initVec,\Acc)$ be a PBWA.
We shall define a subprobability measure
$\mu_{\X}^{\Run^\giry_{\X}}$ over $(\Run^\giry_{\X},\sigalg_\X)$.
It is given, 
 for each partial run $\xi=(a_0,x_0)\ldots(a_{i-1},x_{i-1})x_i\in(\myalphabet\times X)^*\times X$, by
\begin{equation}\label{eq:01171719}
  \mu_{\X}^{\Run^\giry_{\X}}\bigl(\,\Cyl_{\X}(\xi)\,\bigr)
  \;:=\;
  \initVec_{x_{0}}\cdot P_{\X}(\xi)\enspace.
\end{equation}
Here   $P_{\X}(\xi)$
is
  defined inductively as follows. 
\begin{equation*}
  \begin{aligned}
 P_{\X}\bigl(
   \xi
    \bigr)
   :=
   \begin{cases}
    \NDL_{\X}(x_0)
    & (\text{$i=0$})
    \\
    \bigl(M(a_0)\bigr)_{x_0,x_1}\cdot P_\X\bigl((a_1,x_1)\ldots(a_{i-1},x_{i-1})x_i\bigr)
    & (\text{$i>0$})\,.
    \end{cases}
  \end{aligned}
\end{equation*}
\end{defi}

\begin{prop}\label{prop:wellDef:def:muXOnRunX}
Definition~\ref{def:muXOnRunX} is well-defined.
That is to say, there exists a unique subprobability measure 
$\mu_{\X}^{\Run^\giry_{\X}}$ over $(\Run_\X,\sigalg_\X)$ that satisfies the equation
(\ref{eq:01171719}).
\end{prop}

\begin{proof}
We first prove that for each $i\in\mathbb{N}$,
$a_0,\ldots,a_{i-1}\in\myalphabet$ and $x_0,\ldots,x_i\in X$, we have:
\[
P_\X\bigl((a_0,x_0)\ldots(a_{i-1},x_{i-1})x_i\bigr)=
\sum_{a_i\in\myalphabet}\sum_{x_{i+1}\in X}
P_\X\bigl((a_0,x_0)\ldots(a_{i-1},x_{i-1})(a_i,x_i)x_{i+1}\bigr)\,.
\]

We prove it by the induction on $i$.
\begin{itemize}
\item
If $i=0$, then $\xi=x_0$ and hence we have:
\allowdisplaybreaks[4]
\begin{align*}
P_\X(\xi)
&=\NDL_\X(x_0) \\
&=\lim_{k\to\infty}\NDL_{\X,k}(x) \\
&=\lim_{k\to\infty}\sum_{a\in\myalphabet}\sum_{x_1\in X}\left(M(a)\right)_{x_0,x_1}\cdot\NDL_{\X,k-1}(x_1) \\
&=\lim_{k\to\infty}\sum_{a\in\myalphabet}\sum_{x_1\in X}\left(M(a)\right)_{x_0,x_1}\cdot\NDL_{\X,k}(x_1) \\
&= \sum_{a_0\in\myalphabet}\sum_{x_{1}\in X}\left(M(a_0)\right)_{x_0,x_{1}}\cdot\bigl(\lim_{k\to\infty}\NDL_{\X,k}(x_1)\bigr) \\
&= \sum_{a_0\in\myalphabet}\sum_{x_{1}\in X}\left(M(a_0)\right)_{x_0,x_{1}}\cdot P_\X(x_1)\\
&= \sum_{a_0\in\myalphabet}\sum_{x_{1}\in X} P_\X\left((a_0,x_0)x_{1}\right)
\end{align*}

\item
If $i>0$, then 
we have:
\begin{align*}
P_\X(\xi)
&= \bigl(M(a_0)\bigr)_{x_0,x_1}\cdot P_\X\bigl((a_1,x_1)\ldots(a_{i-1},x_{i-1})x_i\bigr)\\
&= \bigl(M(a_0)\bigr)_{x_0,x_1}\cdot \sum_{a_i\in\myalphabet}\sum_{x_{i+1}\in X} P_\X\bigl((a_1,x_1)\ldots(a_{i-1},x_{i-1})(a_i,x_i)x_{i+1}\bigr) \\
& \qquad\qquad\qquad\qquad\qquad\qquad\qquad\qquad\qquad\qquad\qquad(\text{by the induction hypothesis})\\
&=  \sum_{a_i\in\myalphabet}\sum_{x_{i+1}\in X} \bigl(M(a_0)\bigr)_{x_0,x_1}\cdot P_\X\bigl((a_1,x_1)\ldots(a_{i-1},x_{i-1})(a_i,x_i)x_{i+1}\bigr)\\
&=  \sum_{a_i\in\myalphabet}\sum_{x_{i+1}\in X} P_\X\bigl((a_0,x_0)(a_1,x_1)\ldots(a_{i-1},x_{i-1})(a_i,x_i)x_{i+1}\bigr)\,.
\end{align*}
\end{itemize}

Hence we have: 
\[
\mu_\X^{\Run^\giry_\X}((a_0,x_0)\ldots(a_{i-1},x_{i-1})x_i)=
\sum_{a_i\in\myalphabet}\sum_{x_{i+1}\in X}
\mu_\X^{\Run^\giry_\X}((a_0,x_0)\ldots(a_{i-1},x_{i-1})(a_i,x_i)x_{i+1})\,.
\]
Therefore Proposition~\ref{prop:wellDef:def:muXOnRunX} 
is immediate from Carath\'{e}odory's extension theorem (see~\cite{ashD2000probability} for example).
\end{proof}


Now we can define the language of a PBWA $\mathcal{X}$.

\begin{defi}[language of PBWA]
\label{def:muXOnTreeSigma}
$\X=(X,\myalphabet,M,\initVec,\Acc)$ be a PBWA.
A subprobability measure $\lang(\mathcal{X})$ over
 $(\myalphabet^\omega,\sigalg_{\myalphabet^\omega})$ is defined as follows:
for each $w\in\myalphabet^*$,
\begin{equation}\label{eq:01171734}
 \lang(\X)(\Cyl(w))
 \;:=\;
 \mu_{\X}^{\Run^\giry_{\X}}
\left(\,
  \DelSt^{-1}(\Cyl(w))
   \cap
  \AccRun^\giry_{\X}
\,\right)\enspace.
\end{equation}
Note here that 
 $\DelSt^{-1}(\Cyl(w)) 
 \;=\;
  \textstyle\bigcup_{\xi\in\DelSt^{-1}(\{w\})} \Cyl_{\X}(\xi)$.
As $w$ is a finite word and the state space $X$ is countable,
the union in the above equation is a countable one.
Hence the set $\DelSt^{-1}(\Cyl(w))$ is measurable.
\end{defi}

The following proposition can be proved
 in a similar manner to 
Proposition~\ref{prop:wellDef:def:muXOnRunX}.

\begin{prop}\label{prop:wellDef:def:muXOnTreeSigma}
Definition~\ref{def:muXOnTreeSigma} is well-defined.
That is, there exists a unique subprobability measure 
$\lang(\mathcal{X})$ over $(\Sigma^\omega,\sigalg_{\myalphabet^\omega})$ that satisfies the equation
(\ref{eq:01171734}).
\qed
\end{prop}

\begin{exa}\label{example:PBWALang}
Let $\mathcal{X}$ be the PBWA in Example~\ref{example:PBWA}.
For each cylinder set $w\myalphabet^\omega$ where $w\in\myalphabet^*$, 
the subprobability measure $\lang(\mathcal{X})$ assigns the following probability.
\[
\lang(\mathcal{X})(w\myalphabet^\omega)=\begin{cases}
\frac{1}{2^{n}}\cdot\frac{1}{3} & (\text{$w=\underbrace{a\ldots a}_n$, \; $\underbrace{a\ldots a}_n b$ \; or \; $\underbrace{a\ldots a}_n b a\ldots a$}) \\
0 & (\text{otherwise})
\end{cases}
\]
\end{exa}

 \subsection{Fair Simulation for PBWAs}
 \label{subsec:defSimSoundnessProb}
We continue to introduce \emph{fair simulation} for PBWAs. This is one of our main contributions:
to the best of our knowledge this
is the first one for \emph{probabilistic} B\"{u}chi
(word)  automata. Note that our simulation is given by a matrix and not
by a relation; this follows our previous
work~\cite{hasuo06genericforward,urabeH17MatrixSimulationJourn}.

\begin{defi}[fair simulation for PBWAs] 
\label{def:fwdFairBuechiSimProbMatrix}
Let $\X=(X,\myalphabet,M_\X,\initVec_\X,\Acc_\X)$ and $\mathcal{Y}=(Y,\myalphabet,M_\mathcal{Y},\initVec_\mathcal{Y},\Acc_\mathcal{Y})$ 
be probabilistic B\"{u}chi word automata
with the same alphabet $\myalphabet$.
Let $A$ be a matrix such that $A\in [0,1]^{Y\times X}$.
 We define
 $X_1=X\setminus \Acc_\X$ and $X_2=\Acc_\X$ (like in
 Definition~\ref{def:simForNondetBuechiRel}), and similarly for $Y_{1}$ and $Y_{2}$.
 Moreover, let $M_{\mathcal{X},i}(a)\in[0,1]^{X_i\times X}$, $M_{\mathcal{Y},j}(a)\in[0,1]^{Y_j\times Y}$ and
 $A_{ji}\in[0,1]^{Y_j\times X_i}$ denote the obvious partial matrices
 of $M_{\mathcal{X}}(a)\in[0,1]^{X\times X}$,  $M_{\mathcal{Y}}(a)\in[0,1]^{Y\times Y}$ and $A\in[0,1]^{Y\times X}$, respectively.
We say that the matrix $A$ is a \emph{fair simulation} from 
 $\X$ to $\mathcal{Y}$ 
 if it is 
satisfies the following conditions. 
 \begin{enumerate}
 \item\label{item:beingProbMat}
 The matrix $A$ is a substochastic matrix:
 $\sum_{x\in X} A_{y,x} \leq 1$ for each $y\in Y$.
  \item\label{item:f21Andf22Mat}
  The matrix $A$ is a \emph{forward simulation
       matrix}~\cite{urabeH14CONCUR,urabeH17MatrixSimulationJourn}, that is,
\begin{math}
 \initVec_{\mathcal{X}} \leq \initVec_{\mathcal{Y}}\cdot A
\end{math}
and 
\begin{math}
  A\cdot M_{\mathcal{X}}(a) \leq M_{\mathcal{Y}}(a)\cdot A
\end{math}
for each $a\in\myalphabet$.

  \item \label{item:f11Andf12Mat} 
  The partial matrices $A_{11}\in[0,1]^{Y_1\times X_1}$ and $A_{12}\in[0,1]^{Y_1\times X_2}$ 
   come	with 
	their \emph{approximation sequences}.
	They are 
	increasing sequences of length 
	$\overline{\alpha}\leq\omega$:
       \[
       \begin{array}{ll}
        A_{11}^{(0)} \leq
        A_{11}^{(1)} \leq
        \cdots 
	\leq
	A_{11}^{(\overline{\alpha})} 
	\,\in
	[0,1]^{Y_{1}\times X_{1}}
	\quad\text{and}\quad	
        A_{12}^{(0)} \leq
        A_{12}^{(1)} \leq
        \cdots \leq
	A_{12}^{(\overline{\alpha})} 
	\,\in
	[0,1]^{Y_{1}\times X_{2}}
\end{array}       
\]
	such that:
\begin{enumerate}
 \item\label{item:f11Andf12_2Mat}
 \textbf{(Approximate $A_{11}$ and $A_{12}$)} 
       We have
       $A_{11}^{(\overline{\alpha})}=A_{11}$ and
       $A_{12}^{(\overline{\alpha})}=A_{12}$.
 \item\label{item:f11Andf12_3Mat}
 \textbf{($A_{11}^{(\alpha)}$)} 
       For each  $\alpha\le\overline{\alpha}$ and $a\in\myalphabet$
       we have:
       \begin{math}
       A_{11}^{(\alpha)}\cdot M_{\mathcal{X},1}(a)\;\leq\;
       M_{\mathcal{Y},1}(a)\cdot 
	       {\scriptsize
       \begin{pmatrix} A_{11}^{(\alpha)} & A_{12}^{(\alpha)} \\ 
       A_{21} & A_{22} \end{pmatrix}
	}
  \end{math}.
 \item\label{item:f11Andf12_4Mat}
 \textbf{($A_{12}^{(\alpha)}$,  base)} 
       	      The $0$-th approximant $A_{12}^{(0)}$ is the zero matrix $O$.
 \item\label{item:f11Andf12_5Mat}
 \textbf{($A_{12}^{(\alpha)}$,  step)} 
        For each  $\alpha<\overline{\alpha}$ and $a\in\myalphabet$:
%
	  \begin{math}
	         A_{12}^{(\alpha+1)}\cdot M_{\mathcal{X},2}(a)\;\leq\;
       M_{\mathcal{Y},1}(a)\cdot 
       {\scriptsize
	   \begin{pmatrix} A_{11}^{(\alpha)} & A_{12}^{(\alpha)} \\ 
	   A_{21} & A_{22} \end{pmatrix}}
	  \end{math}.

 \item \label{item:f11Andf12_6Mat}
 \textbf{($A_{12}^{(\alpha)}$,  limit)}
	$(A_{12}^{(\omega)})_{y,x} = \sup_{\alpha'<\omega} (A_{12}^{(\alpha')})_{y,x}$ 
              for each $y\in Y_1$ and $x\in X_2$, in case $\overline{\alpha}=\omega$. 
\end{enumerate}
 \end{enumerate}
\end{defi}

\noindent
This notion is the combination of: 1)
Kleisli simulation
(see~\cite{urabeH17MatrixSimulationJourn} and also Table~\ref{table:knownResultsOnCoalgTraceAndSimulation}(c) later) for 
mimicking one-step behaviors; and 2) progress measure~\cite{HasuoSC16}
that accounts for the nonlocal ``fairness'' constraint
(Section~\ref{sec:prelim}). Indeed, Condition~(\ref{item:f21Andf22Mat})
and~(\ref{item:f11Andf12_3Mat}) express the \emph{invariant/gfp} intuition---note that (bi)simulation (without fairness) is a coinductive notion---while Condition~(\ref{item:f11Andf12_4Mat})--(\ref{item:f11Andf12_6Mat})
bears the \emph{ranking function/lfp} flavor, mirroring
the Cousot-Cousot approximation sequence $\bot\sqsubseteq
f(\bot)\sqsubseteq\cdots$.

\begin{thm}[soundness]\label{thm:soundnessOfSimulationGiryMainMatrix}
Assume 
 $\mathcal{Y}$ has a finite state space. 
 Existence of a fair simulation (Definition~\ref{def:fwdFairBuechiSimProbMatrix})
implies trace inclusion:
$
L(\X)(P)
\leq 
L(\mathcal{Y})(P)$ for 
any measurable 
$P
\subseteq \myalphabet^{\omega}
$.
\end{thm}
The proof is presented later in Section~\ref{sec:soundnessProof}, after we introduce coalgebraic machinery behind the definition of simulation.

We emphasize again that, differently from the nondeterministic setting, 
soundness of simulation is ensured only for \emph{word} automata with a
\emph{finite} state space on the simulating side.


A (nontrivial) example of such a fair simulation 
is given 
below.

\begin{exa}\label{example:fwdFairSimProbMain}
Let $\mathcal{X}$ and $\mathcal{Y}$ be the PBWAs 
illustrated 
below.
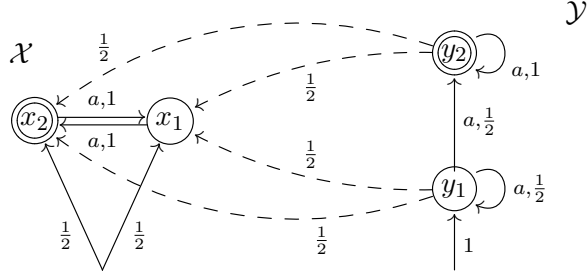
\[
\xy <1.8mm,0mm>:
(6,0)*+[Fo]{y_1} = "x0",
(6,10)*+[Foo]{y_2} = "x1",
(-15,5)*+[Fo]{x_1} = "y0", 
(-25,5)*+[Foo]{x_2} = "y1", 
(15,13)*{\mathcal{Y}}="",
(-26,10)*{\mathcal{X}}="",
\ar _(.3){1} (6,-6);"x0"*+++[o]{}
\ar _{a,\frac{1}{2}} "x0";"x1"*+++[o]{}
\ar @(ru,rd)^{a,\frac{1}{2}} "x0";"x0"
\ar @(ru,rd)^(.6){a,1} "x1";"x1"
\ar _(.3){\frac{1}{2}} (-20,-6);"y0"*+++[o]{}
\ar ^(.3){\frac{1}{2}} (-20,-6);"y1"*+++[o]{}
\ar @<.5mm>^{a,1} "y0";"y1"*+++[o]{}
\ar @<.5mm>^{a,1} "y1";"y0"*+++[o]{}
%
\ar @/^3mm/@{-->}_{\frac{1}{2}} "x0";"y0"*+++[o]{}
\ar @/^8mm/@{-->}^(.3){\frac{1}{2}} "x0";"y1"
\ar @/_3mm/@{-->}^{\frac{1}{2}} "x1";"y0"*+++[o]{}
\ar @/_8mm/@{-->}_(.8){\frac{1}{2}} "x1";"y1"
\endxy
\]
We define 
$A\in [0,1]^{\{y_1,y_2\}\times\{x_1,x_2\}}$ by
$A_{y_i,x_j}=\frac{1}{2}$ for each $i,j\in\{1,2\}$.
Then $A$ is a fair matrix simulation from $\mathcal{X}$ to $\mathcal{Y}$.
Here the approximation sequences 
$A_{11}^{(0)}\sqsubseteq A_{11}^{(1)}\sqsubseteq\cdots\sqsubseteq A_{11}^{(\omega)}$ and
$A_{12}^{(0)}\sqsubseteq A_{12}^{(1)}\sqsubseteq\cdots\sqsubseteq A_{12}^{(\omega)}$ are given by
$A_{11}^{(i)}=\bigl(\frac{1}{2}-\left(\frac{1}{2}\right)^{i+1}\bigr)\in[0,1]^{\{y_1\}\times\{x_1\}}$ and 
$A_{12}^{(i)}=\bigl(\frac{1}{2}-\left(\frac{1}{2}\right)^{i+1}\bigr)\in[0,1]^{\{y_1\}\times\{x_2\}}$
for each $i$.
\end{exa}

\section{Coalgebraic Background}
\label{sec:coalgebraicBackgrounds}

The  fair simulation notions in Sections~\ref{sec:fwdFairSimNondet}--\ref{sec:fwdFairSimProb} (for
nondeterminism and probability)  may look different, but they 
arise from the same source, namely our coalgebraic study of
B\"{u}chi automata~\cite{urabeSH16parityTrace}. 

\subsection{Modeling a System as a Function $X\to TFX$}
\label{subsec:XToTFX} 
The conventional coalgebraic modeling of
systems---as a function $X\to FX$---is known to capture
\emph{branching-time} semantics such as
bisimilarity~\cite{jacobs16CoalgBook,Rutten00a}. In contrast
accepted languages of B\"{u}chi automata with nondeterministic or
probabilistic branching constitute \emph{linear-time} semantics;
see~\cite{vanGlabbeek01} for the so-called \emph{linear time-branching
time spectrum}.

For the coalgebraic modeling of such linear-time semantics we 
follow the ``Kleisli modeling'' tradition~\cite{PowerT97,jacobs04tracesemantics,hasuo07generictrace}.
Here a system is parametrized by a monad $T$ and an endofunctor $F$ on $\Sets$:
the former represents the \emph{branching type} while
the latter represents the \emph{(linear-time) transition type}; and 
a system is modeled as a function of the type $X\to TFX$.\footnote{
Another
 eminent approach to coalgebraic linear-time semantics is the
 \emph{Eilenberg-Moore} one (see~\cite{Jacobs0S15,AdamekBHKMS12} for example):
 notably
in the latter
 a system is expressed as $X\to FTX$. 
The Eilenberg-Moore approach can be seen as a categorical
 generalization of \emph{determinization} or the \emph{powerset construction}. This however makes the approach hard to apply to
 \emph{infinite} words or trees, since already for B\"{u}chi word
 automata, it is known that deterministic ones are less expressive than
 general, nondeterministic ones.
}

A \emph{monad} $T$ is a construct from category theory~\cite{MacLane71}: it is
 a functor $T\colon \C\to\C$ equipped with \emph{unit}
$\eta^{T}_{X}\colon X\to TX$ and \emph{multiplication}
$\mu^{T}_{X}\colon T^{2}X\to TX$, both given by arrows in $\C$ for each
 object $X\in\C$,
subject to some axioms. 
In this paper we use two examples $T=\pow,\giry$: the
\emph{powerset monad} $\pow$  (on the category $\Sets$ of sets and
functions) for nondeterminism; and 
the \emph{sub-Giry monad} $\giry$ (on  $\Meas$ of measurable
spaces and
measurable functions) for probabilistic branching. The latter is a
``sub'' variant of the well-known \emph{Giry
monad}~\cite{giry82categoricalapproach}.

\begin{defi}[the monads $\pow$ and $\giry$]\label{def:powersetMonadAndGiryMonad}
The \emph{powerset monad} $\pow$ on $\Sets$ carries a set $X$ to $\pow
 X=\{S\subseteq X\}$, and a function $f\colon X\to Y$ to $\pow f\colon \pow
 X\to \pow Y$, $S\mapsto f[S]=\{f(x)\mid x\in S\}$. For each set $X$, its unit
 $\eta^{\pow}_{X}\colon X\to \pow X$ is given by the singleton map
 $x\mapsto \{x\}$; and its multiplication 
 $\mu^{\pow}_{X}\colon \pow^{2}X\to \pow X$ is given by union $M\mapsto
 \bigcup_{A\in M}A$. 

The \emph{sub-Giry monad} $\giry$ on $\Meas$ carries a measurable space
$(X,\sigalg_X)$ to $(\giry X, \sigalg_{\giry X})$, where $\giry X$ is
the set of all \emph{subprobability measures} on $X$
and
 $\sigalg_{\giry X}$ is the
 smallest $\sigma$-algebra such that, for each
$S\in\sigalg_X$, the function $\text{ev}_S:\giry X\to[0,1]$ defined by
$\text{ev}_S(P)=P(S)$ is measurable. The action of $\giry$ on arrows is
 given by the pushforward measure: 
for $f\colon X\to Y$, $P\in \giry X$
 and $T\in \sigalg_Y$,  $(\giry f)(P)(T)=P(f^{-1}(T))$. 
The unit
 $\eta^{\giry}_{X}\colon X\to \giry X$ is given by the \emph{Dirac measure}
 $\eta^{\giry}_{X}(x)= \delta_{x}$;
and
 $\mu^{\giry}_{X}\colon \giry^{2}X\to \giry X$ is given by
 $\Psi\mapsto 
\bigl(S\mapsto\int_{\giry (X,\sigalg_X)}
\text{ev}_S \,d\Psi
\bigr)
$.
\end{defi}
Intuitively
$\eta^{T}_{X}\colon X\to TX$ 
\emph{turns an element into a trivial branching}
while $\mu^{T}_{X}\colon T^{2}X\to TX$
\emph{suppresses two successive branchings into one}. See~\cite{hasuo07generictrace} for further illustration.

For the other parameter $F$---for the type of linear-time
behaviors---we use the following.
\begin{defi}[the functors $F_{\Sigma}$ on $\Sets$ and
 $F_{\myalphabet}$ on $\Meas$]\label{def:FSigmaAndFLambda}
Let $\Sigma$ be a ranked alphabet. The
 functor $F_{\Sigma}\colon \Sets\to\Sets$ carries a set $X$ to
 $F_{\Sigma}X=\coprod_{\sigma\in \Sigma}X^{|\sigma|}$; 
 and a function
  $f$ to $\coprod_{\sigma\in \Sigma}f^{|\sigma|}$. 
Let $\myalphabet$ be a countable alphabet, thought of as a measurable set with the
 discrete $\sigma$-algebra. 
 The functor $F_{\myalphabet}=\myalphabet\times(\place)\colon \Meas\to\Meas$ carries a measurable
 space $X$ to the product space $\myalphabet\times X$; and a measurable map
 $f$
 to $\id_{\myalphabet}\times f$. 
\end{defi}

Our system models
in Sections~\ref{sec:fwdFairSimNondet}--\ref{sec:fwdFairSimProb} readily
allow categorical modeling as arrows $X\to TFX$: the transition function
of an NBTA (Definition~\ref{def:nondetBuechiTreeAutom}) is a function $X\to
\pow F_{\Sigma}X$; and the transition matrices of a PBWA
(Definition~\ref{def:generativeProbBuechiTreeAutom}) collectively give a
(measurable) function $X\to \giry F_{\myalphabet}X$.

\subsection{Coalgebras in a Kleisli Category}
\label{subsec:coalgInKleisli}

Given a monad $T$ on a category $\C$, the standard construction of the \emph{Kleisli category}
$\Kl(T)$ is defined as follows (see~\cite{MacLane71} for example): its objects are those of $\C$; its arrows
$f\colon X\kto Y$ are precisely arrows $f\colon X\to TY$ in
$\C$; and its identity and composition $\odot$ are defined with the aid
of 
unit
$\eta^{T}$ and multiplication $\mu^{T}$.\footnote{For distinction we write $\kto$ for arrows in $\Kl(T)$ (not $\to$), and $\odot$ for composition in $\Kl(T)$
(not $\circ$).}
It is known that an arrow $f:X\to Y$ in $\C$ can be lifted to the Kleisli category $\Kl(T)$ 
by the \emph{Kleisli inclusion functor} $J\colon \Sets\to \Kl(T)$ that is defined by 
$f\mapsto \eta_Y\circ f$~\cite{MacLane71}. 

Intuitively
 a Kleisli arrow $f\colon
X\kto Y$ is a function from $X$ to $Y$ \emph{with
$T$-branching}. 
Then a system dynamics $X\to TFX$ with $T$-branching over
linear-time $F$-behaviors is a Kleisli arrow $X\kto {\oF}X$, a (proper)
$\oF$-coalgebra
in $\Kl(T)$.
Here $\oF\colon \Kl(T)\to\Kl(T)$ is a canonical lifting of $F\colon\C\to \C$, which
is formally defined as follows.

\begin{defi}\label{def:liftFunct}
For $F:\C\to\C$, 
a functor $\overline{F}:\Kl(T)\to\Kl(T)$ is called a \emph{lifting} of a functor $F:\C\to\C$ if
$\oF X= FX$ and
$\oF\circ J=J\circ F$. 
\end{defi}
A canonical lifting can be explicitly described
 when $T=\pow$ and $F=F_{\Sigma}$ on $\Sets$, and when $T=\giry$ and
$F=F_{A}$ on $\Meas$. See~\cite{hasuo07generictrace,urabeH15CALCO} for example.

%

Studies of coalgebras $X\kto {\oF}X$ are initiated in~\cite{PowerT97}
and developed henceforth 
in~\cite{jacobs04tracesemantics,hasuo07generictrace,cirstea10genericinfinite,kerstan13coalgebraictrace,urabeH17MatrixSimulationJourn,urabeH15CALCO} for example,
leading to the following \emph{coalgebraic}\newline\emph{theory of  trace
and simulation}. 

 \myparagraph{(Table~\ref{table:knownResultsOnCoalgTraceAndSimulation}(a))}
       In~\cite{hasuo07generictrace} it is shown that, for $T=\pow$
       (for nondeterminism) and $\dist$ (the 
       \emph{subdistribution} monad on $\Sets$ for discrete probabilities), and for
       a suitable functor $F$ on $\Sets$, an initial $F$-algebra
       $\alpha\colon FA\iso A$ in $\Sets$ yields a final
       $\oF$-coalgebra $J\alpha^{-1}\colon A\kto\oF
       A$. 
       In case $F=F_{\Sigma}$ 
       an initial algebra is given by
       the set of all \emph{finite} $\Sigma$-trees; and 
       the unique morphism $\tr(c)\colon X\kto A$---namely a function $\tr(c)\colon
       X\to TA$, see Table~\ref{table:knownResultsOnCoalgTraceAndSimulation}(a)---is nothing but the \emph{finite trace semantics} of 
       the automaton $c\colon X\kto \oF X$, capturing all the
       linear-time behaviors that \emph{eventually terminate}.

 \myparagraph{(Table~\ref{table:knownResultsOnCoalgTraceAndSimulation}(b))} 
       For \emph{infinitary trace
       semantics}  
       its coalgebraic characterization is more
       involved~\cite{jacobs04tracesemantics,cirstea10genericinfinite}. 
       Here we consider 
       all possibly nonterminating linear-time behaviors of an automaton.
       In the above setting, and also for $T=\giry$ on $\Meas$,
       it is
       shown that a final coalgebra $\zeta\colon Z\iso
       FZ$---we have $Z\iso\myTree_{\Sigma}$
       when
       $F=F_{\Sigma}$---yields
       a \emph{weakly final} coalgebra $J\zeta$ in $\Kl(T)$. Given  $c$ there is thus at least one morphism from $c$ to
       $J\zeta$;   there is also a \emph{maximal}
       such $\trinf(c)$, and this is how we capture infinitary trace. 
       In
       Table~\ref{table:knownResultsOnCoalgTraceAndSimulation}(b) we indicate
       this maximality by $\nu$.

      \myparagraph{(Table~\ref{table:knownResultsOnCoalgTraceAndSimulation}(c))} 
      In~\cite{hasuo06genericforward} it is shown that \emph{lax/oplax
      homomorphisms}
      (Table~\ref{table:knownResultsOnCoalgTraceAndSimulation}(c)) witness
      \emph{finite trace inclusion} $\tr(c)\sqsubseteq \tr(d)$. When $T=\pow$ these notions specialize to
       \emph{forward} and \emph{backward simulation}
      in~\cite{lynch95forwardand}, namely binary relations that
      ``mimic.''
      In~\cite{urabeH15CALCO} they
      are shown to witness \emph{infinitary trace inclusion} too;
      this is the starting point of the current study of (forward)  simulation
      for B\"{u}chi automata. 
      Note that, when $T=\giry$, our (forward) ``simulation'' is not a
      relation but a ``function with probabilistic branching''
      $f\colon Y\to \giry X$. The latter is roughly a matrix
      of dimension $|Y|\times |X|$; and algorithms to find such are
      studied in~\cite{urabeH17MatrixSimulationJourn}. 

 \begin{table}[tbp]
 \begin{tabular}{cccc}
    \begin{xy}\scriptsize
         \def\labelstyle{\textstyle}
     \xymatrix@R=1.4em@C=1.2em{
  {\overline{F} X} \ar@{}[drr]|{=} \kar@{-->}[rr]^{\overline{F}(\tr(c))} & & {\overline{F}A}  \\
  {X} \kar[u]^{c}  \kar@{-->}_{\tr(c)}[rr]  & & {A} \kar[u]_{J\alpha^{-1}}
  }
   \end{xy}
  &
    \begin{xy}\scriptsize
         \def\labelstyle{\textstyle}
 \xymatrix@R=1.4em@C=1.2em{
 {\overline{F} X} \ar@{}[drr]|{\color{red}=_{\mathbf{\nu}}} \kar@{->}[rr]^{\overline{F}(\trinf(c))} & & {\overline{F}Z}  \\
 {X} \kar[u]^{c}  \kar@{->}_{\trinf(c)}[rr]  & & {Z} \kar[u]_{J\zeta}
 }
   \end{xy}
  &
    \begin{xy}\scriptsize
         \def\labelstyle{\textstyle}
  \xymatrix@R=1.4em@C=1.2em{
  {\overline{F} X} \ar@{}[drr]|{\sqsubseteq}  & & {\overline{F} Y} \kar[ll]_{\overline{F} f } \\
  {X} \kar[u]^{c} 
  & & {Y} \kar[u]_{d}
    \ar_(.3){f}|-*\dir{|}[ll] 
}
   \end{xy}
  
    \begin{xy}\scriptsize
         \def\labelstyle{\textstyle}
  \xymatrix@R=1.4em@C=1.2em{
  {\overline{F} X} \ar@{}[drr]|{\sqsubseteq} \kar[rr]^{\overline{F} b } & & {\overline{F} Y}  \\
  {X} \kar[u]^{c} 
  \ar^(.3){b}|-*\dir{|}[rr]  &
    & {Y} \kar[u]_{d}  
}
   \end{xy}
   \\
  \begin{tabular}{l}   \scriptsize
    (a) \emph{Coalgebraic finite trace}:
   \\\scriptsize
   $FA\stackrel{\alpha}{\to}A$ is an init.\ alg.\ in $\Sets$
  \end{tabular}   
  \!\!\!\!\!\!\!\!\!
  & 
      \begin{tabular}{l}   \scriptsize
   (b)
   \emph{Coalgebraic infinitary trace}:
       \\\scriptsize
       $Z\stackrel{\zeta}{\to}FZ$
  is a final coalg.\ in $\Sets$
   \end{tabular}
  \!\!\!\!\!\!\!\!\!
   &
      \begin{tabular}{l}   \scriptsize
   (c)
   \emph{Coalgebraic fwd.\ and bwd.\ simulation}: 
       \\   \scriptsize
   here $c$ is simulating by $d$
       \end{tabular}
      \end{tabular}
 \caption{Some known results in the coalgebraic theory of trace and
  simulation}
 \label{table:knownResultsOnCoalgTraceAndSimulation}
 \end{table}

\begin{wrapfigure}[4]{r}{2.5cm}
\hspace{-.1em}\scriptsize
\raisebox{-.5cm}[0pt][0cm]{\
\begin{math}
 \vcenter{\xymatrix@C-.4em@R=1.7em{
  {FX}
     \ar[r]^-{Ff}
  &
  {FY}
  \\
  {X}
    \ar[r]_{f}
    \ar[u]^{c}
 &
  {Y}
    \ar[u]_{d}
}}
\end{math}
}
\end{wrapfigure}
\subsection{Coalgebraic Modeling of B\"{u}chi Automata}
\label{subsec:coalgebraicModelingOfBuechiAutomata}
In the above theory---and in the theory of coalgebra in
general---the B\"{u}chi acceptance condition has long been considered a big challenge:
its  nonlocal character (``visit $\accstate$ infinitely often'') 
does not go along  with the coalgebraic, local idea of behaviors that is
centered around \emph{homomorphisms} of coalgebras ($f$ in the
diagram).

Our answer~\cite{urabeSH16parityTrace} to the
challenge,
inspired by Table~\ref{table:knownResultsOnCoalgTraceAndSimulation}(b) 
and our recent~\cite{HasuoSC16},
 consists of: 1) regarding the distinction of $\nonaccstate$ vs.\
$\accstate$ as a \emph{partition} $X=X_{1}+X_{2}$ of the state space; and 2) introducing explicit $\mu$'s
and $\nu$'s in commuting diagrams, hence regarding them as part of
\emph{equational systems} (Section~\ref{sec:prelim}). This forces our departure
from the  coalgebraic reasoning principle of
\emph{finality}---namely \emph{existence} of a \emph{unique}
homomorphism---by moving from 
Table~\ref{table:knownResultsOnCoalgTraceAndSimulation}(a)
to~(\ref{eq:diagramsForEqSysForBuechiAcceptance}) below.  We however believe this is a necessary step forward,
for the theory of coalgebras to cope with its long-standing challenges
like
the B\"{u}chi condition and weak bisimilarity. 

We review the part of the theory
in~\cite{urabeSH16parityTrace} that is relevant to us.

\begin{defi}[B\"{u}chi $(T,F)$-system]\label{def:buechiTFSys}
 Let $T$ be a monad, and $F$ be an endofunctor, both on some category
 $\C$ with binary coproducts $+$ and a nullary product $1$. Assume also 
 that $F$ lifts to $\oF\colon \Kl(T)\to\Kl(T)$ (Definition~\ref{def:liftFunct}). 

 A \emph{B\"{u}chi $(T,F)$-system}
is given by a tuple
$\X=
 \bigl(
  (X_{1},X_{2}),
c\colon X\kto \oF X,
  s\colon 1\kto X)$
where:
\begin{itemize}
 \item $X_{1}$ and $X_{2}$ are objects of $\C$ (with the intuition that
       $X_{1}=\{\text{non-accepting states }\nonaccstate\}$ and
       $X_{2}=\{\text{non-accepting states }\accstate\}$), and 
we define $X:=X_{1}+X_{2}$;
 \item $c\colon X \kto \oF X$ is an arrow in $\Kl(T)$ for  \emph{dynamics}; and
 \item $s\colon 1\kto X$ is an arrow in $\Kl(T)$ for \emph{initial states}.
\end{itemize}
For each $i\in
 \{1,2\}$,
we define $c_{i}\colon X_{i}\kto \oF X$ to be
the restriction $c\circ \kappa_{i}\colon X_{i}\to TFX$ of $c$ along the
 coprojection $\kappa_{i}\colon X_{i}\hookrightarrow X$.
\end{defi}
Thus a B\"uchi $(T,F)$-system is a (Kleisli) coalgebra $X\kto \oF X$
augmented with the information on accepting and initial states. 
We can regard NBTAs and PBWAs as B\"uchi $(T,F)$-systems as follows; note that an arrow $1\kto X$ in
$\Kl(\giry)$ is nothing but a probability subdistribution over $X$. 

\begin{exa}\label{exa:NBTAandPBWAInduceBuechiSystems}
  \leavevmode
\begin{enumerate}[beginpenalty=99]
 \item 
An NBTA $\mathcal{X}=(X,\Sigma,\delta,\initSet,\Acc)$ 
(Definition~\ref{def:nondetBuechiTreeAutom})
gives rise to
a B\"uchi $(\pow,F_{\Sigma})$-system 
  $\X'=
 \bigl(
  (X_{1},X_{2}),
c\colon X\kto \overline{F_\Sigma} X,
  s\colon 1\kto X)$ that is defined by:
    \begin{itemize}
  \item $X_1=\Acc$ and $X_2=X\setminus \Acc$;
  \item $c(x)=\delta(x)$; and
  \item $s(*)=I$.
  \end{itemize}


 \item 
 A PBWA $\mathcal{X}=(X,\myalphabet,M,\initVec,\Acc)$ (Definition~\ref{def:generativeProbBuechiTreeAutom}) 
 gives rise to a B\"uchi $(\giry,F_{\myalphabet})$-system 
 $\X'=
 \bigl(
  (X_{1},X_{2}),
c\colon X\kto \overline{F_\myalphabet} X,
  s\colon 1\kto X)$ that is defined by:
  \begin{itemize}
  \item $X_1=(\Acc,\pow\Acc)$ and $X_2=(X\setminus \Acc,\pow(X\setminus \Acc))$;
  \item $c(x)(\{(a,x')\})=(M(a))_{x,x'}$; and
  \item $s(*)(\{x\})=\initVec_x$.
  \end{itemize}
  Here $c$ and $s$ are well-defined as $X_1+X_2\in\Meas$ is equipped with the discrete $\sigma$-algebra.
%
\end{enumerate}
\end{exa}

The next is  the main theorem of~\cite{urabeSH16parityTrace}.\footnote{In fact this is a special case of 
the main theorem because the original theorem considers \emph{parity $(\pow,\FSigma)$-systems}, which 
generalizes B\"uchi $(\pow,\FSigma)$-systems and is identified with parity tree automata. 
Note that the B\"uchi acceptance condition is a special case of the parity acceptance condition.} 
Recall that 
  $\myTree_{\Sigma}$ is the set of  (possibly infinite)
  $\Sigma$-trees
 (Section~\ref{subsec:basicNBTA}); it 
 carries a final coalgebra
 $\zeta\colon \myTree_{\Sigma}\iso F_{\Sigma}(\myTree_{\Sigma})$ in
 $\Sets$.
 We will be using natural orders $\sqsubseteq_{X,Y}$ on the homsets $\Kl(\pow)(X,Y)$ and
$\Kl(\giry)(X,Y)$, given by inclusion and pointwise extension of the
  order on $[0,1]$, respectively. Namely,
  \begin{equation}\label{eq:deforder}
  \begin{aligned}
  f\sqsubseteq_{X,Y} g\;&\defarrow\; \forall x\in X.\; f(x)\subseteq g(x) & \quad\text{for}\quad T=\pow\quad\text{and} &\\
  f\sqsubseteq_{X,Y} g\;&\defarrow\; \forall x\in X.\;\forall A\in\sigalg_Y.\; f(x)(A)\leq g(x)(A) & \quad\text{for}\quad T=\giry\,. &
  \end{aligned}
  \end{equation}
%
\begin{thmC}[{\cite{urabeSH16parityTrace}}] \leavevmode
\label{thm:sanityCheckResult}
\begin{enumerate}
 \item 
  Let
 $\X=
 \bigl(
  (X_{1},X_{2}),
 c,
  s)$
  be a B\"uchi
 $(\pow,F_{\Sigma})$-system. 
 Consider an equational system
 	\begin{equation}\label{eq:eqSysForBuechiAcceptance}
	 \begin{aligned}
	  u_{1} 
	  &\;=_{\mu}\;
	  (J\zeta)^{-1}
	  \odot
	  \overline{F_{\Sigma}}[u_{1},u_{2}]
	  \odot
	  c_{1}\enspace,\qquad
	  &
	  u_{2} 
	  &\;=_{\nu}\;
	  (J\zeta)^{-1}
	  \odot
	  \overline{F_{\Sigma}}[u_{1},u_{2}]
	  \odot
	  c_{2}
	 \end{aligned}
	\end{equation}		
	where $u_{i}$ ranges over the homset $\Kl(\pow)(X_{i},\myTree_{\Sigma})$ for
	$i\in\{1,2\}$. Diagrammatically:
	\begin{equation}\label{eq:diagramsForEqSysForBuechiAcceptance}
	    \vcenter{\xymatrix@C+3.3em@R=.8em{
   {F_{\Sigma}X}
       \kar[r]^{\overline{F_{\Sigma}}[
       u_{1}, u_{2}
 ]}
       \ar@{}[rd]|{\color{blue}=_{\mu}}
   &
   {F_{\Sigma}(\myTree_{\Sigma})}
   \\
   {X_{1}}
       \kar[u]^{c_{1}}
       \kar[r]_(.6){
       u_{1}
       }
   &
   {\myTree_{\Sigma} \mathrlap{\enspace,}}
       \kar[u]_{J\zeta}^{\cong}
 }}
 \qquad
	    \vcenter{\xymatrix@C+3.3em@R=.8em{
   {F_{\Sigma}X}
       \kar[r]^{\overline{F_{\Sigma}}[
       u_{1}, u_{2}
       ]}
       \ar@{}[rd]|{\color{red}=_{\nu}}
   &
   {F_{\Sigma}(\myTree_{\Sigma})}
   \\
   {X_{2}}
       \kar[u]^{c_{2}}
       \kar[r]_(.6){
       u_{2}
 }
   &
   {\myTree_{\Sigma} \mathrlap{\enspace.}}
       \kar[u]_{J\zeta}^{\cong}
 }}
	\end{equation}
 \begin{enumerate}
 \item\label{item:thm:sanityCheckResult1a} The equational system has a solution, denoted by
       $\trB(c_{i})\colon X_{i}\kto \myTree_{\Sigma}$ for $i\in\{1,2\}$.
 \item Let
       \begin{math}
		 \trB(\X )
	 :=
	 \bigl(\,
	 \{*\}=1
	 \stackrel{s}{\longkto} 
	 X=X_{1}+X_{2}
	 \stackrel{[\trB(c_{1}),\trB(c_{2})]}{\longkto}
	 \myTree_{\Sigma}
	 \,\bigr)
       \end{math} be a composite in $\Kl(\pow)$. In case $\X$ is induced by an NBTA, the set
       $\trB(\X)(*)\subseteq \myTree_{\Sigma}$ coincides with the
       (B\"{u}chi) language $L(\X)$ of $\X$
       (Definition~\ref{def:BuechiLangConventionally}).
 \end{enumerate}
 \item Let $\X$ be a B\"uchi $(\giry,F_{\myalphabet})$-system, and consider the
       same equational system as~(\ref{eq:eqSysForBuechiAcceptance}),
       but with $\giry, F_{\myalphabet}, \myalphabet^{\omega}$ replacing $\pow, F_{\Sigma}, \myTree_{\Sigma}$. Then:
       \begin{enumerate}
	\item\label{item:thm:sanityCheckResult2a} The equational system has a solution.
	\item\label{item:thm:sanityCheckResult:ProbSanity}
	 Let $\X$ be induced by a PBWA
	      (Definition~\ref{def:generativeProbBuechiTreeAutom}).  For the
	      same composite $\trB(\X )\colon 1\kto \myalphabet^{\omega}$ as above
	      we have $\trB(\X)(*)=\lang(\X)\in \giry(\myalphabet^{\omega})$, the
	      B\"uchi language of the PBWA
	      (Section~\ref{sec:fwdFairSimProb}). 
	      \qed
       \end{enumerate}
\end{enumerate}
\end{thmC}

In the proof of the above theorem, (\ref{item:thm:sanityCheckResult1a}) and (\ref{item:thm:sanityCheckResult2a}) are 
proved using Proposition~\ref{prop:suffCondSol}. More concretely, 
if $T=\pow$ and $F=F_\Sigma$ then Condition~(\ref{item:prop:suffCondSol1}) of Proposition~\ref{prop:suffCondSol} is
satisfied by the equational system.
In contrast, if $T=\giry$ and $F=F_\myalphabet$ then
Condition~(\ref{item:prop:suffCondSol2}) 
is satisfied.

\section{Coalgebraic Account on Fair Simulations and Soundness
 Proofs}\label{sec:soundnessProof}
Here we lay out our coalgebraic study of fair simulations. We will be firstly led to 
a simulation notion ``with dividing'' that is coalgebraically neat
but not desirable from a practical viewpoint. Circumventing the
dividing construct  we obtain the
simulation notions that we have presented
in Sections~\ref{sec:fwdFairSimNondet}--\ref{sec:fwdFairSimProb}. 

The last part of circumventing dividing is different for $T=\pow$
and $\giry$; this is why we have different definitions of simulation. While one would
hope for uniformity, we suspect it to be hard, for the
following reason. We observed~\cite{urabeH15CALCO} that the characterization of
infinite trace
(Table~\ref{table:knownResultsOnCoalgTraceAndSimulation}(b)) is true for $T=\pow$ and $\giry$, but because of 
 different categorical machineries. Since infinite trace is a
special case of B\"uchi acceptance (where every state is accepting) and 
our soundness proof should rely on its characterization, we expect that this sharp
contrast would still stand.

\subsection{$\Cppo$-enriched Categories and Functors; Codomain Restrictions and Joins}\label{subsec:catPrelim}
In this section, we review four categorical constructs that are used in the definition of 
our categorical simulation notion.

\subsubsection{$\Cppo$-enriched category and $\Cppo$-enriched functor}\label{subsubsec:Cppoenriched}
Recall that in the categorical definition of B\"uchi languages,
we used a partial order $\sqsubseteq_{X,Y}$ on each homset $\Kl(T)(X,Y)$.
The first two notions---$\Cppo$-enriched category and $\Cppo$-enriched functor, see e.g.~\cite{borceux1994Handbook2}---add
certain assumptions to the ordered structure.
The same assumptions are also used in~\cite{hasuo07generictrace} where finite trace semantics of nondeterministic and probabilistic
systems are captured categorically.

\begin{defi}[$\Cppo$-enriched category and $\Cppo$-enriched functor]\label{def:CppoEnriched}
A category $\mathbb{C}$ is called a \emph{$\Cppo$-enriched category} if it satisfies the following conditions:

\begin{enumerate}
\item\label{item:def:CppoEnrichedCat1}
 Each homset $\mathbb{C}(X,Y)$ carries a partial order $\sqsubseteq_{X,Y}$. 
Moreover each homset $\mathbb{C}(X,Y)$ is a \emph{pointed cpo} with respect to the order, i.e.\ 
it has the least element $\bot_{X,Y}$ and each increasing sequence $f_0\sqsubseteq_{X,Y} f_1\sqsubseteq_{X,Y}\cdots\in\mathbb{C}(X,Y)$ has
the least upper bound $\bigsqcup_{i\in\omega}f_i:X\to Y$.

\item\label{item:def:CppoEnrichedCat2}
For each $X,Y,Z\in\mathbb{C}$, the composition 
$(\place\circ\place):\mathbb{C}(Y,Z)\times\mathbb{C}(X,Y)\to\mathbb{C}(X,Z)$ is monotone with respect to the product order.

\item\label{item:def:CppoEnrichedCat3}
The composition $\circ$ is $\omega$-continuous. That is, for an increasing
sequence $f_0\sqsubseteq_{X,Y} f_1$ ${\sqsubseteq_{X,Y}\cdots}$ of arrows, 
\begin{equation}\label{eq:1702101518}
\bigl(\bigsqcup_{i<\omega} f_i\bigr)\circ g=\bigsqcup_{i<\omega} \bigl(f_i\circ g\bigr)
\qquad\text{and}\qquad
h\circ \bigl(\bigsqcup_{i<\omega} f_i\bigr)=\bigsqcup_{i<\omega} \bigl(h\circ f_i\bigr)\,.
\end{equation}
\end{enumerate}
Let $\mathbb{C}$ be a $\Cppo$-enriched category.
A functor $F:\mathbb{C}\to\mathbb{C}$ is called a \emph{$\Cppo$-enriched functor} if
it satisfies the following conditions.

\begin{enumerate}
  \renewcommand{\labelenumi}{(\alph{enumi})}
  \renewcommand{\theenumi}{\alph{enumi}}
\item\label{item:def:CppoEnrichedFunct1}
It is \emph{locally monotone}, that is, for each $X,Y\in\mathbb{C}$ and $f,g:X\to Y$, 
$f\sqsubseteq_{X,Y} g$ implies $Ff\sqsubseteq_{FX,FY} Fg$.

\item\label{item:def:CppoEnrichedFunct2}
It is \emph{locally $\omega$-continuous}, that is, for each $X,Y\in\mathbb{C}$ and 
increasing sequence $f_0\sqsubseteq_{X,Y} f_1\sqsubseteq_{X,Y}\cdots\in\mathbb{C}(X,Y)$, 
we have $F\bigl(\bigsqcup_{i<\omega}f_i\bigr)=\bigsqcup_{i<\omega}(F f_i)$.
\end{enumerate}
\end{defi}

If confusion is unlikely,
we omit subscripts and just write $\sqsubseteq$ and $\bot$ for $\sqsubseteq_{X,Y}$ and $\bot_{X,Y}$.

\begin{rem}\label{rem:monotonicityExplicitly}
In a definition of $\Cppo$-enriched category, monotonicity of compositions (Condition~(\ref{item:def:CppoEnrichedCat2}) in the definition above)
 is often omitted (see~\cite{hasuo07generictrace,bonchiMSZ14killEpsilons} for example).
However we require it explicitly to ensure that if $(f_i)_{i\in\omega}$ is an increasing sequence then
$(f_i\circ g)_{i\in\omega}$ and $(h\circ f_i)_{i\in\omega}$ are also increasing sequences
and hence the suprema  $\bigsqcup_{i<\omega} \bigl(f_i\circ g\bigr)$ and $\bigsqcup_{i<\omega} \bigl(h\circ f_i\bigr)$ 
in (\ref{eq:1702101518}) in the definition is well-defined.
We require a $\Cppo$-enriched functor to be locally monotone (Condition~(\ref{item:def:CppoEnrichedFunct1}) in the definition) 
for the same reason.
\end{rem}

\begin{rem}\label{rem:ContNotNecessary}
The notions of $\Cppo$-enriched category and $\Cppo$-enriched functor are instances of
well-known categorical notions of $\mathbb{V}$-enriched category and $\mathbb{V}$-enriched functor 
 (see~\cite{borceux1994Handbook2} for example).
%
Later in the soundness proof of our categorical  fair simulation, 
 we will be assuming $\Kl(T)$ and $\oF$ to be $\Cppo$-enriched.
We do so for conceptual simplicity: technically speaking this is stronger than needed, since $\omega$-continuity of composition is not used in the proof. 

\end{rem}
It is not so hard to see that $\Kl(\pow)$ are $\Kl(\giry)$ are both $\Cppo$-enriched categories 
and
$\overline{F_\Sigma}$ and $\overline{F_\myalphabet}$ are both $\Cppo$-enriched functors,
with respect to the
orders in (\ref{eq:deforder}).


\begin{rem}\label{rem:buechiLangCppo}
Let $T$ be a monad and $F$ be a functor with a final coalgebra $\zeta:Z\to FZ$.
If $\Kl(T)$ and $\overline{F}$ are $\Cppo$-enriched then 
the equational system~(\ref{eq:diagramsForEqSysForBuechiAcceptance}) 
(with $T$, $F$ and $Z$ replacing $\pow$, $F_\Sigma$ and $\myTree_\Sigma$) satisfies
the assumptions (\ref{item:thm:correctnessOfProgMeasEqSys1}) and (\ref{item:thm:correctnessOfProgMeasEqSys2})
in Theorem~\ref{thm:correctnessOfProgMeasEqSys}.
Therefore by Theorem~\ref{thm:correctnessOfProgMeasEqSys}, 
progress measures for the equational system satisfy soundness and completeness.

%
There is another way to ensure soundness and completeness of progress measures for the equational system.
In the original version of Theorem~\ref{thm:correctnessOfProgMeasEqSys} in~\cite{HasuoSC16},
instead of the assumptions
(\ref{item:thm:correctnessOfProgMeasEqSys1}) and (\ref{item:thm:correctnessOfProgMeasEqSys2}), 
the following assumption is required: 
\begin{itemize}
\item for each $i\in[1,m]$, the poset $L_i$ is a complete lattice.
\end{itemize}
This implies that 
soundness and completeness of progress measures for the equational system~(\ref{eq:diagramsForEqSysForBuechiAcceptance}) 
are satisfied if we have the following condition:
\begin{enumerate}
\setcounter{enumi}{1}
\item[($\ddagger$)]
each homset of $\Kl(T)$ carries a complete lattice;
 Kleisli compositions are monotone; and $\overline{F}$ is locally monotone.
\end{enumerate}
However, as we have mentioned in Remark~\ref{rem:diffFromHSC16_1}, a homset of $\Kl(\giry)$ does not necessarily carries a complete lattice, and hence $T=\giry$ does not satisfy
($\ddagger$) above.
\end{rem}

\begin{exa}\label{example:giryNotDCpo}
A homset of $\Kl(\giry)$ is not necessarily a dcpo; here is a counterexample.

Let $\sigalg_{[0,1]}$ be the $\sigma$-algebra over the unit interval $[0,1]$ consisting of the Borel sets (see~\cite{ashD2000probability} for example). 
It is known that there exists $V\subseteq [0,1]$ such that $V\notin\sigalg_{[0,1]}$
(see~\cite{herrlich06axiomofchoice} for example).
It is easy to see that 
$\giry(1,\pow 1)\cong([0,1],\sigalg_{[0,1]})$.

We define a set $\mathfrak{A}_V\subseteq \Kl(\giry)\bigl(([0,1],\sigalg_{[0,1]}),(1,\pow 1)\bigr)$ of 
Kleisli arrows by 
\[
\mathfrak{A}_V\;:=\;
\bigl\{\chi_X:([0,1],\sigalg_{[0,1]})\to\giry(1,\pow 1)
\,\mid\,
X\in \sigalg_{[0,1]}\;\text{and}\; X\subseteq V \bigr\}\,.
\]
%
Here $\chi_X$ denotes 
the characteristic function of $X$, that is,
$\chi_X(x)=1$ if $x\in X$ and $\chi_X(x)=0$ otherwise.
Note that a Kleisli arrow $\chi_X:([0,1],\sigalg_{[0,1]})\to\giry(1,\pow 1)$ 
is 
a measurable function $\chi_X:([0,1],\sigalg_{[0,1]})\to([0,1],\sigalg_{[0,1]})$\,.
It is easy to see that $\chi_X\sqsubseteq \chi_{X'}$ if and only if $X\subseteq X'$.

Assume that $\Kl(\giry)(([0,1],\sigalg_{[0,1]}),(1,\pow 1))$ is a dcpo.
Then there exists the least upper bound 
$\bigsqcup_{X\subseteq V}\chi_X:([0,1],\sigalg_{[0,1]})\to\giry(1,\pow 1)$ of $\mathfrak{A}_V$.

Let $V':=(\bigsqcup_{X\subseteq V}\chi_X)^{-1}(\{1\})$.
Then as $\{1\}\in\sigalg_{[0,1]}$ and $\bigsqcup_{X\subseteq V}\chi_X$ is a measurable function,
we have $V'\in\sigalg_{[0,1]}$. 
Moreover as $\bigsqcup_{X\subseteq V}\chi_X$ is an upper bound of $\mathfrak{A}_V$,
we have $V\subseteq V'$.
Therefore by $V\notin \sigalg_{[0,1]}$, there exists $v\in V'$ such that $V\subseteq V'\setminus \{v\}$.
As $\{v\},V'\in\sigalg_{[0,1]}$, we have $V'\setminus \{v\}\in\sigalg_{[0,1]}$.
It is easy to see that $\chi_{V'\setminus\{v\}}$ is an upper bound of $\mathfrak{A}_V$.
This contradicts the fact $\bigsqcup_{X\subseteq V}\chi_X$ is the least upper bound of $\mathfrak{A}_V$.
Hence $\Kl(\giry)\bigl(([0,1],\sigalg_{[0,1]}),(1,\pow 1)\bigr)$ is not a dcpo.
%
\end{exa}

\subsubsection{Codomain Restriction and Codomain Join}\label{subsubsec:codomRJ}
The other two notions we describe in Section~\ref{subsec:catPrelim} are codomain restriction and codomain join of Kleisli arrows in $\Kl(T)$. 
The latter combines two arrows $f_1:X\kto Y_1$ and $f_2:X\kto Y_2$ into $f:X\kto Y_1+Y_2$ while
the former does the inverse. We note that the former operation is not necessarily total (see Example~\ref{exa:codomJR}).
It is known that if the monad $T$ satisfies a certain condition, then its Kleisli category comes with the two operations.

\begin{defC}[\cite{cirstea13fromBranching,jacobs10trace}]\label{def:PAM}
Let $\mathbb{C}$ be a category with an initial object $0$, binary products and binary coproducts.
A monad $T$ on $\mathbb{C}$ is called a \emph{partially additive monad} if it satisfies the following conditions:
\begin{enumerate}
\item\label{item:def:PAM1}
The object $T0$ is a final object in $\mathbb{C}$.\footnote{This implies that $0$ is both an initial and final object in $\Kl(T)$. Such $0$ is called a \emph{zero object}.}

\item\label{item:def:PAM2}
Let $X_1,X_2\in\mathbb{C}$. We define $p_1:X_1+X_2\to TX_1$ and $p_2:X_1+X_2\to TX_2$ 
by $p_1:=[\eta_{X_1},\bot_{X_2,X_1}]$ and 
$p_2:=[\bot_{X_1,X_2},\eta_{X_2}]$.
Here for each $X,Y\in\mathbb{C}$, 
$\bot_{X,Y}:X\to TY$ is given by $\bot_{X,Y}:=X\xrightarrow{\finalarrow_X} T0\xrightarrow{T\initarrow_{Y}} TY$ where
$\finalarrow_X:X\to T0$ and $\initarrow_T:0\to Y$ denote the unique arrows
(see also Remark~\ref{rem:botNoConfusion}). 

We require the following arrow be a monomorphism.
\[
T(X_1+X_2) \xrightarrow{\langle\mu_{X_1}\circ Tp_1,\mu_{X_2}\circ Tp_2 \rangle} TX_1\times T X_2
\]
\end{enumerate}
\end{defC}

\begin{defi}[codomain restriction and codomain join, \cite{cirstea13fromBranching,jacobs10trace}]\label{def:codomJoinRestr}
Let $T$ be a 
partially  additive monad.
Then $\Kl(T)$ comes with two operations on arrows called \emph{codomain restriction} and \emph{codomain join}.
Codomain restriction takes an arrow $g\colon V\to X_1+X_2$ and returns $\codomRestr{g}{X_{i}}\colon V\to X_{i}$ for $i\in\{1,2\}$. 
Here $\codomRestr{g}{X_{i}}$ is defined by
\[
\codomRestr{g}{X_{i}}\colon
V\xrightarrow{g} T(X_1+X_2)\xrightarrow{\langle\mu_{X_1}\circ Tp_1,\mu_{X_2}\circ Tp_2 \rangle} TX_1\times TX_2\xrightarrow{\pi_i}TX_i\,.
\]
Codomain join is a \emph{partial} operation that takes a pair
$g_1:V\to X_1$ and $g_2:V\to X_2$, and returns a (necessarily unique) arrow $\codomJoin{g_1,g_2}\colon V\to X_1+X_2$ such that
\[
\langle g_1,g_2\rangle = \langle\mu_{X_1}\circ Tp_1,\mu_{X_2}\circ Tp_2 \rangle\circ \codomJoin{g_1,g_2}\,.
\] 
The situation is illustrated below.
\[
 \vcenter{\xymatrix@C=6.4em@R=2.7em{
  {T(X_1+X_2)}
     \ar@{>->}[r]^-{\langle\mu_{X_1}\circ Tp_1,\mu_{X_2}\circ Tp_2 \rangle}
  &
  {TX_1\times TX_2}
  \\
 &
  {V}
    \ar[u]_{\langle g_1,g_2\rangle}
        \ar[ul]^{\codomJoin{g_1,g_2}}
}}
\]
\end{defi}

These operations may look
unfamiliar, but $T=\pow,\giry$ are partially additive monads, and
codomain joins and operations are given
by suitably
restricting/joining subsets/distributions. 

\begin{exa}\label{exa:codomJR}
For $T=\pow$, we can define codomain restrictions and codomain joins 
by
\[
\codomRestr{g}{X_{i}}(v)=\{x\in X_i\mid x\in g(v)\}
\qquad\text{and}\qquad
\codomJoin{g_1,g_2}(v)=g_1(v)\cup g_2(v)\,.
\]

For $T=\giry$, the definitions are as follows.
\begin{align*}
\codomRestr{g}{X_{i}}(v)(A)&=g(v)(A),
\qquad \text{and} \\
\codomJoin{g_1,g_2}(v)(A)&=\begin{cases}
g_1(v)(A\cap X_1)+g_2(v)(A\cap X_2) & (g_1(v)(A\cap X_1)+g_2(v)(A\cap X_2)\leq 1) \\
\text{undefined} & (\text{otherwise})\,.
\end{cases}
\end{align*}
Note that codomain join for $T=\giry$ is a partial operation in the sense that
it is not always defined.
\end{exa}

\begin{rem}\label{rem:botNoConfusion}
Let $T$ be a partially additive monad such that $\Kl(T)$ is a $\Cppo$-enriched category.
We have used the same symbol $\bot_{X,Y}$ for (1) the least element in a homset of a $\Cppo$-enriched category 
(Definition~\ref{def:CppoEnriched})
and (2) the arrow $T\initarrow_Y\circ \finalarrow_X$ in the Kleisli category of a partially additive monad.
In the next section we assume that composition in $\Kl(T)$ is left-strict, i.e.\ $\bot_{X,Y}\odot g=\bot_{Z,Y}$ for each $g:Z\kto Y$, where $\bot_{X,Y}$ and $\bot_{Z,Y}$ refer to the least arrows
(see Theorem~\ref{thm:soundnessFwdFairBuechi}).
It is easy to see that under this assumption, the arrow $T\initarrow_Y\circ \finalarrow_X$ coincides with the least element in
$\Kl(T)(X,Y)$. This justifies  overriding  the symbol $\bot$. 
\end{rem}

We conclude this section with some properties of codomain restrictions and joins.
They are proved by easy diagram-chasing.

\begin{lem}\label{lem:propertyCodomRJ}
Let $T$ be a partially additive monad.
We have the following:
\begin{enumerate}
\item\label{item:def:codomJoinRestr1}
Codomain restrictions and codomain joins are partially mutually inverse, in the following sense. 
Given $g\colon V\to X_1+X_2$, the
codomain join $\codomJoin{\codomRestr{g}{X_{1}},\codomRestr{g}{X_{2}}}$
	    is always \linebreak
	    defined and equal to $g$. Conversely, provided
	    that $\codomJoin{g_1,g_2}$ is defined,
	    we have
	    $\codomRestr{(\codomJoin{g_1,g_2})}{X_{i}}$ $=g_{i}$ for $i\in\{1,2\}$.

\item\label{item:lem:propertyCodomRJ1}
For  $f:W\to V$, $g_1:V\to X_1$, $g_2:V\to X_2$, $h_1:X_1\to Y_1$ and $h_2:X_2\to Y_2$ such that
$\codomJoin{g_1,g_2}$ is defined, we have:
\[
\codomJoin{g_1,g_2}\odot f=\codomJoin{g_1\circ f,g_2\circ f}\quad\text{and}\quad
(h_1+h_2)\circ\codomJoin{g_1,g_2}=\codomJoin{h_1\circ g_1,h_2\circ g_2}\,.
\]

\item\label{item:lem:propertyCodomRJ2}
For $g:V\to X$, $\codomJoin{g,\bot_{V,X}}$ and $\codomJoin{\bot_{V,X},g}$ are always defined and we have
\[
[\id_X,\id_X]\odot\codomJoin{g,\bot_{V,X}}=[\id_X,\id_X]\odot\codomJoin{\bot_{V,X},g}=g\,.
\tag*{\qEd}
\]
\end{enumerate}
\end{lem}


\subsection{Coalgebraic Fair Simulation with Dividing}\label{subsec:coalgFairSimWithDividing}
%
We make the following requirements in this section
so that our definitions will make sense.
\begin{asm}
\label{asm:wellDefinedSimulation}
In this section
we
assume the following conditions on $T$ and $F$ on $\C$. 
\begin{enumerate}\setcounter{enumi}{\value{asmenumi}}


\item\label{item:asm:wellDefinedTrace2}
The functor $F$ has a final coalgebra $\zeta:Z\iso FZ$ in
      $\mathbb{C}$.

\item\label{item:asm:wellDefinedTrace3}
The functor $F:\mathbb{C}\to\mathbb{C}$ lifts to
$\overline{F}:\Kl(T)\to\Kl(T)$ (see Definition~\ref{def:liftFunct}).

      
\item\label{item:asm:wellDefinedTrace4}
The Kleisli category $\Kl(T)$ and the lifting $\overline{F}:\Kl(T)\to\Kl(T)$ of $F$ are both \linebreak
$\Cppo$-enriched (Definition~\ref{def:CppoEnriched}).



%
%

\item\label{item:asm:wellDefinedSimulation1}
The monad $T$ is a partially additive monad.
Moreover 
the codomain join is downward closed.
That is,
for $f_{i}, g_{i}\colon V\kto X_{i}$  such that
       $f_{i}\sqsubseteq g_{i}$ for each $i\in I$, if
 $\codomJoin{g_{i}}_{i\in I}
	    $ is defined, then 
       so is $\codomJoin{f_{i}}_{i\in I}
	    $. 

%
     
\item\label{item:asmthm:soundnessFwdFairBuechi2}
Codomain restriction $\codomRestr{(\place)}{X_i}$, codomain join $\codomJoin{\place,\place}$ 
and cotupling $[\place,\place]$ of Kleisli arrows are all monotone 
with respect to 
the order 
$\sqsubseteq$.
     
\setcounter{asmenumi}{\value{enumi}}
\end{enumerate}
\end{asm}
%
%
%
%
Note that by Definition~\ref{def:codomJoinRestr},
 Condition~(\ref{item:asm:wellDefinedSimulation1}) implies that $\Kl(T)$ comes with codomain restrictions and joins.

With the help of codomain restrictions/joins we define a
categorical fair simulation.

\begin{defi}[(forward) fair simulation with dividing] 
\label{def:fwdFairBuechiSimWithDiv}
Let $T$ and $F$ be subject to
Assumption~\ref{asm:wellDefinedSimulation};
\begin{math}
  \X= 
  \bigl(
   (X_{1},X_{2}),
   c,
   s
\bigr)
\end{math}
and
\begin{math}
  \mathcal{Y}=
  \bigl(
   (Y_{1},Y_{2}),
    d,
    t
 \bigr)
 \end{math} 
be B\"uchi $(T,F)$-systems; and 
 $\overline{\alpha}$ be an ordinal.
A \emph{(forward, $\overline{\alpha}$-bounded) fair simulation with dividing} from
 $\X$ to $\mathcal{Y}$ is 
 an arrow $f:Y\kto X$ 
in $\Kl(T)$
 subject to the following conditions. 
 Below, for simplicity,  
 a domain-and-codomain restriction
 $\codomRestr{(f\odot \kappa_j)}{X_{i}}\colon Y_{j}\kto X_{i}$
 (Definition~\ref{def:codomJoinRestr}) shall be denoted by $f_{ji}$;
 and we refer to $f_{11},f_{12},f_{21},f_{22}$ as \emph{components}
 of a fair simulation $f$. 
 \begin{enumerate}

  \item\label{item:f21Andf22}
  The arrow 
  $f:Y\kto X$ 
  is a forward  simulation from $\mathcal{X}$ to $\mathcal{Y}$
  in the sense of~\cite{hasuo06genericforward} (see also
       Table~\ref{table:knownResultsOnCoalgTraceAndSimulation}(c)). 
       That is: $c\odot f\sqsubseteq \oF f\odot d$ and
       $s\sqsubseteq f\odot t$.

  \item \label{item:f11Andf12} 
  The components $f_{11}\colon Y_{1}\kto
	X_{1}$ and $f_{12}\colon Y_{1}\kto X_{2}$ come
	with a
	 \emph{dividing} $d_{11},d_{12}$ of the component $d_1\colon Y_{1}\kto \overline{F}Y$ of $d$,  and 
	 \emph{approximation sequences}.
	The former is a pair $d_{11},d_{12}:Y_1\kto \overline{F}Y$ 
	such that
	$[\id_{\overline{F}Y},\id_{\overline{F}Y}]\odot\codomJoin{d_{11},d_{12}}=d_1$.
    The latter are (possibly
	transfinite) increasing sequences 
	of length $\overline{\alpha}$:
       \[
\begin{array}{ll}
         f_{11}^{(0)} \sqsubseteq
        f_{11}^{(1)} \sqsubseteq
        \cdots
	\sqsubseteq
	f_{11}^{(\overline{\alpha})} 
	\,\colon
	Y_{1}\kto X_{1},
	\quad\text{and}\quad
        &
        f_{12}^{(0)} \sqsubseteq
        f_{12}^{(1)} \sqsubseteq
        \cdots 
	\sqsubseteq
	f_{12}^{(\overline{\alpha})} 
	\,\colon
	Y_{1}\kto X_{2},
	\quad\text{such that}
\end{array}       
\]
\begin{enumerate}

 \item\label{item:f11Andf12_2}
 \textbf{(Approximate $f_{11}$ and $f_{12}$)} 
       We have
       $f_{11}^{(\overline{\alpha})}=f_{11}$ and
       $f_{12}^{(\overline{\alpha})}=f_{12}$.
 \item\label{item:f11Andf12_3}
 \textbf{($f_{11}^{(\alpha)}$)} 
       For each ordinal $\alpha$ such that $\alpha\le\overline{\alpha}$,
       the inequality~(\ref{eq:f11PostFixedPoint}) below holds.
       Note that the required
       codomain joins 
	 do exist.

 \item\label{item:f11Andf12_4}
 \textbf{($f_{12}^{(\alpha)}$, the base case)} 
 For the $0$-th approximant, we have $f_{12}^{(0)}=\bot$.

 \item\label{item:f11Andf12_5}
 \textbf{($f_{12}^{(\alpha)}$, the step case)} 
        For each ordinal $\alpha$ such that $\alpha<\overline{\alpha}$,
      the inequality~(\ref{eq:item:f11Andf12_3}) holds.

 \item \label{item:f11Andf12_6}
 \textbf{($f_{12}^{(\alpha)}$, the limit case)} 
 If $\alpha$ is a limit ordinal, then 
 the supremum $\bigsqcup_{\alpha'<\alpha}f_{12}^{(\alpha')}$ exists and 
       \begin{math}
	f_{12}^{(\alpha)} \sqsubseteq \bigsqcup_{\alpha'<\alpha} f_{12}^{(\alpha')}
       \end{math}. 
\end{enumerate}
 \end{enumerate}
\noindent
 \begin{minipage}{0.45\hsize}
    \begin{equation}\label{eq:f11PostFixedPoint}
  \small
   \vcenter{\xymatrix@C+8em@R=.8em{
   {FY}
       \kar[r]^{\overline{F}\bigl[\,
         \codomJoin{f_{11}^{(\alpha)},f_{12}^{(\alpha)}},\,
	 \codomJoin{f_{21},f_{22}}
        \,\bigr]}
       \ar@{}[rd]|{\sqsupseteq}
   &
   {FX}
   \\
   {Y_{1}}
       \kar[u]^{d_{11}}
       \kar[r]_{f_{11}^{(\alpha)}}
   &
   {X_{1} }
       \kar[u]_{c_{1}}
 }}
  \end{equation}
 \end{minipage}
\qquad
 \begin{minipage}{0.45\hsize}
  	  \begin{equation}\label{eq:item:f11Andf12_3}
	   \small
	   \vcenter{\xymatrix@C+8em@R=.8em{
	   {FY}
	   \kar[r]^{\overline{F}\bigl[\,
	   \codomJoin{f_{11}^{(\alpha)},f_{12}^{(\alpha)}},\,
	   \codomJoin{f_{21},f_{22}}
	   \,\bigr]}
	       \ar@{}[rd]|{\sqsupseteq}
	   &
	   {FX}
	   \\
	   {Y_{1}}
	       \kar[u]^{d_{12}}
	       \kar[r]_{f_{12}^{(\alpha+1)}}
	   &
	   {X_{2}}
	       \kar[u]_{c_{2}}
	}}
\end{equation}
 \end{minipage}
 
\end{defi}
In the definition above, note the direction:
a simulation from $\mathcal{X}$ to $\mathcal{Y}$ has type $Y\kto X$.

\begin{thm}[soundness]
\label{thm:soundnessFwdFairBuechi}
Let $\overline{\alpha}$ be an ordinal.
Assume
Assumption~\ref{asm:wellDefinedSimulation} and
the following.
\begin{enumerate}\setcounter{enumi}{\value{asmenumi}}
\item\label{item:asmthm:soundnessFwdFairBuechi1}
For an arbitrary B\"uchi $(T,F)$-system $\X$, 
the equational system~(\ref{eq:eqSysForBuechiAcceptance}), with $F,
     Z$ replacing $F_{\Sigma}, \myTree_{\Sigma}$, has a
     (necessarily unique) 
     solution $\trB(c_{1}), \trB(c_{2})$. 




\item\label{item:asmthm:soundnessFwdFairBuechi3-2}
The Kleisli composition $\odot$ is 
both \emph{left} and \emph{right-strict}: $\bot\odot f=\bot$ and $f\odot \bot=\bot$.



\item\label{item:asmthm:soundnessFwdFairBuechi5}
For each limit ordinal $\alpha\leq\overline{\alpha}$,
post-composition in $\Kl(T)$ is $\alpha$-continuous, i.e.\ if the supremum $\bigsqcup_{i<\alpha} f_i$ exists then
$\bigsqcup_{i<\omega}(g\odot f_i)$ also exists and 
$g\odot (\bigsqcup_{i<\alpha} f_i)=\bigsqcup_{i<\omega}(g\odot f_i)$.

\setcounter{asmenumi}{\value{enumi}}
\end{enumerate}
Then a fair $\overline{\alpha}$-bounded simulation 
with dividing,
 from one B\"uchi $(T,F)$-system $\X$ to another $\mathcal{Y}$,
witnesses
trace inclusion $\trB(\X )\sqsubseteq
 \trB(\mathcal{Y})\colon 1\to TZ$.
\end{thm}

This theorem follows immediate from the following lemma.

\begin{lem}\label{lem:soundnessFwdFairBuechiSim}
Assume Assumption~\ref{asm:wellDefinedSimulation} and the assumptions (\ref{item:asmthm:soundnessFwdFairBuechi1})--(\ref{item:asmthm:soundnessFwdFairBuechi5}) in Thm.~\ref{thm:soundnessFwdFairBuechi}.
Let $f:Y\kto X$ be a forward, $\overline{\alpha}$-bounded simulation with dividing from 
\begin{math}
  \X= 
  \bigl(
   (X_{1},X_{2}),
   c,
   s
\bigr)
\end{math}
and
\begin{math}
  \mathcal{Y}=
  \bigl(
   (Y_{1},Y_{2}),
    d,
    t
 \bigr)\,
 \end{math} 
 (Definition~\ref{def:fwdFairBuechiSimWithDiv}).
We define arrows $\trB(c_1):X_1\kto Z$, $\trB(c_2):X_2\kto Z$,
$\trB(d_1):Y_1\kto Z$ and $\trB(d_2):Y_2\kto Z$ as in 
Theorem~\ref{thm:sanityCheckResult}.
Then we have: 
\[
[\trB(c_1),\trB(c_2)]\odot\codomJoin{f_{11},f_{12}}
\sqsubseteq
\trB(d_1),
\quad\text{and}\quad
[\trB(c_1),\trB(c_2)]\odot\codomJoin{f_{21},f_{22}}
\sqsubseteq
\trB(d_2).
\]
\end{lem}

To prove this lemma we need two sublemmas.

\begin{sublem}
\label{sublem:soundnessFwdFairBuechiSim1}
We assume that $T$ and $F$ satisfy
Assumption~\ref{asm:wellDefinedSimulation} and
the assumptions 
in Theorem~\ref{thm:soundnessFwdFairBuechi}.
We further assume the situation in Definition~\ref{def:fwdFairBuechiSimWithDiv}.
Let $d_{11},d_{12}:X_1\kto \overline{F}X$ be the dividing of $d_1:X_1\kto \overline{F}X$.
Recall that 
$\trB(c_1):X_1\kto Z$ and $\trB(c_2):X_2\kto Z$
are given by the solutions $u_1^\sol$ and $u_2^\sol$ of the following equational system.
(Theorem~\ref{thm:sanityCheckResult}).
	\begin{equation}\label{eq:1702080948}
	 \begin{aligned}
	  u_{1} 
	  &\;=_{\mu}\;
	  (J\zeta)^{-1}
	  \odot
	  \oF[u_{1},u_{2}]
	  \odot
	  c_{1}
	  &\in\Kl(T)(X_1,Z)
	  \\
	  u_{2} 
	  &\;=_{\nu}\;
	  (J\zeta)^{-1}
	  \odot
	  \oF[u_{1},u_{2}]
	  \odot
	  c_{2}
	  	  &\in\Kl(T)(X_2,Z)
	 \end{aligned}
	\end{equation}	
%
By completeness of progress measure (Theorem~\ref{thm:correctnessOfProgMeasEqSys}.\ref{item:completenessProgressMeas}),
there exists a progress measure
\begin{equation*}
 p_{\mathcal{X}}=((\overline{\beta_1}),(u_{1}(\beta_1):X_1\kto Z,u_{2}(\beta_1):X_2\kto Z)_{\beta_1\leq\overline{\beta_1}})
\end{equation*}
for (\ref{eq:1702080948}) 
such that $\overline{\beta_1}\leq\omega$, $u_1(\overline{\beta_1})=\trB(c_1)$ and $u_2=\trB(c_2)$.
We define two ordinals $\overline{\gamma_1}$ and 
$\overline{\gamma_2}$ by $\overline{\gamma_1}=\overline{\beta_1}$ and
$\overline{\gamma_2}=\overline{\alpha}$.
Moreover for each pair of ordinals $\gamma_1\leq \overline{\gamma_1}$ 
and $\gamma_2\leq \overline{\gamma_2}$,
we define three arrows 
$h_1(\gamma_1,\gamma_2):Y_1\kto Z$, 
$h_2(\gamma_1,\gamma_2):Y_1\kto Z$, and
$h_3(\gamma_1,\gamma_2):Y_2\kto Z$ by:
\begin{multline*}
h_1(\gamma_1,\gamma_2)= u_{1}(\gamma_1)\odot f^{(\gamma_2)}_{11},  \quad
h_2(\gamma_1,\gamma_2)= u_{2}(\overline{\gamma_1})\odot f^{(\gamma_2)}_{12}, \\
 \text{and}\quad
h_3(\gamma_1,\gamma_2)= [u_{1}(\overline{\gamma_1}),u_{2}(\overline{\gamma_1})]\odot \codomJoin{f_{21},f_{22}}\,.
\end{multline*}
%

We claim that,
if we let
\begin{equation}\label{eq:probMeasSoundnessFwdFairBuechiSim1}
p:=\bigl((\overline{\gamma_1},\overline{\gamma_2}),
(h_1(\gamma_1,\gamma_2),h_2(\gamma_1,\gamma_2),h_3(\gamma_1,\gamma_2))_{\gamma_1\leq\overline{\gamma_1},\gamma_2\leq\overline{\gamma_2}}\bigr)\,,
\end{equation}
then it is a progress measure for the following equational system.
\begin{equation}
\label{eq:eqSysSublemsoundnessFwdFairBuechiSim10}
 \begin{aligned}
	  h_{1} 
	  &\;=_{\mu}\;
	  (J\zeta)^{-1}
	  \odot
	  \oF\bigl[[\id_{Z},\id_{Z}]\odot\codomJoin{h_{1},h_{2}},h_{3}\bigr]
	  \odot
	  d_{11}
	  &\in\Kl(T)(Y_1,Z)
	  \\ 
	  h_{2} 
	  &\;=_{\mu}\;
	  (J\zeta)^{-1}
	  \odot
	  \oF\bigl[[\id_{Z},\id_{Z}]\odot\codomJoin{h_{1},h_{2}},h_{3}\bigr]
	  \odot
	  d_{12}
	  &	  \in\Kl(T)(Y_1,Z)
	  \\ 
	  h_{3} 
	  &\;=_{\nu}\;
	  (J\zeta)^{-1}
	  \odot
	  \oF\bigl[[\id_{Z},\id_{Z}]\odot\codomJoin{h_{1},h_{2}},h_{3}\bigr]
	  \odot
	  d_{2}
	  &	  \in\Kl(T)(Y_2,Z)
 \end{aligned}
\end{equation}
 Note here that if $h_1,h_2\in\Kl(T)(Y_1,Z)$ satisfy the equations above,
 then their codomain join $\codomJoin{h_1,h_2}$ is always defined.
\end{sublem}

\begin{figure}
  \[
\underbrace{\!\!\!\!\!\!\xymatrix@C+2em@R=1.4em{
   {FY}
       \kar[rr]^{\overline{F}\bigl[
         \codomJoin{f_{11}^{(\gamma_2)},f_{12}^{(\gamma_2)}},
	 \codomJoin{f_{21},f_{22}}
        \bigr]}
       \ar@{}[rrd]|{\sqsupseteq}
   & &
   {FX}
       \kar[r]^{\text{\raisebox{2mm}{$
       \overline{F}\bigl[u_{1}(\gamma_1),u_{2}(\gamma_1)
        \bigr]$}}}
       \ar@{}[rd]|{\sqsupseteq}
   &
   {FZ}
   \\
   {Y_1}
       \kar[u]_{d_{11}}
       \kar[rr]_{f_{11}^{(\gamma_2)}}
   & &
   {X_1}
       \kar[u]_{c_{1}}
       \kar[r]_{u_{1}(\gamma_1+1)}
   &
   {Z}
       \kar[u]^{J\zeta}_{\cong}
 }\!\!\!\!\!\!}_{h_1(\gamma_1,\gamma_2)}
    \qquad \underbrace{\!\!\!\!\!\!\xymatrix@C+2em@R=1.4em{
   {FY}
       \kar[rr]^{\overline{F}\bigl[
         \codomJoin{f_{11}^{(\gamma_2)},f_{12}^{(\gamma_2)}},
	 \codomJoin{f_{21},f_{22}}
        \bigr]}
       \ar@{}[rrd]|{\sqsupseteq}
   & &
   {FX}
       \kar[r]^{\text{\raisebox{2mm}{$
       \overline{F}\bigl[u_{1}(\overline{\gamma_1}),u_{2}(\overline{\gamma_1})
        \bigr]$}}}
       \ar@{}[rd]|{\sqsupseteq}
   &
   {FZ}
   \\
   {Y_1}
       \kar[u]_{d_{12}}
       \kar[rr]_{f_{12}^{(\gamma_2+1)}}
   & &
   {X_2}
       \kar[u]_{c_{2}}
       \kar[r]_{u_{2}(\overline{\gamma_1})}
   &
   {Z}
       \kar[u]^{J\zeta}_{\cong}
 }\!\!\!\!\!\!}_{h_2(\gamma_1,\gamma_2)}
\]
\\
\[
    \underbrace{\!\!\!\!\!\!\xymatrix@C+3.3em@R=1.4em{
   {FY}
       \kar[rr]^{\overline{F}\bigl[
         \codomJoin{f_{11}^{(\overline{\gamma_2})},f_{12}^{(\overline{\gamma_2})}},
	 \codomJoin{f_{21},f_{22}}
        \bigr]}
       \ar@{}[rrd]|{\sqsupseteq}
   & &
   {FX}
          \kar[r]^{\overline{F}\bigl[u_{1}(\overline{\gamma_1}),u_{2}(\overline{\gamma_1})
        \bigr]}
       \ar@{}[rd]|{\sqsupseteq}
   &
   {FZ}
   \\
   {Y_2}
       \kar[u]_{d_{2}}
       \kar[rr]_{\codomJoin{f_{21},f_{22}}}
   & &
   {X}
       \kar[u]_{c}
       \kar[r]_{[u_{1}(\overline{\gamma_1}),u_{2}(\overline{\gamma_1})]}
   &
   {Z}
       \kar[u]^{J\zeta}_{\cong}
 }\!\!\!\!\!\!}_{h_3(\gamma_1,\gamma_2)}
\]
\caption{The progress measure $p$ in (\ref{eq:probMeasSoundnessFwdFairBuechiSim1}) as diagrams.}
\label{fig:probMeasSoundnessFwdFairBuechiSim1}
\end{figure}

\begin{proof}
We check that
$p$ in (\ref{eq:probMeasSoundnessFwdFairBuechiSim1})
satisfies the axioms of progress measure (Definition~\ref{def:progMeas}). 
See also Figure~\ref{fig:probMeasSoundnessFwdFairBuechiSim1}.
\begin{enumerate}
\item{\bf(Monotonicity)}
We assume $\gamma_1\leq\gamma'_1$ and $\gamma_2\leq\gamma'_2$. 
Then 
 by 
 Assumption~\ref{asm:wellDefinedSimulation}.\ref{item:asm:wellDefinedTrace4} and that 
 $(f_{11}^{(\alpha)})_{\alpha\leq\overline{\alpha}}$ and $(f_{12}^{(\alpha)})_{\alpha\leq\overline{\alpha}}$
 are increasing sequence, we have:
\[
h_1(\gamma_1,\gamma_2)
\,=\,u_{1}(\gamma_1)\odot f^{(\gamma_2)}_{11}
\,\sqsubseteq\, u_{1}(\gamma'_1)\odot f^{(\gamma'_2)}_{11}
\,=\,h_1(\gamma'_1,\gamma'_2)\,.
\]
Hence monotonicity of $h_1$ is proved. Monotonicity of $h_2$ and $h_3$ are proved similarly.

\item{\bf($\mu$-variables, base case)}
By Condition~(\ref{item:f11Andf12_4}) in Definition~\ref{def:fwdFairBuechiSimWithDiv} and 
Condition~(\ref{item:asmthm:soundnessFwdFairBuechi3-2}) in Theorem~\ref{thm:soundnessFwdFairBuechi},
we have:
\begin{equation*}
h_1(0,\gamma_2) =u_{1}(0)\odot f^{(\gamma_2)}_{11}=\bot\odot f^{(\gamma_2)}_{11}=\bot\, \quad\text{and} \quad
h_2(\gamma_1,0)=u_2(\overline{\gamma_1})\odot f^{(0)}_{12}=u_2\odot\bot =\bot\,.
\end{equation*}

\item{\bf($\mu$-variables, step case)}
Let $\gamma_1\leq\overline{\gamma_1}$ and $\gamma_2\leq\overline{\gamma_2}$.
%
We have the following (see also Figure~\ref{fig:probMeasSoundnessFwdFairBuechiSim1}).
\allowdisplaybreaks[4]
\begin{align*}
&h_1(\gamma_1+1,\gamma_2) \\
&= u_{1}(\gamma_1+1)\odot f^{(\gamma_2)}_{11} \\ 
&\sqsubseteq (J\zeta)^{-1}\odot \overline{F}[u_{1}(\gamma_1),u_2(\gamma_1)]\odot c_1\odot f^{(\gamma_2)}_{11} \\ 
&\sqsubseteq (J\zeta)^{-1}\odot \overline{F}[u_{1}(\gamma_1),u_2(\overline{\gamma_1})]\odot c_1\odot f^{(\gamma_2)}_{11} \\ 
&\sqsubseteq (J\zeta)^{-1}\odot \overline{F}[u_{1}(\gamma_1),u_2(\overline{\gamma_1})]\odot 
\overline{F}\bigl[\codomJoin{f_{11}^{(\gamma_2)},f_{12}^{(\gamma_2)}},\codomJoin{f_{21},f_{22}}\bigr]
\odot d_{11} \\ 
&= (J\zeta)^{-1}\odot\overline{F}\bigl[[u_{1}(\gamma_1),u_2(\overline{\gamma_1})]\odot\codomJoin{f_{11}^{(\gamma_2)}, f_{12}^{(\gamma_2)}},
[u_{1}(\gamma_1),u_2(\overline{\gamma_1})]\odot\codomJoin{f_{21},f_{22}}\bigr]\odot d_{11}  \\
&= (J\zeta)^{-1}\odot
\overline{F}\bigl[[\id_Z,\id_Z]\odot\codomJoin{u_{1}(\gamma_1)\odot f_{11}^{(\gamma_2)},u_2(\overline{\gamma_1})\odot f_{12}^{(\gamma_2)}}, 
\\
&\qquad\qquad\qquad\qquad\qquad\qquad\qquad\qquad
[u_{1}(\gamma_1),u_2(\overline{\gamma_1})]\odot\codomJoin{f_{21},f_{22}}\bigr]\odot d_{11}  \\
&\sqsubseteq (J\zeta)^{-1}\odot
\overline{F}\bigl[[\id_Z,\id_Z]\odot\codomJoin{u_{1}(\gamma_1)\odot f_{11}^{(\gamma_2)}, u_2(\overline{\gamma_1})\odot f_{12}^{(\gamma_2)}},
\\
&\qquad\qquad\qquad\qquad\qquad\qquad\qquad\qquad
[u_{1}(\overline{\gamma_1}),u_2(\overline{\gamma_1})]\odot\codomJoin{f_{21},f_{22}}\bigr]\odot d_{11}  \\
&= (J\zeta)^{-1}\odot
\overline{F}\bigl[[\id_Z,\id_Z]\odot
\codomJoin{h_1(\gamma_1,\gamma_2), h_2(\gamma_1,\gamma_2)},h_3(\gamma_1,\gamma_2)\bigr]\odot d_{11}\,.  
\end{align*}
%
%
We can prove in a similar manner that there exists an ordinal $\gamma'_1$ such that
\[
h_2(\gamma_1,\gamma_2+1)\sqsubseteq\overline{F}\bigl[\codomJoin{h_1(\gamma'_1,\gamma_2),h_2(\gamma_1,\gamma_2)},h_3(\gamma_1,\gamma_2)\bigr]\odot d_{12}\,.
\]

\item{\bf($\mu$-variables, limit case)}
Let $\gamma_1$  be a limit ordinal such that $\gamma_1\leq\overline{\gamma_1}$.
By $\overline{\gamma_1}=\overline{\beta_1}\leq\omega$ and 
Assumption~\ref{asm:wellDefinedSimulation}.\ref{item:asm:wellDefinedTrace4},
Kleisli composition in $\Kl(T)$
is $\gamma_1$-continuous.
Hence for 
each 
ordinal $\gamma_2$, we have:
\begin{align*}
h_1(\gamma_1,\gamma_2) 
\,=\, u_{1}(\gamma_1)\odot f^{(\gamma_2)}_{11} 
&\,\sqsubseteq\,\bigl(\bigsqcup_{\gamma'_1<\gamma_1} u_{1}(\gamma'_1)\bigr)\odot f^{(\gamma_2)}_{11} \\
&\,=\,\bigsqcup_{\gamma'_1<\gamma_1}\bigl( u_{1}(\gamma'_1)\odot f^{(\gamma_2)}_{11}\bigr) 
\,=\,\bigsqcup_{\gamma'_1<\gamma_1}h_1(\gamma'_1,\gamma_2)\,. 
\end{align*}
In a similar manner we can prove that 
for an ordinal $\gamma_1$ and a limit ordinal $\gamma_2$,
\begin{equation*}
h_2(\gamma_1,\gamma_2) 
\,=\, \bigsqcup_{\gamma'_2<\gamma_2} h_2(\gamma_1,\gamma'_2)\,.
\end{equation*}

\item{\bf($\nu$-variables)}
Similarly to the step cases of $\mu$-variables, we can prove that
for ordinals $\gamma_1\leq\overline{\gamma_1}$ and $\gamma_2\leq\overline{\gamma_2}$ we have:
\[
h_3(\gamma_1,\gamma_2)
\sqsubseteq
(J\zeta)^{-1}\odot 
\overline{F}\bigl[\codomJoin{h_1(\overline{\gamma_1},\overline{\gamma_2}),h_2(\overline{\gamma_1},\overline{\gamma_2})},h_3(\overline{\gamma_1},\overline{\gamma_2})\bigr]\odot d_2 \,.
\]
\end{enumerate}
Hence $p$ is a progress measure for the equational system (\ref{eq:eqSysSublemsoundnessFwdFairBuechiSim10}).
\end{proof}

\begin{sublem}\label{sublem:soundnessFwdFairBuechiSim2}
We assume 
Assumption~\ref{asm:wellDefinedSimulation} and
the assumptions 
in Theorem~\ref{thm:soundnessFwdFairBuechi}.
Let $f:Y\kto X$ be a forward, $\overline{\alpha}$-bounded simulation with dividing from 
\begin{math}
  \X= 
  \bigl(
   (X_{1},X_{2}),
   c,
   s
\bigr)
\end{math}
and
\begin{math}
  \mathcal{Y}=
  \bigl(
   (Y_{1},Y_{2}),
    d,
    t
 \bigr)\,
 \end{math} 
%
%
Let $(v^{\sol}_1,v^{\sol}_2,v^{\sol}_3)$ be the solution of equational system 
(\ref{eq:eqSysSublemsoundnessFwdFairBuechiSim10}) in Sublemma~\ref{sublem:soundnessFwdFairBuechiSim1}. 
Then we have:
\begin{equation}
\label{eq:claimSublemSoundnessFwdFairBuechiSim2}
[\id_Z,\id_Z]\odot\codomJoin{v^{\sol}_1,v^{\sol}_2}\sqsubseteq \trB(d_1)\qquad\text{and}\qquad
v^{\sol}_3\sqsubseteq \trB(d_2)\, 
\end{equation}
where $\trB(d_1):Y_1\kto Z$ and $\trB(d_2):Y_2\kto Z$ are defined as 
$\trB(c_1)$ and $\trB(c_2)$.
\end{sublem}

\begin{proof}
It is easy to see that 
the equational system in (\ref{eq:eqSysSublemsoundnessFwdFairBuechiSim10}) is equivalent to the following 
equational system, in the sense that 
$w^{\sol}_1=(v^{\sol}_1,v^{\sol}_2)$ and $w^{\sol}_2=v^{\sol}_3$.
\begin{equation}
\begin{array}{lll}
\label{eq:eqSysSublemsoundnessFwdFairBuechiSim3}
	  w_{1} 
	  &\,=_{\mu}\,
	  \begin{pmatrix}
	  (J\zeta)^{-1}
	  \odot
	  \oF\bigl[[\id_{Z},\id_{Z}]\odot\codomJoin{w_{11},w_{12}},w_{2}\bigr]
	  \odot
	  d_{11},
	  \\ 
	  \quad(J\zeta)^{-1}
	  \odot
	  \oF\bigl[[\id_{Z},\id_{Z}]\odot\codomJoin{w_{11},w_{12}},w_{2}\bigr]
	  \odot
	  d_{12}
	  \end{pmatrix}
	  &\in 
	  \bigl(\Kl(T)(Y_1,Z)\bigr)^2
	  \\ 
	  w_{2} 
	  &\,=_{\nu}\,
	  (J\zeta)^{-1}
	  \odot
	  \oF\bigl[[\id_{Z},\id_{Z}]\odot\codomJoin{w_{11},w_{12}},w_{2}\bigr]
	  \odot
	  d_{2} 
	  	  &\in\Kl(T)(Y_2,Z) 
     \end{array}
\end{equation}
Here $(\Kl(T)(Y_1,Z)\bigr)^2$ is equipped with the product order, and $w_{11}$ and $w_{12}$ denote the first and the second component of $w_1\in\bigl(\Kl(T)(Y_1,Z)\bigr)^2$ respectively.

By completeness of progress measure (Theorem~\ref{thm:correctnessOfProgMeasEqSys}.\ref{item:completenessProgressMeas}), 
there exists a progress measure 
$q=\bigl((\alpha),(w_1(\alpha)=(w_{11}(\alpha),w_{12}(\alpha)),w_2(\alpha))_{\alpha\leq\overline{\alpha}}\bigr)$ for
(\ref{eq:eqSysSublemsoundnessFwdFairBuechiSim3}) such that 
$(w_{11}(\overline{\alpha}),w_{12}(\overline{\alpha}))=w^{\sol}_{1}=(v^{\sol}_1,v^{\sol}_2)$ and 
$w_2(\overline{\alpha})=w^{\sol}_2=v^{\sol}_3$.

For each $\alpha\leq\overline{\alpha}$, we define $v'_1(\alpha):Y_1\kto Z$ and $v'_2(\alpha):Y_2\kto Z$ by 
$v'_1(\alpha)=[\id_Z,\id_Z]\odot\codomJoin{w_{11}(\alpha),w_{12}(\alpha)}$ and
$v'_2(\alpha)=w_2(\alpha)$.

In what follows, we show that $p':=\bigl((\overline{\alpha}),(v'_1(\alpha),v'_2(\alpha))_{\alpha\leq\overline{\alpha}}\bigr)$ 
is a progress measure for the equational system 
that defines $\trB(d_1):Y_1\kto Z$ and $\trB(d_2):Y_2\kto Z$
(see (\ref{eq:eqSysForBuechiAcceptance}) in Theorem~\ref{thm:sanityCheckResult}).
\begin{enumerate}
\item{\bf (Monotonicity)} By the monotonicity of $w_{1}(\alpha)$ and $w_{2}(\alpha)$, $v'_1(\alpha)$ 
and $v'_2(\alpha)$ are also monotone.

\item{\bf ($\mu$-variables, base case)} 
%
We have $(w_{11}(0),w_{12}(0))=w_1(0)=(\bot,\bot)$ by the definition. Hence
 by Condition~(\ref{item:asmthm:soundnessFwdFairBuechi3-2}) of Theorem~\ref{thm:soundnessFwdFairBuechi},
 we have:
\[
v'_1(0)\,=\,[\id_Z,\id_Z]\odot\codomJoin{w_{11}(0),w_{12}(0)}
\,=\,[\id_Z,\id_Z]\odot\codomJoin{\bot,\bot}\,=\,
\bot\,.
\]

\item{\bf ($\mu$-variables, step case)} 
For an ordinal $\alpha\leq\overline{\alpha}$, we have:
\allowdisplaybreaks[4]
\begin{align*}
v'_1(\alpha+1) 
&= [\id_Z,\id_Z]\odot\codomJoin{w_{11}(\alpha+1),w_{12}(\alpha+1)} \\ 
&\sqsubseteq 
[\id_Z,\id_Z]\odot
\codomJoin{
	  (J\zeta)^{-1}
	  \odot
	  \oF\bigl[[\id_{Z},\id_{Z}]\odot\codomJoin{w_{11}(\alpha),w_{12}(\alpha)},w_{2}\bigr]
	  \odot
	  d_{11}, \\
&\qquad\qquad\qquad\qquad\,	  
(J\zeta)^{-1}
	  \odot
	  \oF\bigl[[\id_{Z},\id_{Z}]\odot\codomJoin{w_{11}(\alpha),w_{12}(\alpha)},w_{2}\bigr]
	  \odot
	  d_{12}
	  }   \\
&=
 (J\zeta)^{-1}
	  \odot 	  
	  \oF\bigl[[\id_{Z},\id_{Z}]\odot\codomJoin{w_{11}(\alpha),w_{12}(\alpha)},w_{2}\bigr]
	  \odot 
[\id_{FY},\id_{FY}]\odot\codomJoin{d_{11},d_{12}}   \\
&\sqsubseteq
 (J\zeta)^{-1}
	  \odot
	  \oF\bigl[[\id_{Z},\id_{Z}]\odot\codomJoin{w_{11}(\alpha),w_{12}(\alpha)},w_{2}\bigr]
	  \odot d_1   \\
&=(J\zeta)^{-1}
	  \odot
	  \oF\bigl[v'_{1}(\alpha),v'_{2}\bigr]
	  \odot d_1\,.
\end{align*}

\item{\bf ($\mu$-variables, limit case)} 
For a limit ordinal $\alpha\leq\overline{\alpha}$, we have:
\begin{align*}
v'_1(\alpha) 
&= [\id_Z,\id_Z]\odot\codomJoin{w_{11}(\alpha),w_{12}(\alpha)} \\
&\sqsubseteq [\id_Z,\id_Z]\odot\codomJoin{\textstyle{\bigsqcup_{\beta<\alpha}w_{11}(\beta),\bigsqcup_{\beta<\alpha}v'_{12}(\beta)}} \\
&= \textstyle{\bigsqcup_{\beta<\alpha}[\id_Z,\id_Z]\odot\codomJoin{w_{11}(\beta),w_{12}(\beta)} }\\
&=\textstyle{\bigsqcup_{\beta<\alpha} v'_1(\beta)}\,. 
\end{align*}

\item{\bf ($\nu$-variables)} 
For an ordinal $\alpha\leq\overline{\alpha}$, there exists an ordinal $\beta\leq\overline{\alpha}$ such that:
\begin{align*}
v'_2(\alpha)
&= w_2(\alpha) \\
&\sqsubseteq 
(J\zeta)^{-1}
	  \odot \oF\bigl[[\id_{Z},\id_{Z}]\odot\codomJoin{w_{11}(\beta),w_{12}(\beta)},w_{2}(\beta)\bigr]
	  \odot
	  d_{2} \\
&=
(J\zeta)^{-1}
	  \odot
	  \oF\bigl[v'_{1}(\beta),v'_{2}(\beta)\bigr]
	  \odot
	  d_{2}\,.
\end{align*}
\end{enumerate}
Hence 
$p'=\bigl((\overline{\alpha}),(v'_1(\alpha),v'_2(\alpha))_{\alpha}\bigr)$ 
is a progress measure and by  soundness of progress measures (Theorem~\ref{thm:correctnessOfProgMeasEqSys}.\ref{item:soundnessProgressMeas}), 
we have (\ref{eq:claimSublemSoundnessFwdFairBuechiSim2}).
\end{proof}

\begin{proof}[Proof (Lemma~\ref{lem:soundnessFwdFairBuechiSim})]
Let $E_{\mathcal{X}}$ be the equational system that defines $\trB(c_1)$ and 
$\trB(c_2)$ 
(see (\ref{eq:eqSysForBuechiAcceptance}) in Theorem~\ref{thm:sanityCheckResult}).
Let $p_{\mathcal{X}}=((\overline{\beta_1}),(u_{1}(\beta_1),u_{2}(\beta_1))_{\beta_1\leq\overline{\beta_1}})$ 
be the progress measure in Sublemma~\ref{sublem:soundnessFwdFairBuechiSim1}.
%
By Sublemma~\ref{sublem:soundnessFwdFairBuechiSim1} and soundness of progress measures (Theorem~\ref{thm:correctnessOfProgMeasEqSys}.\ref{item:soundnessProgressMeas}), 
we have:
\begin{multline}
\label{eq:proofSoundnessFwdFairBuechiSim2}
u_1(\overline{\beta_1})\odot f_{11}^{(\overline{\alpha_1})}\sqsubseteq v^{\sol}_1, \quad
u_2(\overline{\beta_1})\odot f_{12}^{(\overline{\alpha_1})}\sqsubseteq v^{\sol}_2, \; \\
\text{and}\qquad
[u_1(\overline{\beta_1}),u_2(\overline{\beta_1})]\odot\codomJoin{f_{21},f_{22}} \sqsubseteq v^{\sol}_3 \,.
\end{multline}
By Sublemma~\ref{sublem:soundnessFwdFairBuechiSim2}, we have:
\begin{equation}
\label{eq:proofSoundnessFwdFairBuechiSim3}
[\id_Z,\id_Z]\odot[v^{\sol}_1,v^{\sol}_2]\sqsubseteq \trB(d_1)
\quad\text{and}\quad
 v^{\sol}_3\sqsubseteq \trB(d_1)\,.
\end{equation}
Therefore we have:
\allowdisplaybreaks[4]
\begin{align*}
[\trB(c_1),\trB(c_2)]\odot\codomJoin{f_{11},f_{12}} 
&=[\trB(c_1),\trB(c_2)]\odot\codomJoin{f^{(\overline{\alpha_1})}_{11},f^{(\overline{\alpha_1})}_{12}}&
(\text{by Definition~\ref{def:fwdFairBuechiSimWithDiv}}) \\
&=[u_1(\overline{\beta_1}),u_2(\overline{\beta_1})]\odot\codomJoin{f^{(\overline{\alpha_1})}_{11},f^{(\overline{\alpha_1})}_{12}} 
&(\text{by definition}) \\
&=[\id_Z,\id_Z]\odot\codomJoin{u_1(\overline{\beta_1})\odot f_{11}^{(\overline{\alpha_1})},
u_2\odot f_{12}^{(\overline{\alpha_1})}}\hspace{-12mm} \\
&\sqsubseteq [\id_Z,\id_Z]\odot\codomJoin{v^{\sol}_1,v^{\sol}_2}
&(\text{by (\ref{eq:proofSoundnessFwdFairBuechiSim2})})\\
&\sqsubseteq \trB(d_1) &(\text{by (\ref{eq:proofSoundnessFwdFairBuechiSim3})}). \\
\intertext{In a similar manner, we can prove:}
[\trB(c_1),\trB(c_2)]\odot\codomJoin{f_{21},f_{22}} 
&\sqsubseteq  \trB(d_2)\,. 
\end{align*}
These conclude the proof. 
\end{proof}

\begin{proof}[Proof (Theorem~\ref{thm:soundnessFwdFairBuechi})]
We have:
\allowdisplaybreaks[4]
\begin{align*}
&\trB(\mathcal{X}) \\
&= [\trB(c_1),\trB(c_2)]\odot s & \tag*{(\text{by definition})} \\
&\sqsubseteq [\trB(c_1),\trB(c_2)]\odot [\codomJoin{f_{11},f_{12}},\codomJoin{f_{21},f_{22}}] \odot t 
& \tag*{(\text{by Condition~(\ref{item:f21Andf22}) 
 in Definition~\ref{def:fwdFairBuechiSimWithDiv}})} \\
&\sqsubseteq  \bigl[[\trB(c_1),\trB(c_2)]\odot\codomJoin{f_{11},f_{12}}, 
[\trB(c_1),\trB(c_2)]\odot\codomJoin{f_{21},f_{22}}\bigr] \odot t &\\
&\sqsubseteq  [\trB(d_1),\trB(d_2)]\odot t & \tag*{(\text{by Lemma~\ref{lem:soundnessFwdFairBuechiSim}})} \\
&=\trB(\mathcal{Y}). &\tag*{(\text{by definition})}
\end{align*}
This concludes the proof.
\end{proof}

We have thus obtained a sound simulation notion.
%
%
The proposition below shows that 
 soundness theorem (Theorem~\ref{thm:soundnessFwdFairBuechi}) applies to the combinations of
monads and functors in Definition~\ref{def:powersetMonadAndGiryMonad}--\ref{def:FSigmaAndFLambda}.
%
\begin{prop}\label{prop:concreteCondSatisfied}
The combinations of $\pow$ and $F_\Sigma$, and $\giry$ and $F_\myalphabet$,
respectively,
satisfy the assumptions in Assumption~\ref{asm:wellDefinedSimulation} and Theorem~\ref{thm:soundnessFwdFairBuechi}.
%
\qed
\end{prop}


Therefore by Theorem~\ref{thm:soundnessFwdFairBuechi} and Theorem~\ref{thm:sanityCheckResult},
by regarding NBTAs and PBWAs as B\"uchi $(T,F)$-systems as in Example~\ref{exa:NBTAandPBWAInduceBuechiSystems},
we can obtain fair simulation notions for these systems whose soundness comes for free.

A problem here is that the coalgebraic definition in Definition~\ref{def:fwdFairBuechiSimWithDiv} requires
a dividing $d_{11}, d_{12}\colon Y_{1}\kto \oF Y$ of $d_{1}\colon Y_{1}\kto \oF Y$. 
Intuitively this is to divide the simulator's
``resources'' of transitions into two parts, 
%
one for the challenger's
\emph{non-accepting} states and the other for \emph{accepting} states.

To describe an intuition, 
we hereby interpret the notion of dividing for NBTAs and PBWAs,
with respect to the  correspondence in Example~\ref{exa:NBTAandPBWAInduceBuechiSystems}.
For the NBTA $\mathcal{Y}$
in Def.~\ref{def:simForNondetBuechiRel},
a dividing is understood as a pair $\delta_{\mathcal{Y},11},\delta_{\mathcal{Y},12}:Y_1\to\pow(\coprod_{\sigma\in\Sigma} X^{|\sigma|})$ of functions such that $\delta_{\mathcal{Y},11}(x)\cup\delta_{\mathcal{Y},12}(y)=\delta_{\mathcal{Y}}(y)$ for each $y\in Y_1$.
If $\mathcal{Y}$ is the PBWA 
in Def.~\ref{def:fwdFairBuechiSimProbMatrix}, then
a dividing is understood as a pair $M_{\mathcal{Y},11}(a),M_{\mathcal{Y},12}(a)\in [0,1]^{Y_1\times Y}$ of matrices such that
$M_{\mathcal{Y},11}(a)+M_{\mathcal{Y},12}(a)=M_{\mathcal{Y},1}(a)$ for each $a\in\myalphabet$.


This dividing requirement is naturally inherited by the resulting concrete simulation notions for 
NBTAs and PBWAs.
Unfortunately finding such ``resource allocation'' is a challenge in practice;
  additionally, insistence on such allocation being
\emph{static} is overly restrictive, as
we will later see in Example~\ref{ex:staticDividingIsRestrictive}.

The following definition is more desirable in this respect; it
indeed yields Definition~\ref{def:simForNondetBuechiRel}
and~\ref{def:fwdFairBuechiSimProbMatrix}---the concrete simulation notions that we have introduced earlier---as its instances.
Note that the following definition is \emph{not sound} in the general sense of
Theorem~\ref{thm:soundnessFwdFairBuechi} (see Example~\ref{example:joinedSimNotSound} for a counterexample).
The
rest of the paper is devoted
to finding special cases in which it is sound. 

\begin{defi}[fair simulation without dividing] 
\label{def:fwdFairBuechiSimWithoutDividing}
 In the setting of Definition~\ref{def:fwdFairBuechiSimWithDiv}, 
a \emph{(forward, $\overline{\alpha}$-bounded) fair simulation without
 dividing} 
 is defined almost the same way as one with dividing in Definition~\ref{def:fwdFairBuechiSimWithDiv}, except that 
 Condition~(\ref{item:f11Andf12}) is replaced by the following.
 \begin{enumerate}
   \item[(\ref{item:f11Andf12}')]
  The components $f_{11}\colon Y_{1}\kto
	X_{1}$ and $f_{12}\colon Y_{1}\kto X_{2}$ come
	with 
	 \emph{approximation sequences}
       \[
\begin{array}{ll}
         f_{11}^{(0)} \sqsubseteq
        f_{11}^{(1)} \sqsubseteq
        \cdots
	\sqsubseteq
	f_{11}^{(\overline{\alpha})} 
	\,\colon
	Y_{1}\kto X_{1},
	\quad\text{and}\quad
        &
        f_{12}^{(0)} \sqsubseteq
        f_{12}^{(1)} \sqsubseteq
        \cdots 
	\sqsubseteq
	f_{12}^{(\overline{\alpha})} 
	\,\colon
	Y_{1}\kto X_{2},
\end{array}       
\]
that satisfies \ref{item:f11Andf12_2}, \ref{item:f11Andf12_4} and \ref{item:f11Andf12_6} in 
Definition~\ref{def:fwdFairBuechiSimWithDiv} and the following two conditions.
\begin{enumerate}


 \item[(b')]\label{item:f11Andf12_3ND}
 \textbf{($f_{11}^{(\alpha)}$)} 
       For each ordinal $\alpha$ such that $\alpha\le\overline{\alpha}$,
       the inequality~(\ref{eq:f11PostFixedPointND}) below holds.
    \begin{equation}\label{eq:f11PostFixedPointND}
  \small
   \vcenter{\xymatrix@C+8em@R=1.4em{
   {FY}
       \kar[r]^{\overline{F}\bigl[\,
         \codomJoin{f_{11}^{(\alpha)},f_{12}^{(\alpha)}},\,
	 \codomJoin{f_{21},f_{22}}
        \,\bigr]}
       \ar@{}[rd]|{\sqsupseteq}
   &
   {FX}
   \\
   {Y_{1}}
       \kar[u]^{d_{1}}
       \kar[r]_{f_{11}^{(\alpha)}}
   &
   {X_{1} }
       \kar[u]_{c_{1}}
 }}
  \end{equation}


 \item[(d')]\label{item:f11Andf12_5ND}
 \textbf{($f_{12}^{(\alpha)}$, the step case)} 
        For each 
        $\alpha<\overline{\alpha}$,
      the inequality~(\ref{eq:item:f11Andf12_3ND}) below holds.
  	  \begin{equation}\label{eq:item:f11Andf12_3ND}
	   \small
	   \vcenter{\xymatrix@C+8em@R=1.4em{
	   {FY}
	   \kar[r]^{\overline{F}\bigl[\,
	   \codomJoin{f_{11}^{(\alpha)},f_{12}^{(\alpha)}},\,
	   \codomJoin{f_{21},f_{22}}
	   \,\bigr]}
	       \ar@{}[rd]|{\sqsupseteq}
	   &
	   {FX}
	   \\
	   {Y_{1}}
	       \kar[u]^{d_{1}}
	       \kar[r]_{f_{12}^{(\alpha+1)}}
	   &
	   {X_{2}}
	       \kar[u]_{c_{2}}
	}}
\end{equation}
\end{enumerate}
\end{enumerate}
\end{defi}

The following example shows why this definition is more desirable.

%
\begin{exa}\label{ex:staticDividingIsRestrictive}
Let $\mathcal{X}=((X_1,X_2),c,s)$ and $\mathcal{Y}=((Y_1,Y_2),d,t)$ be
B\"uchi $(\giry,\{a\}\times(\place))$-systems
that model PBWAs illustrated 
below.
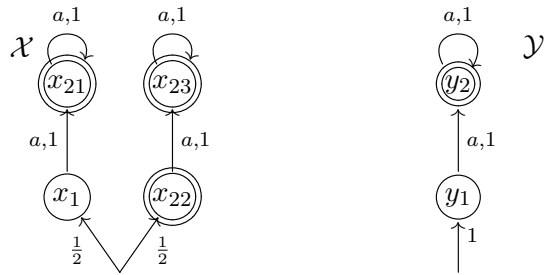
\[
\begin{xy}
(70,0)*+[Fo]{y_1} = "x1",
(70,15)*+[Foo]{y_2} = "x2",
(18,0)*+[Fo]{x_1} = "y1",
(18,15)*+[Foo]{x_{21}} = "y21",
(32,0)*+[Foo]{x_{22}} = "y22",
(32,15)*+[Foo]{x_{23}} = "y23",
(80,20)*{\mathcal{Y}}="",
(12,20)*{\mathcal{X}}="",
\ar _{1} (70,-10);"x1"*+++[o]{}
\ar _{a,1} "x1"*+++[o]{};"x2"*+++[o]{}
\ar @(ul,ur)^{a,1} "x2"*+++[o]{};"x2"
\ar ^{\frac{1}{2}} (25,-10);"y1"*+++[o]{}
\ar ^{a,1} "y1"*+++[o]{};"y21"*+++[o]{}
\ar _{\frac{1}{2}} (25,-10);"y22"*+++[o]{}
\ar _{a,1} "y22"*+++[o]{};"y23"*+++[o]{}
\ar @(ul,ur)^{a,1} "y21"*+++[o]{};"y21"*+++[o]{}
\ar @(ul,ur)^{a,1} "y23"*+++[o]{};"y23"*+++[o]{}
\end{xy}
\]
%
It is easy to see that they exhibit language inclusion.
We define 
$f:Y\kto X$ by 
%
\begin{equation*}
f(y_1)(\{x\})=\begin{cases} \frac{1}{2} & (x\in\{x_1,x_{22}\}) \\ 0 & (\text{otherwise}) \end{cases}
\qquad\text{and}\qquad
f(y_2)(\{x\})=\begin{cases} \frac{1}{2} & (x\in\{x_{21},x_{23}\}) \\ 0 & (\text{otherwise})\,. \end{cases}
\end{equation*}
%
Then $f$ is a fair simulation \emph{without} dividing 
from 
$\mathcal{X}$ to $\mathcal{Y}$.
In contrast, $f$ is not a fair simulation \emph{with} dividing.
In fact, there exists no fair simulation with dividing from $\mathcal{X}$ to $\mathcal{Y}$.
\end{exa}

In what follows we seek for conditions under which this desirable fair simulation notion (without dividing, Definition~\ref{def:fwdFairBuechiSimWithoutDividing}) turns out to be sound. In Section~\ref{subsec:nondetSetting} we study the nondeterministic setting, and in Section~\ref{subsec:specificSettingIIProb} we study the probabilistic setting. The identified conditions and soundness arguments are rather different between Section~\ref{subsec:nondetSetting}
and  Section~\ref{subsec:specificSettingIIProb}.

\subsection{Circumventing Dividing: the Nondeterministic Case}\label{subsec:nondetSetting}
For $T=\pow$ we  show that a simulation
\emph{without} dividing
yields 
one \emph{with} dividing. 
Therefore in the nondeterministic setting fair simulations without dividing are
sound unconditionally.
Here we exploit the \emph{idempotency} property of
$T=\pow$---\emph{one can copy resources as many times as one likes}. 

  \begin{prop} [soundness under idempotency]
 \label{prop:soundnessFwdFairBuechiProb}
 Under  Assumption~\ref{asm:wellDefinedSimulation}, let us assume
 that each arrow  $f:X\kto Y$ in $\Kl(T)$ is \emph{idempotent}, that is,
 the codomain join $\codomJoin{f,f}:X\kto Y+Y$ necessarily exists and
 we have $[\id_Y,\id_Y]\odot \codomJoin{f,f}=f$.
 \begin{enumerate}
  \item\label{item:prop:soundnessFwdFairBuechiProb1} A  simulation without dividing yields one with
	dividing, with the  dividing  $d_{11}=d_{12}=d_{1}$. 
  \item\label{item:prop:soundnessFwdFairBuechiProb2} Under the assumptions of
	Theorem~\ref{thm:soundnessFwdFairBuechi}, a 
	simulation without dividing witnesses trace inclusion.
 \qed
 \end{enumerate}
 \end{prop}

\begin{proof}
The item (\ref{item:prop:soundnessFwdFairBuechiProb1}) is immediate 
from the definitions of simulation with dividing and one without dividing 
(Definition~\ref{def:fwdFairBuechiSimWithDiv} and Definition~\ref{def:fwdFairBuechiSimWithoutDividing}).
%
The item (\ref{item:prop:soundnessFwdFairBuechiProb2}) follows from 
(\ref{item:prop:soundnessFwdFairBuechiProb1}) and soundness of forward fair simulation with dividing (Theorem~\ref{thm:soundnessFwdFairBuechi}).
\end{proof}
 
%
 

\begin{lem} 
\label{lem:idempotencyOfArrowsPow}
Arrows in $\Kl(\pow)$ are idempotent. Hence all the conditions in
 Proposition~\ref{prop:soundnessFwdFairBuechiProb} are satisfied by $T=\pow$
 and $F=F_{\Sigma}$, where $\Sigma$ is a ranked alphabet.
\qed
\end{lem}

There is still a gap between 
the simulation notion  in Definition~\ref{def:fwdFairBuechiSimWithoutDividing} (defined by inequalities)
and that in Definition~\ref{def:simForNondetBuechiRel} (defined by an
equational system).
The gap is filled by
another specific  property of 
$\Kl(\pow)$---\emph{reversibility}. It is much like in the following ``must''
predicate transformers.
%

\begin{lem}
\label{lem:propertyBoxf}
Let $f:B\kto C$ 
be an arrow in $\Kl(\pow)$.
We define 
$\Box_f:\Kl(\pow)(A,C)\to\Kl(\pow)(A,B)$ by
\begin{equation*}
\Box_f(g)(x):=\{y\in B\mid f(y)\subseteq g(x)\}\,.
\end{equation*}
Then we have the following:
\begin{enumerate}
\item\label{item:lem:propertyBoxf1} $f\odot \Box_f(g)\sqsubseteq g$
\item\label{item:lem:propertyBoxf2} $\forall h:A\kto B.\; f\odot h\sqsubseteq g\,\Rightarrow\,h\sqsubseteq \Box_f(g)$\,.
\qed
\end{enumerate}
\end{lem}

\noindent
The  construction $\Box_{f}$ is used to essentially ``reverse'' the
arrows $c_{1}$ and $c_{2}$ on the right of the
diagrams~(\ref{eq:f11PostFixedPoint}--\ref{eq:item:f11Andf12_3}).
This
allows to separate variables $f_{11}$, $f_{12}$, $f_{21}$ and $f_{22}$
alone 
and yield a (proper) equational system as in~(\ref{eq:1601151731}) below.

\begin{prop}\label{prop:equationalSysEquiv}
Let $g^{\sol}_1:Y_1\kto X_1$, $g^{\sol}_2:Y_1\kto X_2$,
$g^{\sol}_3:Y_2\kto X_1$ and $g^{\sol}_4:Y_2\kto X_2$ be the solution of the following equational system.
\begin{equation}\label{eq:1601151731}
\begin{array}{rll}
g_1 &=_{\nu}\, \Box_{c_1}(\overline{F}[\codomJoin{g_1,g_2},\codomJoin{g_3,g_4}]\odot d_1) &\in\Kl(\pow)(Y_1,X_1) \\ 
g_2 &=_{\mu} \, \Box_{c_2}(\overline{F}[\codomJoin{g_1,g_2},\codomJoin{g_3,g_4}]\odot d_1) &\in\Kl(\pow)(Y_1,X_2) \\ 
g_3 &=_{\nu} \, \Box_{c_1}(\overline{F}[\codomJoin{g_1,g_2},\codomJoin{g_3,g_4}]\odot d_2) &\in\Kl(\pow)(Y_2,X_1) \\ 
g_4 &=_{\nu} \, \Box_{c_2}(\overline{F}[\codomJoin{g_1,g_2},\codomJoin{g_3,g_4}]\odot d_2) &\in\Kl(\pow)(Y_2,X_2)  
\end{array}
\end{equation}
Let ${g}^{\sol}=[\codomJoin{{g}^{\sol}_1,{g}^{\sol}_2},
\codomJoin{{g}^{\sol}_3,{g}^{\sol}_4}]:Y\kto X$.
%
Then $s\sqsubseteq {g}^{\sol}\odot t$ 
if and only if
there is
a fair \newline
$\overline{\alpha}$-bounded simulation without dividing (Definition~\ref{def:fwdFairBuechiSimWithoutDividing})
from $\mathcal{X}$ to $\mathcal{Y}$ for some ordinal $\overline{\alpha}$.
\end{prop}

\begin{proof}
As in Definition~\ref{def:fwdFairBuechiSimWithDiv}, we write $f_{ji}:Y_j\kto X_i$ for 
the domain and codomain restriction of $f:Y\kto X$.

\myparagraph{($\Rightarrow$)}
Assume $s\sqsubseteq {g}^{\sol}\odot t$.
By completeness of progress measure (Theorem~\ref{thm:correctnessOfProgMeasEqSys}.\ref{item:completenessProgressMeas}), 
there exists a progress measure 
$g=\left((\overline{\alpha}), (g_i(\alpha))_{1\leq i\leq 4,\alpha\leq\overline{\alpha}}\right)$
such that ${g}^{\sol}_i=g_i(\overline{\alpha})$ for each $i$.
We define 
two sequences 
$(f_{11}^{(\alpha)}:Y_1\kto X_1)_{\alpha\leq\overline{\alpha}}$ and
 $(f_{12}^{(\alpha)}:Y_1\kto X_2)_{\alpha\leq\overline{\alpha}}$ 
by 
$f_{11}^{(\alpha)}=g_1(\alpha)$ and $f_{12}^{(\alpha)}=g_2(\alpha)$.
We define $f:Y\kto X$ by $f=g^{\sol}$.
We show that $f$ is a fair simulation without dividing from $\mathcal{X}$ to $\mathcal{Y}$ 
whose approximation sequences are given by $(f_{11}^{(\alpha)})_{\alpha\leq\overline{\alpha}}$ and
 $(f_{12}^{(\alpha)})_{\alpha\leq\overline{\alpha}}$.

We first show that $f$ satisfies Condition~(\ref{item:f21Andf22}) in Definition~\ref{def:fwdFairBuechiSimWithDiv}. 
%
We have:
\begin{align}
&c\odot f \notag\\
&=c\odot\left[\codomJoin{f_{11},f_{12}},\codomJoin{f_{21},f_{22}}\right] \notag\\
&\sqsubseteq c\odot \left[
\codomJoin{
\Box_{c_1}\left(\overline{F}f\odot d_1\right),
\Box_{c_2}\left(\overline{F}f\odot d_1\right)}, 
\codomJoin{
\Box_{c_1}\left(\overline{F}f\odot d_2\right),
\Box_{c_2}\left(\overline{F}f\odot d_2\right)}
\right] \notag\\
&=\left[
\codomJoin{
c_1\odot\Box_{c_1}\left(\overline{F}f\odot d_1\right),
c_2\odot \Box_{c_2}\left(\overline{F}f\odot d_1\right)}, 
\codomJoin{
c_1\odot\Box_{c_1}\left(\overline{F}f\odot d_2\right),
c_2\odot \Box_{c_2}\left(\overline{F}f\odot d_2\right)}
\right]\notag\\
&\sqsubseteq\left[
\codomJoin{
\overline{F}f\odot d_1,
\overline{F}f\odot d_1}, 
\codomJoin{
\overline{F}f\odot d_2,
\overline{F}f\odot d_2}
\right] 
&\tag*{\text{(by Lemma~\ref{lem:propertyBoxf}.\ref{item:lem:propertyBoxf1})}}\\ 
&=\left[
\overline{F}f\odot d_1, 
\overline{F}f\odot d_2
\right] 
&\tag*{
\text{(by  Lemma~\ref{lem:idempotencyOfArrowsPow})}}\\ 
&=
\overline{F}f\odot d\,.  \label{eq:1702051606}
\end{align}
Moreover, by the assumption, we have
$s \sqsubseteq f \odot t$.

Next we show that $f$ satisfies 
the Condition~(\ref{item:f11Andf12}') in Definition~\ref{def:fwdFairBuechiSimWithoutDividing}.
Note that $g=\left((\overline{\alpha}), (g_i(\alpha))_{1\leq i\leq 4,\alpha\leq\overline{\alpha}}\right)$
satisfies 
the axioms of progress measure (Definition~\ref{def:progMeas}).
It is immediate that this implies that conditions \ref{item:f11Andf12_2}, \ref{item:f11Andf12_4} and \ref{item:f11Andf12_6} in Definition~\ref{def:fwdFairBuechiSimWithDiv} are satisfied.
Moreover in a similar manner to (\ref{eq:1702051606}) above,
this implies that 
2'b' and 2'd' in Definition~\ref{def:fwdFairBuechiSimWithoutDividing} are also satisfied.
Therefore 
$f$ is a fair $\overline{\alpha}$-bounded simulation without dividing.

\myparagraph{($\Leftarrow$)}
Conversely, let $f:Y\kto X$ be a  fair simulation without dividing from $\mathcal{X}$ to $\mathcal{Y}$
whose approximation sequences are given by 
$(f_{11}^{(\alpha)})_{\alpha\leq\overline{\alpha}}$ and $(f_{12}^{(\alpha)})_{\alpha\leq\overline{\alpha}}$.
 For each $\alpha\leq\overline{\alpha}$,
 we define arrows $g_1(\alpha):Y_1\kto X_1$, $g_2(\alpha):Y_1\kto X_2$, $g_3(\alpha):Y_2\kto X_1$ and
 $g_4(\alpha):Y_2\kto X_2$, by 
 $g_1(\alpha)=f_{11}^{(\alpha)}$, $g_2(\alpha)=f_{12}^{(\alpha)}$, $g_3(\alpha)=f_{21}$ and $g_4(\alpha)=f_{22}$.
Then by using Lemma~\ref{lem:propertyBoxf}.\ref{item:lem:propertyBoxf2},
we can easily show that 
Condition~(\ref{item:f21Andf22}) in Definition~\ref{def:fwdFairBuechiSimWithDiv} and
Condition~(\ref{item:f11Andf12}') in Definition~\ref{def:fwdFairBuechiSimWithoutDividing}, together with monotonicity of 
$f_{11}^{(\alpha)}$ and $f_{12}^{(\alpha)}$, imply that $g$ satisfies the axioms of a progress measure 
(Definition~\ref{def:progMeas}) with respect to the equational system (\ref{eq:1601151731}).
\end{proof}

\noindent
By Proposition~\ref{prop:equationalSysEquiv} 
and that $\overline{F}$ and $\odot$ in $\Kl(\pow)$ are $\alpha$-continuous
for an arbitrary limit ordinal $\alpha$,
it is not hard to translate
Proposition~\ref{prop:soundnessFwdFairBuechiProb} into
Theorem~\ref{thm:soundnessOfSimulationEqSys}. 

\begin{proof}[Proof (Theorem\ \ref{thm:soundnessOfSimulationEqSys})] \sloppypar
Let $\mathcal{X}$ and $\mathcal{Y}$ be NBTAs and
\begin{center}
$\mathcal{X}=((X_1,X_2),c,s)$ and $\mathcal{Y}=((Y_1,Y_2),d,t)$
\end{center}
be the corresponding B\"uchi $(\pow,F_\Sigma)$-systems
(see Example~\ref{exa:NBTAandPBWAInduceBuechiSystems}).

Note that
for each $A,B\in\Sets$,
a  function $\Delta_{A,B}:\pow(A\times B)\to \Kl(\pow)(B,A)$ 
that is defined by $\Delta_{A,B}(S)(b):=\{a\in A\mid (a,b)\in S\}$ is a bijection.
%
It is  easy to see that 
for the functions in Definition~\ref{def:simForNondetBuechiRel} and Lemma~\ref{lem:propertyBoxf},
we have the following (recall that $F_{\Sigma}A=\coprod_{\sigma\in \Sigma}A^{|\sigma|}$).
%
\[
\begin{array}{rll}
\Delta_{X_i,Y}(\Box_{\mathcal{X},i}(S))&=\Box_{c_i}(\Delta_{FX,B'}(S))
&\text{for $i\in{1,2}$ and $S\subseteq \coprod_{\sigma\in\Sigma}X^{|\sigma|}\times Y$} \\
\Delta_{FX,Y_j}(\Diamond_{\mathcal{Y},j}(T))&=(\Delta_{FX,FY}(T))\odot d_j \quad
&\text{for $j\in\{1,2\}$ and $T\subseteq \coprod_{\sigma\in\Sigma}X^{|\sigma|} \times \coprod_{\sigma\in\Sigma}Y^{|\sigma|}$} \\
\Delta_{FX,FY}(\textstyle{\bigwedge_{\Sigma}(U)})&=\overline{F}(\Delta_{X,Y}(U))
&\text{for $U\subseteq X\times Y$}
\end{array}
\]
It is also easy to see that $R\subseteq Y\times X$ satisfies 
$\forall x\in \initSet_\mathcal{X}.\; \exists y\in \initSet_\mathcal{Y}.\; (x,y)\in R$
if and only if $s\sqsubseteq \Delta_{Y,X}(R)\odot t$.


Hence together with Proposition~\ref{prop:soundnessFwdFairBuechiProb} and Proposition~\ref{prop:equationalSysEquiv}, 
we have:
\begin{align*}
&\text{$R\subseteq X\times Y$ is a fair simulation from $\mathcal{X}$ to $\mathcal{Y}$ in the sense of Definition~\ref{def:simForNondetBuechiRel}} \\
&\Leftrightarrow 
\text{$\forall x\in \initSet_\mathcal{X}.\; \exists y\in \initSet_\mathcal{Y}.\; (x,y)\in R$,}\\
&\qquad \text{and $R\subseteq u^\sol_1\cup u^\sol_2\cup u^\sol_3\cup u^\sol_4$ where $u_1^\sol,\ldots,u_4^\sol$ are the solution of (\ref{eq:1601061008Rel})} \\
&\Leftrightarrow
\text{$s\sqsubseteq \Delta_{Y,X}(R)\odot t$,}\\
&\qquad \text{and $\Delta_{X,Y}(R)\sqsubseteq [\codomJoin{{g}^{\sol}_1,{g}^{\sol}_2},
\codomJoin{{g}^{\sol}_3,{g}^{\sol}_4}]$ where $g_1^\sol,\ldots,g_4^\sol$ are the solution of (\ref{eq:1601151731})} \\
&\Leftrightarrow \text{there is
a fair 
$\overline{\alpha}$-bounded simulation without dividing 
from $\mathcal{X}$ to $\mathcal{Y}$ for some  $\overline{\alpha}$} \\
&\Rightarrow\text{language inclusion, that is, $\lang(\mathcal{X})\subseteq\lang(\mathcal{Y})$}\,.
\end{align*}
%
%
%
This concludes the proof.
\end{proof}

 The results here in Section~\ref{subsec:nondetSetting} are 
axiomatic---with idempotency and reversibility---and they 
 apply to monads other than $\pow$. One  example is the
\emph{lift monad} $\lift =1+(\place)$,
which is used for potential nontermination (see~\cite{hasuo07generictrace} for example).

\subsection{Circumventing Dividing: the Probabilistic Case}\label{subsec:specificSettingIIProb}
We turn to the probabilistic setting and 
prove Theorem~\ref{thm:soundnessOfSimulationGiryMainMatrix}.
The 
strategy for $T=\pow$ does not work here, because
$\giry$ lacks idempotency (see
Proposition~\ref{prop:soundnessFwdFairBuechiProb}). 
We shall rely on other restrictions: from trees to words;
and  finite state spaces on the simulating side.

We start with an axiomatic development.

\begin{prop}\label{prop:unjoinedbotSoundMain2}
Besides
Assumption~\ref{asm:wellDefinedSimulation} and 
 the assumptions in Theorem~\ref{thm:soundnessFwdFairBuechi},
  assume 
  $d_1=\overline{F}(\id_{Y_1}+\bot_{Y_2,Y_2})\odot d_1$\,.
Then existence of a  fair 
simulation \emph{without dividing}
from $\mathcal{X}$ to $\mathcal{Y}$ 
implies
trace inclusion.
\end{prop}

\begin{proof}
Let $(f_{ij}:Y_i\kto X_j)_{i,j\in\{1,2\}}$ be a fair simulation without dividing from $\mathcal{X}$ to $\mathcal{Y}$.
Let $f_{11}^{(0)} \sqsubseteq f_{11}^{(1)} \sqsubseteq \cdots \sqsubseteq f_{11}^{(\overline{\alpha})} =f_{11}$
and $f_{12}^{(0)} \sqsubseteq f_{12}^{(1)} \sqsubseteq \cdots \sqsubseteq f_{12}^{(\overline{\alpha})} =f_{12}$
be the approximation sequences.

Recall that 
$\trB(c_1):X_1\kto Z$ and $\trB(c_2):X_2\kto Z$
are given by the solutions $u_1^\sol$ and $u_2^\sol$ of the following equational system.
(Theorem~\ref{thm:sanityCheckResult}).
	\begin{equation}\label{eq:1702080946}
	 \begin{aligned}
	  u_{1} 
	  &\;=_{\mu}\;
	  (J\zeta)^{-1}
	  \odot
	  \oF[u_{1},u_{2}]
	  \odot
	  c_{1}
	  &\in\Kl(T)(X_1,Z)
	  \\
	  u_{2} 
	  &\;=_{\nu}\;
	  (J\zeta)^{-1}
	  \odot
	  \oF[u_{1},u_{2}]
	  \odot
	  c_{2}
	  	  &\in\Kl(T)(X_2,Z)
	 \end{aligned}
	\end{equation}	
%
By completeness of progress measure (Theorem~\ref{thm:correctnessOfProgMeasEqSys}.\ref{item:completenessProgressMeas}), 
there exists a progress measure 
\[
p_\X=\bigl((\overline{\beta}),(u_1(\beta):X_1\kto Z,u_2(\beta):X_2\kto Z)_{\beta\leq\overline{\beta}}\bigr)
\]
for (\ref{eq:1702080946}) 
such that
$u_1(\overline{\beta})=\trB(c_1)$ and $u_2(\overline{\beta})=\trB(c_2)$.

We define two ordinals $\overline{\gamma_1}$ and 
$\overline{\gamma_2}$ by $\overline{\gamma_1}=\overline{\beta}$ and
$\overline{\gamma_2}=\overline{\alpha}$.
%
%
For each $\gamma_1\leq \overline{\gamma_1}$ 
and $\gamma_2\leq \overline{\gamma_2}$,
we define 
$h_1(\gamma_1,\gamma_2):Y_1\kto Z$ and
$h_2(\gamma_1,\gamma_2):Y_1\kto Z$ 
\[
h_1(\gamma_1,\gamma_2)= u_{1}(\gamma_1)\odot f^{(\gamma_2)}_{11},  \quad\text{and}\quad
h_2(\gamma_1,\gamma_2)= u_{2}(\gamma_1)\odot f^{(\gamma_2)}_{12}, \\
\]

We shall prove $h_1(\gamma_1,\gamma_2)= h_2(\gamma_1,\gamma_2)=\bot$ by transfinite induction on $\gamma_1$ and $\gamma_2$.

\myparagraph{(base case)}
If $\gamma_1=0$, for each $\gamma_2\leq\overline{\gamma_2}$ we have:
\begin{align*}
h_1(\gamma_1,\gamma_2)&=u_{1}(0)\odot f^{(\gamma_2)}_{11} & \tag*{(\text{by definition})} \\
&=\bot\odot f^{(\gamma_2)}_{11} & \tag*{(\text{$p_\X$ is a progress measure})} \\
&=\bot & \tag*{(\text{by Condition~(\ref{item:asmthm:soundnessFwdFairBuechi3-2}) in Theorem~\ref{thm:soundnessFwdFairBuechi}})}
\end{align*}
Similarly, if $\gamma_2=0$, for each $\gamma_1\leq\overline{\gamma_1}$ we have:
$h_2(\gamma_1,\gamma_2)=\bot$\,.

\myparagraph{(step case)}
Assume 
we have 
$h_1(\gamma_1,\gamma_2+1)=h_2(\gamma_1,\gamma_2+1)=\bot$.
Then we have:
\allowdisplaybreaks[4]
\begin{align*}
&h_1(\gamma_1+1,\gamma_2+1)\\
&=u_{1}(\gamma_1+1)\odot f^{(\gamma_2+1)}_{11} & \tag*{(\text{by definition})} \\
&\sqsubseteq J\zeta^{-1}\odot\overline{F}\bigl[u_1(\gamma_1),u_2(\gamma_1)\bigr]\odot c_1\odot  f^{(\gamma_2+1)}_{11} & \tag*{(\text{$p_\X$ is a progress measure})} \\
&\sqsubseteq J\zeta^{-1}\odot\overline{F}\bigl[u_1(\gamma_1),u_2(\gamma_1)\bigr]\odot \overline{F}
\bigl[\codomJoin{f_{11}^{(\gamma_2+1)},f_{12}^{(\gamma_2+1)}},\codomJoin{f_{21},f_{22}}\bigr]\odot d_1  \hspace{-20cm}& \\
& & \tag*{(\text{$f$ is a forward fair simulation})} \\
&= J\zeta^{-1}\odot\overline{F}\bigl[u_1(\gamma_1),u_2(\gamma_1)\bigr]\odot \overline{F}
\bigl[\codomJoin{f_{11}^{(\gamma_2+1)},f_{12}^{(\gamma_2+1)}},\codomJoin{f_{21},f_{22}}\bigr]\odot 
\overline{F}(\id+\bot)\odot d_1 \hspace{-20cm} & \\
& & \tag*{(\text{by the assumption})} \\
&= J\zeta^{-1}\odot\overline{F}\Bigl[
\bigl[u_1(\gamma_1),u_2(\gamma_1)\bigr]\odot\codomJoin{f_{11}^{(\gamma_2+1)},f_{12}^{(\gamma_2+1)}}\odot \id, 
\hspace{-20cm} & \\
& &  
\bigl[u_1(\gamma_1),u_2(\gamma_1)\bigr]\odot\codomJoin{f_{21},f_{22}}\odot \bot\Bigr]
\odot d_1   \\
&= J\zeta^{-1}\odot\overline{F}\Bigl[
\bigl[\id,\id\bigr]\odot\codomJoin{u_1(\gamma_1)\odot f_{11}^{(\gamma_2+1)},u_2(\gamma_1)\odot f_{12}^{(\gamma_2+1)}}, 
\bot\Bigr]
\odot d_1  \hspace{-20cm} &  \\
&= J\zeta^{-1}\odot\overline{F}\Bigl[
\bigl[\id,\id\bigr]\odot\codomJoin{h_1(\gamma_1,\gamma_2+1),h_2(\gamma_1,\gamma_2+1)}, 
\bot\Bigr]
\odot d_1  \hspace{-20cm} & \tag*{\text{(by definition)}} \\
&= J\zeta^{-1}\odot\overline{F}\Bigl[
\bigl[\id,\id\bigr]\odot\codomJoin{\bot,\bot}, \bot\Bigr]
\odot d_1  \hspace{-20cm} &   \tag*{(\text{by the induction hypothesis})} \\
&=\bot\,.  \hspace{-20cm}
\end{align*}
Similarly, if 
$h_1(\gamma_1+1,\gamma_2)=h_2(\gamma_1+1,\gamma_2)=\bot$,
then we have:
%
$h_2(\gamma_1+1,\gamma_2+1)=\bot$\,.

\myparagraph{(limit case)}
Assume that $\gamma_1$ is a limit ordinal and we have 
$h_1(\gamma'_1,\gamma_2)=h_2(\gamma'_1,\gamma_2)=\bot$ for each $\gamma'_1<\gamma_1$, and
$h_1(\gamma_1,\gamma'_2)=h_2(\gamma_1,\gamma'_2)=\bot$ for each $\gamma'_2<\gamma_2$.

We first prove $h_1(\gamma_1,\gamma_2)=\bot$.
\begin{align*}
h_1(\gamma_1,\gamma_2) 
&=u_1(\gamma_1)\odot f_{11}^{(\gamma_2)} & \tag*{(\text{by definition})} \\
&\sqsubseteq\Bigl( \bigsqcup_{\gamma'_1<\gamma_1}u_1(\gamma'_1)\Bigr)\odot f_{11}^{(\gamma_2)} & \tag*{(\text{$p_\X$ is a progress measure})} \\
&= \bigsqcup_{\gamma'_1<\gamma_1}\Bigl(u_1(\gamma'_1)\odot f_{11}^{(\gamma_2)}\Bigr) & 
\tag*{(\text{$\Kl(T)$ is $\Cppo$-enriched})} \\
&= \bigsqcup_{\gamma'_1<\gamma_1}h_1(\gamma'_1,\gamma_2) & 
\tag*{(\text{by definition})} \\
&=\bot & \tag*{(\text{by IH})}
\end{align*}
We next prove $h_2(\gamma_1,\gamma_2)=\bot$.
If $\gamma_2$ is  zero or a successor ordinal, we can prove $h_2(\gamma_1,\gamma_2)=\bot$ 
much like the base case and the step cases in the above.
If $\gamma_2$ is a limit ordinal, 
then in a similar manner to the above, we have
$h_2(\gamma_1,\gamma_2) =\bot$\,.
%
Hence we have $h_1(\gamma_1,\gamma_2)=h_2(\gamma_1,\gamma_2)=\bot$.

Similarly, for a limit ordinal $\gamma_2\leq\overline{\gamma_2}$ such that
$h_1(\gamma'_1,\gamma_2)=h_2(\gamma'_1,\gamma_2)=\bot$ for each $\gamma'_1<\gamma_1$, and
$h_1(\gamma_1,\gamma'_2)=h_2(\gamma_1,\gamma'_2)=\bot$ for each $\gamma'_2<\gamma_2$,
we have $h_1(\gamma_1,\gamma_2)=h_2(\gamma_1,\gamma_2)=\bot$.

Hence we have $h_1(\gamma_1,\gamma_2)=h_2(\gamma_1,\gamma_2)=\bot$ for each $\gamma_1\leq\overline{\gamma_1}$ and
$\gamma_2\leq\overline{\gamma_2}$.
Therefore we have:
\begin{equation}\label{eq:1702091156}
[\trB(c_1),\trB(c_2)]\odot\codomJoin{f_{11},f_{12}}=\bot\sqsubseteq\trB(d_1)\,. 
\end{equation}

We define $h_3:Y_2\to Z$ by 
$h_3=\bigl[\trB(c_1),\trB(c_1)\bigr]\odot\codomJoin{f_{21},f_{22}}$.
Then we have:
\begin{align*}
h_3
&=\bigl[\trB(c_1),\trB(c_1)\bigr]\odot\codomJoin{f_{21},f_{22}} & \tag*{(\text{by definition})} \\
&\sqsubseteq 
J\zeta^{-1}\odot \overline{F}\Bigl[
\bigl[\trB(c_1),\trB(c_1)\bigr]\odot \bigl[f_{11},f_{12}\bigr],
\bigl[\trB(c_1),\trB(c_1)\bigr]\odot \bigl[f_{21},f_{22}\bigr]
\Bigr]\odot d_2 \hspace{-5cm} & \\
&
& \tag*{(\text{similarly to the above})} \\
&=
J\zeta^{-1}\odot \overline{F}\bigl[
\bot,
h_3\bigr]\odot d_2
& \tag*{(\text{by definition and discussions above})} \\
&\sqsubseteq
J\zeta^{-1}\odot \overline{F}\bigl[
l_1^{(1)}(h_3),
h_3\bigr]\odot d_2\,.
\end{align*}
Here $l_{1}^{(1)}:Y_1\to Z$ denotes the first interim solution in Definition~\ref{def:solOfEqSys}.
Note that $\trB(d_2):Y_2\to Z$ is the greatest fixed point of the following function.
\[
g\; \mapsto \; J\zeta^{-1}\odot \overline{F}\bigl[
l_1^{(1)}(g),
g\bigr]\odot d_2 
\]
Hence by the Knaster-Tarski theorem, we have 
\begin{equation}\label{eq:1702091212}
[\trB(c_1),\trB(c_2)]\odot\codomJoin{f_{21},f_{22}}=h_3\sqsubseteq\trB(d_2)\,. 
\end{equation}
By (\ref{eq:1702091156}) and (\ref{eq:1702091212}),
in a similar manner to the proof of Theorem~\ref{thm:soundnessFwdFairBuechi},
we can prove $\trB(\mathcal{X})\sqsubseteq\trB(\mathcal{Y})$.
\end{proof}

The following (non-coalgebraic) lemma states that 
if $T=\giry$, $F=\myalphabet\times(\place)$ and 
the state space of $\mathcal{Y}$ is finite,
 then
we can modify $\mathcal{Y}$ so that the assumption in Proposition~\ref{prop:unjoinedbotSoundMain2}
holds without changing its language.
 The modification is derived from  the well-known \emph{fairness} result on
  Markov chains (see~\cite[Chapter~10]{baierK08principlesofmodelchecking} for example).
   Concretely, 
this result states that a nonaccepting state $\nonaccstate$ from which an accepting state is 
reachable
in
a positive probability
can be changed into an accepting state $\accstate$.
The proof uses the notion of \emph{bottom strongly} \linebreak
\emph{connected component}. 

\begin{lem}\label{lem:nonAcctoAccMain2}
Let $\myalphabet$ be  a countable set and
$\mathcal{Y}=((Y_1,Y_2),d,t)$ be a B\"uchi $(\giry,\myalphabet\times(\place))$-system
such that
$Y_1$ and $Y_2$ are finite sets.
%
Let $y_{>0}\in Y_1$ be a state 
such that
$d(y_{>0})(\myalphabet\times Y_2)>0$.
We define a B\"uchi $(\giry,\myalphabet\times(\place))$-system $\mathcal{Y}'=((Y'_1,Y'_2),d',t')$ by:
$Y'_1=Y_1\setminus \{y_{>0}\}$, $Y'_2=Y_2+\{y_{>0}\}$, $d'=d$ and $t'=t$. 
Note that $d'$ and $t'$ are well-defined because $Y'_1+Y'_2=Y_1+Y_2$.

Then we have $[\trB(d_1),\trB(d_2)]=[\trB(d'_1),\trB(d'_2)]$, and
moreover, $\trB(\mathcal{Y})=\trB(\mathcal{Y}')$.
\end{lem}

\begin{proof}
Let $Y=Y_1+Y_2$. Note that $Y$ is a finite set equipped with a discrete $\sigma$-algebra.

The B\"uchi $(\giry,\myalphabet\times(\place))$-system $\mathcal{Y}=((Y_1,Y_2),t,d)$ induces 
a Markov chain $\mathcal{M}_{\mathcal{Y}}$ such that
the state space is defined by $Y_{\bot}=Y+\{\bot\}$
 and the transition function $\tau_{\mathcal{Y}}:Y_{\bot}\times Y_{\bot}\to[0,1]$ is given by
\[
\tau_{\mathcal{Y}}(y,y')=
\begin{cases}
\sum_{a\in\myalphabet}d(y)(\{(a,y')\}) & (y,y'\in Y) \\
1-\sum_{y'\in Y}\sum_{a\in\myalphabet}d(y)(\{(a,y')\}) & (y\in Y, y'=\bot) \\
1 & (y=y'=\bot) \\
0 & (\text{otherwise})\,.
\end{cases}
\]
A Markov chain $\mathcal{M}_{\mathcal{Y}'}$ is defined similarly.

A subset $B\subseteq Y$ is called a \emph{strongly connected component} (SCC for short) if 
for all $y,y'\in S$, 
there exist $y_0,y_1,\ldots,y_n$ such that $y_0=y$, $y_n=y'$ and $\tau_{\mathcal{Y}}(y_i,y_{i+1})>0$ for each $i$.
An SCC $B$ is called a \emph{bottom strongly connected component} (BSCC for short) if 
$\tau_{\mathcal{Y}}(y,y')=0$ for each $y\in B$ and $y'\notin B$.
For more details, see~\cite{baierK08principlesofmodelchecking} for example.

For $Y'\subseteq Y$, we write $\Pr(y\models\Globally\Future Y')$ for the probability in which
a state in $Y'$ is visited infinitely often on $\mathcal{M}_{\mathcal{Y}}$ from $y\in Y$.
By Theorem~\ref{thm:sanityCheckResult}.\ref{item:thm:sanityCheckResult:ProbSanity}, we have:
\begin{align*}
[\trB(d_1),\trB(d_2)](y)(\myalphabet^{\omega})&=\Pr(y\models\Globally\Future Y_2)\,,\qquad\text{and}\\ 
[\trB(d'_1),\trB(d'_2)](y)(\myalphabet^{\omega})&=\Pr(y\models\Globally\Future (Y_2+\{y_{>0}\}))\,.
\end{align*}

We define $U,U'\subseteq Y$ by 
\begin{align*}
U&:=\bigcup\{B\subseteq Y\mid \text{$B$ is a BSCC and $B\cap Y_2\neq \emptyset$}\}\,, \qquad\text{and} \\
U'&:=\bigcup\{B\subseteq Y\mid \text{$B$ is a BSCC and $B\cap (Y_2+\{y_{>0}\})\neq \emptyset$}\}\,.
\end{align*}
We write $\Pr(y\models\Future U)$ for a probability in which a state in $U$ is reached in $\mathcal{M}_{\mathcal{Y}}$.
It is known that $\Pr(y\models\Globally\Future Y_2)=\Pr(y\models\Future U)$
(see~\cite[Corollary~10.34]{baierK08principlesofmodelchecking} for example).
Similarly, we have $\Pr(y\models\Globally\Future (Y_2+\{y_{>0}\}))=\Pr(y\models\Future U')$.

Assume that $y_{>0}\in B$ for some BSCC $B$ in $\mathcal{M}_{\mathcal{Y}}$. 
As $B$ is a BSCC, it has no outgoing transition, on one hand.
On the other hand, by $d(y_{>0})(\myalphabet\times Y_2)>0$, $y_{>0}$ has an accepting successor state.
Hence by the definition of $U$, 
we have $B\cap Y_2\neq \emptyset$ and this implies that $U=U'$.

If $y_{>0}\notin B$ for any BSCC $B$, then by the definitions of $U$ and $U'$ we have $U=U'$.

Therefore in both cases, for each $y\in Y$, we have:
\begin{align*}
[\trB(d_1),\trB(d_2)](y)(\myalphabet^{\omega})
&\;=\;\Pr(y\models\Globally\Future Y_2)\\
&\;=\;\Pr(y\models\Future U)\\
&\;=\;\Pr(y\models\Future U')\\
&\;=\;\Pr\left(y\models\Globally\Future (Y_2+\{y_{>0}\})\right)\\
&\;=\;[\trB(d'_1),\trB(d'_2)](y)(\myalphabet^{\omega})\,.
\end{align*}

It remains to prove 
$[\trB(d_1),\trB(d_2)](y)(A)=[\trB(d'_1),\trB(d'_2)](y)(A)$
for each measurable set $A\subseteq\myalphabet^{\omega}$.
To this end,
by Carath\'eodory's extension theorem (see~\cite{ashD2000probability} for example),
it suffices to prove 
$[\trB(d_1),\trB(d_2)](y)(w\myalphabet^\omega)=[\trB(d'_1),\trB(d'_2)](y)(w\myalphabet^\omega)$
for each $w\in\myalphabet^*$.

We inductively define a function 
$\chi_{\mathcal{Y}}:Y\times\myalphabet^*\to \giry Y$ by 
\begin{align*}
\chi_{\mathcal{Y}}(y,\empseq)(\{y'\})&=\begin{cases} 1 & (y=y') \\ 0 & (\text{otherwise}) \end{cases}& \text{and}\\
\chi_{\mathcal{Y}}(y,aw)(\{y'\})&=\sum_{y''\in Y} d(y)(\{(a,y'')\})\cdot \chi_{\mathcal{Y}}(y'',w)(\{y'\})
\end{align*}
where $a\in\myalphabet$ and $w\in\myalphabet^*$.

Then for each $y\in Y$ and $w\in\myalphabet^*$, we have:
\begin{align*}
[\trB(d_1),\trB(d_2)](y)(w\myalphabet^{\omega})
&\;=\;\sum_{y'\in Y} \chi_{\mathcal{Y}}(y,w)(y')\cdot [\trB(d_1),\trB(d_2)](y')(\myalphabet^{\omega}) \\
&\;=\;\sum_{y'\in Y} \chi_{\mathcal{Y}}(y,w)(y')\cdot [\trB(d'_1),\trB(d'_2)](y')(\myalphabet^{\omega}) \\
&\;=\;[\trB(d'_1),\trB(d'_2)](y)(w\myalphabet^{\omega})\,.
\end{align*}

\sloppy
By Carath\'eodory's extension theorem,
this implies 
$[\trB(d_1),\trB(d_2)](y)(A)=[\trB(d'_1),\trB(d'_2)](y)(A)$
for each measurable set $A$.
Hence 
we have $\trB(\mathcal{Y})=\trB(\mathcal{Y}')$.
\end{proof}

\noindent
With Lemma~\ref{lem:nonAcctoAccMain2} discharging its assumptions,
Proposition~\ref{prop:unjoinedbotSoundMain2} easily yields
Theorem~\ref{thm:soundnessOfSimulationGiryMainMatrix} as follows.

\begin{proof}[Proof (Theorem \ref{thm:soundnessOfSimulationGiryMainMatrix})]
Let $\mathcal{X}=((X_1,X_2),c,s)$ and $\mathcal{Y}=((Y_1,Y_2),d,t)$. 
By using the bijective correspondence between 
probabilistic matrices
and
arrows in $\Kl(\giry)(X,Y)$ where $X$ and $Y$ are equipped with discrete $\sigma$-algebras,
we can easily see that
a forward fair matrix simulation $A\in [0,1]^{Y\times X}$  
 from $\mathcal{X}$ to $\mathcal{Y}$ (Definition~\ref{def:fwdFairBuechiSimProbMatrix}) exists
if and only if
a  fair simulation $f:Y\kto X$ without dividing from $\mathcal{X}$ to $\mathcal{Y}$ 
(Definition~\ref{def:fwdFairBuechiSimWithoutDividing}) exists
(see also~\cite{urabeH14CONCUR,urabeH17MatrixSimulationJourn}).

We define $Y_{12}\subseteq Y_1$ by 
\[
Y_{12}=\{y\in Y_1\mid \exists y_0,\ldots y_n\in Y, y=y_0,y_n\in Y_2,\forall i.\; d(y_i)(\myalphabet\times \{y_{i+1}\}) > 0\}\,.
\]
As $Y_1$ is finite, $Y_{12}$ is also finite.
We define a B\"uchi $(\giry,\myalphabet\times(\place))$-system 
$\mathcal{Y}'=((Y'_1,Y'_2),d',t')$ 
by $Y'_1=Y_1\setminus Y_{12}$, $Y'_2=Y_2+Y_{12}$, $d'=d$ and $t'=t$.
%
As $Y_{12}$ is finite,
by repeatedly applying Lemma~\ref{lem:nonAcctoAccMain2}, 
we have $\trB(\mathcal{Y}')=\trB(\mathcal{Y})$.

It is easy to see that
$f$ is also a  fair simulation without dividing
from $\mathcal{X}$ to $\mathcal{Y}'$.
Moreover, by its definition, $\mathcal{Y}'$ satisfies $\trB(d'_1)=\bot$.
Therefore by Proposition~\ref{prop:unjoinedbotSoundMain2}, 
we have $\trB(\mathcal{X})\sqsubseteq\trB(\mathcal{Y}')$.
Hence  $\trB(\mathcal{X})\sqsubseteq\trB(\mathcal{Y})$ holds.
\end{proof}

It is still open 
whether the restriction from trees to words 
is necessary.
We  note that an analogous statement to Lemma~\ref{lem:nonAcctoAccMain2} does not hold for B\"uchi $(\giry,\myalphabet\times(\place)\times(\place))$-systems,
which can be regarded as probabilistic B\"uchi tree automata.
A counterexample is as follows.
\begin{exa}\label{exa:PBTADiv}
We define a B\"uchi $(\giry,\{a\}\times(\place)\times(\place))$-system
$\mathcal{Y}=((Y_1,Y_2),d,t)$ by: 
\begin{align*}
Y_1&=(\{y_1\},\pow\{y_1\})\,,\quad Y_2=(\{y_2\},\pow\{y_2\})\,, \\
d(y_1)(\{(a,y,y')\})&=\begin{cases}
\frac{1}{2} & (y=y'=y_1) \\
\frac{1}{4} & (y=y_1,y'=y_2\text{ or }y=y'=y_2) \\
0 &(\text{otherwise})\,,
\end{cases}\\
d(y_2)(\{(a,y,y')\})&=\begin{cases}
1 & (y=y'=y_2) \\
0 &(\text{otherwise})\,,
\end{cases}
\quad\text{and}\quad
t(*)(\{y\})=\begin{cases}
1 & (y=y_1) \\
0 &(\text{otherwise})\,.
\end{cases}
\end{align*}
Note that $d(y_1)(\{a\}\times Y_2\times Y_2)=\frac{1}{2}>0$.
Note also that the carrier set of the final $(\{a\}\times(\place)\times(\place))$-coalgebra is given by a singleton $\{*\}$.
It is not hard to see that $\trB(d_1)(y_1)(\{*\})=\frac{1}{2}$.
In contrast, if we define a B\"uchi $(\giry,\{a\}\times(\place)\times(\place))$-system
$\mathcal{Y}'=((Y'_1,Y'_2),d',t')$ by $Y'_1=\emptyset$, $Y'_2=\{y_1,y_2\}$, $d'=d$ and $t=t'$, 
then we have $\trB(d'_1)(y_1)(\{*\})=1$.
\end{exa}

In contrast, it turns out that the finiteness restriction  of $\mathcal{Y}$ in
Theorem~\ref{thm:soundnessOfSimulationGiryMainMatrix} is strict: 
an example below shows that
without it
soundness fails. 

\begin{figure}
\begin{center}
\begin{equation*}
\small
\xy <1.5mm,0mm>:
(0,0)*+[Foo]{x'_0} = "x20",
(0,15)*+[Foo]{x'_1} = "x21",
(0,30)*+[Foo]{x'_2} = "x22",
(15,0)*+[Fo]{x_0} = "x10", 
(15,15)*+[Fo]{x_1} = "x11", 
(15,30)*+[Fo]{x_2} = "x12", 
(45,0)*+[Fo]{y_0} = "y10",
(45,15)*+[Fo]{y_1} = "y11",
(45,30)*+[Fo]{y_2} = "y12",
(60,0)*+[Foo]{y'_0} = "y20",
(60,15)*+[Foo]{y'_1} = "y21",
(60,30)*+[Foo]{y'_2} = "y22",
(7.5,40)*{\vdots} = "",
%
(52.5,40)*{\vdots} = "",
(-5,40)*{\mathcal{X}}="",
(65,40)*{\mathcal{Y}}="",
\ar |{\frac{1}{2}p_2} (15,-15);"x10"*+++[o]{}
\ar |{1-\frac{1}{2}p_2} (15,-15);"x20"
\ar |{p_3} "x10"*+++[o]{};"x11"*+++[o]{}
\ar |{p_4} "x11"*+++[o]{};"x12"*+++[o]{}
\ar ^{} "x12"*+++[o]{};(15,35)
\ar |{1-p_3} "x10"*+++[o]{};"x21"*+++[o]{}
\ar |{1-p_4} "x11"*+++[o]{};"x22"*+++[o]{}
\ar ^{} "x12"*+++[o]{};(10,35)
\ar ^{1} "x20"*+++[o]{};"x21"*+++[o]{}
\ar ^{1} "x21"*+++[o]{};"x22"*+++[o]{}
\ar ^{} "x22"*+++[o]{};(0,35)
\ar _{1} (45,-15);"y10"*+++[o]{}
\ar |{p_2} "y10"*+++[o]{};"y11"*+++[o]{}
\ar |{p_3} "y11"*+++[o]{};"y12"*+++[o]{}
\ar ^{} "y12"*+++[o]{};(45,35)
\ar |{1-p_2} "y10"*+++[o]{};"y21"*+++[o]{}
\ar |{1-p_3} "y11"*+++[o]{};"y22"*+++[o]{}
\ar ^{} "y12"*+++[o]{};(50,35)
\ar _{1} "y20"*+++[o]{};"y21"*+++[o]{}
\ar _{1} "y21"*+++[o]{};"y22"*+++[o]{}
\ar ^{} "y22"*+++[o]{};(60,35)
\ar @{-->}|{\;\frac{1}{2}p_2\;} "y10";"x10"
\ar @{-->}|{\;\frac{1}{2}p_3\;} "y11";"x11"
\ar @{-->}|{\;\frac{1}{2}p_4\;} "y12";"x12"
\ar @/^5mm/@{-->}^(.3){1-\frac{1}{2}p_2} "y10";"x20"
\ar @/^5mm/@{-->}^(.3){1-\frac{1}{2}p_3} "y11";"x21"
\ar @/^5mm/@{-->}^(.3){1-\frac{1}{2}p_4} "y12";"x22"
\ar @/_5mm/@{-->}|{1} "y20";"x20"
\ar @/_5mm/@{-->}|{1} "y21";"x21"
\ar @/_5mm/@{-->}|{1} "y22";"x22"
\endxy
\end{equation*}
\caption{\small B\"uchi $(\giry,\{a\}\times(\underline{\protect\phantom{n}}\,))$-systems 
$\mathcal{X}$ and $\mathcal{Y}$ (whose transitions are denoted by solid lines), 
and a forward fair simulation (denoted by dashed lines). Labels are omitted. 
}
\label{fig:1702071655}
\end{center}
\end{figure}

\begin{exa}\label{example:joinedSimNotSound}
Let $\mathcal{X}=((X_1,X_2),c,s)$ and $\mathcal{Y}=((Y_1,Y_2),d,t)$ be B\"uchi  
$\left(\giry,\{a\}\times(\place)\right)$-systems
that 
are illustrated as PBWAs in  Figure~\ref{fig:1702071655}.
For each $i\in\omega$, we define $p_i\in[0,1]$ by $p_i=1-{1}/{i^2}$.
%
%
%
%
%
%
We define a family $f=(f_{ij}:Y_i\kto X_j)_{i,j\in\{1,2\}}$ of Kleisli arrows as follows
(see also dashed lines in Figure~\ref{fig:1702071655}):
\begin{itemize}
\item $f_{11}(y_i)(\{x_j\})=\frac{1}{2}p_{i+2}$ if $j=i$ and $0$ otherwise;

\item $f_{12}(y_i)(\{x'_j\})=1-\frac{1}{2}p_{i+2}$ if $j=i$ and $0$ otherwise;

\item $f_{21}(y'_i)(\{x_j\})=0$; and

\item $f_{22}(y'_i)(\{x'_j\})=1$ if $j=i$ and $0$ otherwise.
\end{itemize}

Moreover, for an ordinal $\alpha\in\{0,1\}$, we define Kleisli arrows $f_{11}(\alpha):Y_1\kto Y_1$ and 
$f_{12}(\alpha):Y_1\kto X_2$ by 
$f_{11}(0)=\bot$, $f_{11}(1)=f_{11}$, $f_{12}(0)=\bot$ and $f_{12}(1)=f_{12}$.

Then 
the following inequalities hold.
\begin{align*}
s&\sqsubseteq 
\left[\codomJoin{f_{11},f_{12}},\codomJoin{f_{21},f_{22}}\right]\odot t \\
%
c\odot\codomJoin{f_{11}(0),f_{12}(1)} &\sqsubseteq 
\overline{F}\left[\codomJoin{f_{11}(0),f_{12}(0)},\codomJoin{f_{21},f_{22}}\right]\odot d_1 \\
%
c\odot\codomJoin{f_{11}(1),f_{12}(1)} &\sqsubseteq 
\overline{F}\left[\codomJoin{f_{11}(1),f_{12}(1)},\codomJoin{f_{21},f_{22}}\right]\odot d_1 \\
c\odot\codomJoin{f_{21},f_{22}} &\sqsubseteq 
\overline{F}\left[\codomJoin{f_{11},f_{12}},\codomJoin{f_{21},f_{22}}\right]\odot d_2 
\end{align*}
Therefore $f$ is a forward fair simulation without dividing (Definition~\ref{def:fwdFairBuechiSimWithoutDividing}) from $\mathcal{X}$ to $\mathcal{Y}$.

In contrast, we also have:
\allowdisplaybreaks[4]
\begin{align*}
\tr(\mathcal{X})(\{a^{\omega}\})
&=\frac{1}{2}p_2\cdot (1-\prod_{i\in\omega}p_{i+3})+(1-\frac{1}{2}p_2)\cdot 1 = 1-\frac{1}{2}\prod_{i\in\omega}p_{i+2}\qquad\text{and}\\
\tr(\mathcal{Y})(\{a^{\omega}\})
&=1-\prod_{i\in\omega}p_{i+2}\,.
\end{align*}
As $\prod_{i\in\omega}p_{i+2}=\frac{1}{2}>0$, we have $\tr(\mathcal{X})\mathrel{\cancel{\sqsubseteq}}\tr(\mathcal{Y})$.
Therefore language inclusion fails.
\end{exa}

\section{Conclusions and Future Work}\label{sec:concluFW}
We defined notions of fair simulation for two types of transition systems with B\"uchi acceptance conditions, namely:
nondeterministic B\"uchi tree automata (NBTAs) and probabilistic B\"uchi word automata (PBWAs, with an additional finiteness assumption on the simulating side).

The simulation notion for NBTAs is defined in terms of fixed-point equations (Definition~\ref{def:simForNondetBuechiRel}).
The resulting notion is almost the same as the one in~\cite{Bomhard08BScThesis}, except that infinite state spaces 
are allowed in our definition because of the fixed-point formulation. 
In contrast, the fair simulation notion for PBWAs is new to the best of our knowledge.
It is a combination of a notion of matrix simulation~\cite{urabeH17MatrixSimulationJourn} and 
a notion of (lattice-theoretic) progress measure~\cite{Jurdzinski00,HasuoSC16}.

These simulation notions originate from categorical backgrounds that are built upon a categorical characterization of parity languages
developed in~\cite{urabeSH16parityTrace}. 
Using the theory in~\cite{urabeSH16parityTrace}, we have introduced a categorical simulation notion called 
(forward) fair simulation with dividing (Definition~\ref{def:fwdFairBuechiSimWithDiv})
and proved its soundness.

Our soundness proof for the categorical notion of fair simulation with dividing is mathematically clean, 
and we can easily obtain sound simulation notions for NBTAs and PBWAs
 by specializing the categorical notion.
However, the resulting notions inherit dividing requirement, and
it is a big disadvantage.
For NBTAs and PBWAs, we have shown that by using properties that are specific to these systems, and by imposing a finite-state restriction to the simulating side for PBWAs, 
the dividing requirement can be lifted.

\subsection{Future Work} 
Generalization from the B\"{u}chi
condition to the \emph{parity} one is certainly what we aim at next.
It is already not very clear how our coalgebraic definition with dividing
(Section~\ref{subsec:coalgFairSimWithDividing}) would generalize: 
for example, in case of parity automata, there is little sense in
comparing the priority of the challenger's state with that of the simulator's.
It is even
less clear how to circumvent dividing.

Aside from fair simulation, notions of \emph{delayed simulation} 
are known for B\"uchi
automata~\cite{etessamiWS05fairsimulation,fritzW06simulationrelations}:
they are subject to slightly different ``fairness'' constraints. Accommodating them in the current setting is another future work.

On decidability and complexity, 
while the obtained simulation notion for NBTAs is decidable if the state spaces are finite,
the decidability of that for PBWAs is still open even for finite-state systems. 
Obviously it is one of possible directions of future work.

We are also interested in automatic discovery of simulations---via mathematical programming for example---and its implementation. 
In this direction of future work we will
be based on our previous
work~\cite{urabeH14CONCUR,urabeH17MatrixSimulationJourn}. Another direction is
to use the current results for \emph{program verification}---where the
$\mathsf{Integer}$ type makes problems inherently infinitary---exploiting our
non-combinatorial presentation by equational systems that allows infinite state spaces. 
We could do so 
automatically by synthesizing a symbolic simulation or, interactively 
on a proof assistant.

\subsection*{Acknowledgments}
Thanks are due 
to Shunsuke Shimizu, Kenta Cho, Eugenia Sironi, and the anonymous referees,
for discussions and comments.
The authors are supported by
ERATO HASUO Metamathematics for Systems Design Project (No.\ JPMJER1603), JST, and
Grant-in-Aid No.\ 15KT0012, JSPS.
Natsuki Urabe is supported by Grant-in-Aid No.\ 16J08157 for JSPS Fellows.

\bibliographystyle{alpha} 
\bibliography{myref} %

\newcommand{\etalchar}[1]{$^{#1}$}
\begin{thebibliography}{ABH{\etalchar{+}}12}

\bibitem[ABH{\etalchar{+}}12]{AdamekBHKMS12}
Jir\'{\i} Ad{\'a}mek, Filippo Bonchi, Mathias H{\"u}lsbusch, Barbara K{\"o}nig,
  Stefan Milius, and Alexandra Silva.
\newblock A coalgebraic perspective on minimization and determinization.
\newblock In Lars Birkedal, editor, {\em FoSSaCS}, volume 7213 of {\em Lect.
  Notes Comp. Sci.}, pages 58--73. Springer, 2012.

\bibitem[Ad{\'{a}}74]{Adamek74}
Ji\v{r}\'{i} Ad{\'{a}}mek.
\newblock Free algebras and automata realizations in the language of
  categories.
\newblock {\em Comment. Math. Univ. Carolin.}, 15:589^^e2^^80^^93602, 1974.

\bibitem[ADD00]{ashD2000probability}
Robert~B. Ash and Catherine Doleans-Dade.
\newblock {\em Probability and measure theory Second Edition}.
\newblock Academic Press, 2000.

\bibitem[AK79]{adamekK79leastfixed}
Jir{\'{\i}} Ad{\'{a}}mek and V{\'{a}}clav Koubek.
\newblock Least fixed point of a functor.
\newblock {\em J. Comput. Syst. Sci.}, 19(2):163--178, 1979.

\bibitem[AMV11]{AdamekMV11}
Jir{\'{\i}} Ad{\'{a}}mek, Stefan Milius, and Jiri Velebil.
\newblock Elgot theories: a new perspective on the equational properties of
  iteration.
\newblock {\em Mathematical Structures in Computer Science}, 21(2):417--480,
  2011.

\bibitem[AN01]{ArnoldNiwinski}
Andr\'{e} Arnold and Damian Niwi\'{n}ski.
\newblock {\em Rudiments of $\mu$-Calculus}.
\newblock Studies in Logic and the Foundations of Mathematics. Elsevier,
  Amsterdam, 2001.

\bibitem[BC03]{blondel03undecidableproblems}
Vincent~D. Blondel and Vincent Canterini.
\newblock Undecidable problems for probabilistic automata of fixed dimension.
\newblock {\em Theory Comput. Syst.}, 36(3):231--245, 2003.

\bibitem[BK08]{baierK08principlesofmodelchecking}
Christel Baier and Joost{-}Pieter Katoen.
\newblock {\em Principles of model checking}.
\newblock {MIT} Press, 2008.

\bibitem[BMSZ14]{bonchiMSZ14killEpsilons}
Filippo Bonchi, Stefan Milius, Alexandra Silva, and Fabio Zanasi.
\newblock How to kill epsilons with a dagger - {A} coalgebraic take on systems
  with algebraic label structure.
\newblock In Marcello~M. Bonsangue, editor, {\em Coalgebraic Methods in
  Computer Science - 12th {IFIP} {WG} 1.3 International Workshop, {CMCS} 2014,
  Colocated with {ETAPS} 2014, Grenoble, France, April 5-6, 2014, Revised
  Selected Papers}, volume 8446 of {\em Lecture Notes in Computer Science},
  pages 53--74. Springer, 2014.

\bibitem[Bor94]{borceux1994Handbook2}
Francis Borceux.
\newblock {\em Handbook of Categorical Algebra:}, volume~2.
\newblock Cambridge University Press, Cambridge, 11 1994.

\bibitem[CC79]{cousotC79}
Patrick Cousot and Radhia Cousot.
\newblock Constructive versions of {Tarski}'s fixed point theorems.
\newblock {\em Pacific Journal of Mathematics}, 82(1):43--57, 1979.

\bibitem[CHS14]{carayolHS14}
Arnaud Carayol, Axel Haddad, and Olivier Serre.
\newblock Randomization in automata on infinite trees.
\newblock {\em {ACM} Trans. Comput. Log.}, 15(3):24:1--24:33, 2014.

\bibitem[C{\^{\i}}r10]{cirstea10genericinfinite}
Corina C{\^{\i}}rstea.
\newblock Generic infinite traces and path-based coalgebraic temporal logics.
\newblock {\em Electr. Notes Theor. Comput. Sci.}, 264(2):83--103, 2010.

\bibitem[C{\^{\i}}r13]{cirstea13fromBranching}
Corina C{\^{\i}}rstea.
\newblock From branching to linear time, coalgebraically.
\newblock In David Baelde and Arnaud Carayol, editors, {\em Proceedings
  Workshop on Fixed Points in Computer Science, {FICS} 2013, Turino, Italy,
  September 1st, 2013.}, volume 126 of {\em {EPTCS}}, pages 11--27, 2013.

\bibitem[CKS92]{CleavelandKS92}
Rance Cleaveland, Marion Klein, and Bernhard Steffen.
\newblock Faster model checking for the modal mu-calculus.
\newblock In Gregor von Bochmann and David~K. Probst, editors, {\em Computer
  Aided Verification, Fourth International Workshop, {CAV} '92, Montreal,
  Canada, June 29 - July 1, 1992, Proceedings}, volume 663 of {\em Lecture
  Notes in Computer Science}, pages 410--422. Springer, 1992.

\bibitem[Doo94]{doob94measuretheory}
J.L. Doob.
\newblock {\em Measure Theory}.
\newblock Graduate Texts in Mathematics. Springer New York, 1994.

\bibitem[EWS05]{etessamiWS05fairsimulation}
Kousha Etessami, Thomas Wilke, and Rebecca~A. Schuller.
\newblock Fair simulation relations, parity games, and state space reduction
  for {B\"{u}chi} automata.
\newblock {\em {SIAM} J. Comput.}, 34(5):1159--1175, 2005.

\bibitem[FW06]{fritzW06simulationrelations}
Carsten Fritz and Thomas Wilke.
\newblock Simulation relations for alternating parity automata and parity
  games.
\newblock In Oscar~H. Ibarra and Zhe Dang, editors, {\em Developments in
  Language Theory, 10th International Conference, {DLT} 2006, Santa Barbara,
  CA, USA, June 26-29, 2006, Proceedings}, volume 4036 of {\em Lecture Notes in
  Computer Science}, pages 59--70. Springer, 2006.

\bibitem[Gir82]{giry82categoricalapproach}
Michele Giry.
\newblock {A categorical approach to probability theory.}
\newblock In {\em Proc. Categorical Aspects of Topology and Analysis}, volume
  915 of {\em Lect. Notes Math.}, pages 68--85, 1982.

\bibitem[Has06]{hasuo06genericforward}
Ichiro Hasuo.
\newblock Generic forward and backward simulations.
\newblock In Christel Baier and Holger Hermanns, editors, {\em CONCUR}, volume
  4137 of {\em Lecture Notes in Computer Science}, pages 406--420. Springer,
  2006.

\bibitem[Her06]{herrlich06axiomofchoice}
Horst Herrlich.
\newblock {\em Axiom of Choice}.
\newblock Lecture Notes in Mathematics. Springer Berlin Heidelberg, 2006.

\bibitem[HJS07]{hasuo07generictrace}
Ichiro Hasuo, Bart Jacobs, and Ana Sokolova.
\newblock Generic trace semantics via coinduction.
\newblock {\em Logical Methods in Computer Science}, 3(4), 2007.

\bibitem[HKR02]{henzingerKR02fairsimulation}
Thomas~A. Henzinger, Orna Kupferman, and Sriram~K. Rajamani.
\newblock Fair simulation.
\newblock {\em Inf. Comput.}, 173(1):64--81, 2002.

\bibitem[HSC15]{hasuoSC15arXiv}
Ichiro Hasuo, Shunsuke Shimizu, and Corina C{\^{\i}}rstea.
\newblock Lattice-theoretic progress measures and coalgebraic model checking
  (with appendices).
\newblock {\em CoRR}, abs/1511.00346, 2015.

\bibitem[HSC16]{HasuoSC16}
Ichiro Hasuo, Shunsuke Shimizu, and Corina C{\^{\i}}rstea.
\newblock Lattice-theoretic progress measures and coalgebraic model checking.
\newblock In Rastislav Bodik and Rupak Majumdar, editors, {\em Proceedings of
  the 43rd Annual {ACM} {SIGPLAN-SIGACT} Symposium on Principles of Programming
  Languages, {POPL} 2016, St. Petersburg, FL, USA, January 20 - 22, 2016},
  pages 718--732. {ACM}, 2016.

\bibitem[Jac04]{jacobs04tracesemantics}
Bart Jacobs.
\newblock Trace semantics for coalgebras.
\newblock {\em Electr. Notes Theor. Comput. Sci.}, 106:167--184, 2004.

\bibitem[Jac10]{jacobs10trace}
Bart Jacobs.
\newblock From coalgebraic to monoidal traces.
\newblock In {\em Coalgebraic Methods in Computer Science (CMCS 2010)}, volume
  264 of {\em Elect. Notes in Theor. Comp. Sci.}, pages 125--140. Elsevier,
  Amsterdam, 2010.

\bibitem[Jac16]{jacobs16CoalgBook}
Bart Jacobs.
\newblock {\em Introduction to Coalgebra: Towards Mathematics of States and
  Observation}.
\newblock Cambridge Tracts in Theoretical Computer Science. Cambridge
  University Press, 2016.

\bibitem[JL91]{JonssonL91}
Bengt Jonsson and Kim~Guldstrand Larsen.
\newblock Specification and refinement of probabilistic processes.
\newblock In {\em LICS}, pages 266--277. IEEE Computer Society, 1991.

\bibitem[JP06]{JuvekarP06}
Sudeep Juvekar and Nir Piterman.
\newblock Minimizing generalized {B{\"u}chi} automata.
\newblock In Thomas Ball and Robert~B. Jones, editors, {\em Computer Aided
  Verification: 18th International Conference, CAV 2006, Seattle, WA, USA,
  August 17-20, 2006. Proceedings}, pages 45--58, Berlin, Heidelberg, 2006.
  Springer Berlin Heidelberg.

\bibitem[JSS15]{Jacobs0S15}
Bart Jacobs, Alexandra Silva, and Ana Sokolova.
\newblock Trace semantics via determinization.
\newblock {\em J. Comput. Syst. Sci.}, 81(5):859--879, 2015.

\bibitem[Jur00]{Jurdzinski00}
Marcin Jurdzinski.
\newblock Small progress measures for solving parity games.
\newblock In Horst Reichel and Sophie Tison, editors, {\em STACS}, volume 1770
  of {\em Lecture Notes in Computer Science}, pages 290--301. Springer, 2000.

\bibitem[KK13]{kerstan13coalgebraictrace}
Henning Kerstan and Barbara K{\"{o}}nig.
\newblock Coalgebraic trace semantics for continuous probabilistic transition
  systems.
\newblock {\em Logical Methods in Computer Science}, 9(4), 2013.

\bibitem[LV95]{lynch95forwardand}
Nancy~A. Lynch and Frits~W. Vaandrager.
\newblock Forward and backward simulations. {I}.~{U}ntimed systems.
\newblock {\em Inf. Comput.}, 121(2):214--233, 1995.

\bibitem[LV96]{lynch96forwardand}
Nancy~A. Lynch and Frits~W. Vaandrager.
\newblock Forward and backward simulations. {II}.~{T}iming based systems.
\newblock {\em Inf. Comput.}, 128(1):1--25, 1996.

\bibitem[{Mac}98]{MacLane71}
S.~{Mac~Lane}.
\newblock {\em Categories for the Working Mathematician}.
\newblock Springer, Berlin, 2nd edition, 1998.

\bibitem[PT97]{PowerT97}
J.~Power and H.~Thielecke.
\newblock Environments, continuation semantics and indexed categories.
\newblock In M.~Abadi and T.~Ito, editors, {\em Theoretical Aspects of Computer
  Software}, number 1281 in Lect. Notes Comp. Sci., pages 391--414. Springer,
  Berlin, 1997.

\bibitem[Rut00]{Rutten00a}
J.~J. M.~M. Rutten.
\newblock Universal coalgebra: a theory of systems.
\newblock {\em Theor. Comp. Sci.}, 249:3--80, 2000.

\bibitem[TW{\etalchar{+}}02]{thomas2002automata}
Wolfgang Thomas, Thomas Wilke, et~al.
\newblock {\em Automata, logics, and infinite games: a guide to current
  research}, volume 2500.
\newblock Springer Science \& Business Media, 2002.

\bibitem[UH14]{urabeH14CONCUR}
Natsuki Urabe and Ichiro Hasuo.
\newblock Generic forward and backward simulations {III:} quantitative
  simulations by matrices.
\newblock In Paolo Baldan and Daniele Gorla, editors, {\em {CONCUR} 2014 -
  Concurrency Theory - 25th International Conference, {CONCUR} 2014, Rome,
  Italy, September 2-5, 2014. Proceedings}, volume 8704 of {\em Lecture Notes
  in Computer Science}, pages 451--466. Springer, 2014.

\bibitem[UH15]{urabeH15CALCO}
Natsuki Urabe and Ichiro Hasuo.
\newblock Coalgebraic infinite traces and {Kleisli} simulations.
\newblock In {\em Algebra and Coalgebra in Computer Science - 6th International
  Conference, {CALCO} 2015, Nijmegen, Netherlands, June 24-26, 2015.
  Proceedings}, 2015.

\bibitem[UH17]{urabeH17MatrixSimulationJourn}
Natsuki Urabe and Ichiro Hasuo.
\newblock Quantitative simulations by matrices.
\newblock {\em Inf. Comput.}, 252:110--137, 2017.

\bibitem[UHH17]{UrabeHH17toAppear}
Natsuki Urabe, Masaki Hara, and Ichiro Hasuo.
\newblock Categorical liveness checking by corecursive algebras.
\newblock In {\em Proc.\ LICS 2017}, 2017.
\newblock To appear.

\bibitem[USH16]{urabeSH16parityTrace}
Natsuki Urabe, Shunsuke Shimizu, and Ichiro Hasuo.
\newblock Coalgebraic trace semantics for buechi and parity automata.
\newblock In Jos{\'{e}}e Desharnais and Radha Jagadeesan, editors, {\em 27th
  International Conference on Concurrency Theory, {CONCUR} 2016, August 23-26,
  2016, Qu{\'{e}}bec City, Canada}, volume~59 of {\em LIPIcs}, pages
  24:1--24:15. Schloss Dagstuhl - Leibniz-Zentrum fuer Informatik, 2016.

\bibitem[Var95]{Vardi95}
Moshe~Y. Vardi.
\newblock An automata-theoretic approach to linear temporal logic.
\newblock In Faron Moller and Graham~M. Birtwistle, editors, {\em Banff Higher
  Order Workshop}, volume 1043 of {\em Lecture Notes in Computer Science},
  pages 238--266. Springer, 1995.

\bibitem[vB08]{Bomhard08BScThesis}
Thomas von Bomhard.
\newblock Minimization of tree automata.
\newblock {BSc} thesis, Universit\"{a}t des Saarlandes, September 2008.

\bibitem[vBW05]{breugelW05BehaviouralPseudometric}
Franck van Breugel and James Worrell.
\newblock A behavioural pseudometric for probabilistic transition systems.
\newblock {\em Theor. Comput. Sci.}, 331(1):115--142, 2005.

\bibitem[vG01]{vanGlabbeek01}
R.~J. van Glabbeek.
\newblock The linear time--branching time spectrum {I}; the semantics of
  concrete, sequential processes.
\newblock In J.~A. Bergstra, A.~Ponse, and S.~A. Smolka, editors, {\em Handbook
  of Process Algebra}, chapter~1, pages 3--99. Elsevier, 2001.

\end{thebibliography}

\end{document}